\documentclass[11pt]{article}
\usepackage{amsmath,amsfonts,amsthm,amssymb,ifthen,graphicx}
\usepackage{bbm}
\usepackage[square,sort,comma,numbers]{natbib}
\graphicspath{{./}{./figures/}}
\newboolean{@laptop}
\newboolean{@todoon}
\setboolean{@laptop}{true}
\setboolean{@todoon}{true}

\ifthenelse{\boolean{@laptop}}{}{\usepackage{mathbbol}}

\def\todo#1{\ifthenelse{\boolean{@todoon}}{\marginpar{\textit{#1}}}{}}

\usepackage{times}
\usepackage{tabu}
\usepackage{thmtools}
\usepackage{thm-restate}

\usepackage{helvet}
\usepackage{courier}
\usepackage{natbib}
\usepackage{colortbl,xcolor}
\usepackage{amsmath}
\usepackage{amssymb}

\usepackage{mathabx}
\usepackage{algorithm}
\usepackage{algorithmicx}
\usepackage[noend]{algpseudocode}
\newtheorem{lemma}{Lemma}

\newtheorem{theorem}[lemma]{Theorem}
\newtheorem{definition}{Definition}
\newtheorem{remark}{Remark}

\newtheorem{corollary}[lemma]{Corollary}
\newtheorem{example}{Example}

\newtheorem{conjecture}{Conjecture}

\DeclareMathOperator*{\argmax}{\arg\,\max}

\newcommand{\Neighbor}{\mathcal{N}}

\newcommand{\Interval}{\mathbb{I}}
\newcommand{\DependencyGraph}{G}

\newcommand{\BipartiteGraph}{H}
\newcommand{\EventSet}{\mathcal{A}}
\newcommand{\EventSetB}{\mathcal{B}}
\newcommand{\Event}{A}
\newcommand{\EventB}{B}
\newcommand{\VariableSet}{\mathcal{X}}
\newcommand{\VariableSetY}{\mathcal{Y}}

\newcommand{\Variable}{X}
\newcommand{\VariableY}{Y}

\newcommand{\Neg}[1]{\overline{#1}}
\renewcommand{\vec}[1]{\mathbf{#1}}
\renewcommand{\Pr}{\mathbb{P}}
\newcommand{\DiscreteDimen}{d}
\newcommand{\DiscreteVec}{\mathbf{d}}
\newcommand{\Interior}{\mathcal{I}}
\newcommand{\Mup}{MUP}
\newcommand{\Int}{INT}
\newcommand{\EventVariable}{event-variable graph}

\newcommand{\InducedDependencyGraph}{base graph}

\usepackage{ifthen}

\ifthenelse{\isundefined{\GenerateShortVersion}}
{
\def\LongVersion{}
\def\LongVersionEnd{}
\long\def\ShortVersion#1\ShortVersionEnd{}
}
{
\def\ShortVersion{}
\def\ShortVersionEnd{}
\long\def\LongVersion#1\LongVersionEnd{}
}

\def\GenerateSingleColumn{}
\ifthenelse{\isundefined{\GenerateSingleColumn}}
{
\def\DoubleColumn{}
\def\DoubleColumnEnd{}
\long\def\SingleColumn#1\SingleColumnEnd{}
}
{
\def\SingleColumn{}
\def\SingleColumnEnd{}
\long\def\DoubleColumn#1\DoubleColumnEnd{}
}

\begin{document}

\title{Variable Version Lov\'{a}sz Local Lemma: Beyond Shearer's Bound\thanks{Part of the work has been published at FOCS2017}}

\author{
Kun He\thanks{
Institute of Computing Technology,
Chinese Academy of Sciences. University of Chinese Academy of Sciences. Beijing, China.
Email:\texttt{hekun@ict.ac.cn}
}
\and
Liang Li\thanks{
Department of Artificial Intelligence, 
Ant Financial Services Group, China. 
Email:\texttt{liangli.ll@antfin.com}.
}
\and
Xingwu Liu\thanks{Correspondence author} \thanks{
Institute of Computing Technology,
Chinese Academy of Sciences. University of Chinese Academy of Sciences. Beijing, China.
Email:\texttt{liuxingwu@ict.ac.cn}. 
}
\and
Yuyi Wang\thanks{
Disco Group, ETH Z\"{u}rich, Switzerland. 
Email:\texttt{yuwang@ethz.ch}.
}
\and
Mingji Xia\thanks{
State Key Laboratory of Computer Science, Institute of Software,
Chinese Academy of Sciences. University of Chinese Academy of Sciences. Beijing, China.
Email:\texttt{mingji@ios.ac.cn}
}
}

\date{\today}

\maketitle
\maketitle
\begin{abstract}
A tight criterion under which the abstract version Lov\'{a}sz Local Lemma (abstract-LLL) holds was given by Shearer \cite{shearer1985problem} decades ago. 
However, little is known about that of the variable version LLL (variable-LLL) where events are generated by independent random variables, though this model of events is applicable to
%variable-LLL naturally models and is enough for 
almost all applications of LLL. We introduce a necessary and sufficient criterion for variable-LLL, in terms of the probabilities of the events and the event-variable graph specifying the dependency among the events. Based on this new criterion, we obtain boundaries for two families of event-variable graphs, namely, cyclic and treelike bigraphs. 
These are the first two non-trivial cases where the variable-LLL boundary is fully determined. As a byproduct, we also provide a universal constructive method to find a set of events whose union has the maximum probability, given the probability vector and the event-variable graph. 

Though it is \#P-hard in general to determine variable-LLL boundaries, we can to some extent decide whether a gap exists between a variable-LLL boundary and the corresponding abstract-LLL boundary. 
In particular, we show that the gap existence can be decided without solving Shearer's conditions or checking our variable-LLL criterion. 
Equipped with this powerful theorem, we show that there is no gap if the \InducedDependencyGraph ~of the event-variable graph is a tree, while gap appears if the \InducedDependencyGraph ~has an induced cycle of length at least $4$. The problem is almost completely solved except when the \InducedDependencyGraph ~ has only $3$-cliques, in which case we also get partial solutions. 

A set of reduction rules are established that facilitate to infer gap existence of an event-variable graph from known ones. As an application, various event-variable graphs, in particular combinatorial ones, are shown to be gapful/gapless. 
\end{abstract}

\LongVersion
\newpage
\LongVersionEnd

\section{Introduction}\footnote{Accepted by FOCS 2017}
Lov\'{a}sz Local Lemma, or LLL for short, is one of the most important probabilistic methods that has numerous applications since proposed in 1975 by Erd{\H{o}}s and Lov{\'a}sz \cite{erdos1975problems}. 
%Basically, given a set $\EventSet$ of \emph{bad} events $\Event_1,...,\Event_n$ in a probability space and the vector $\vec{p}=(\Pr(\Event_1),...,\Pr(\Event_n))$ of their probabilities, LLL aims at finding conditions on $\EventSet$ and $\vec{p}$ such that all the bad events can be avoided, namely $\Pr(\Neg{\Event_1} \wedge \ldots \wedge \Neg{\Event_n}) >0$. 
Basically, LLL aims at finding conditions under which any given set $\EventSet$ of bad events in a probability space can be avoided simultaneously, namely $\Pr(\cap_{\Event\in\EventSet}\Neg{\Event}) >0$. 
In the most general setting, the dependency among $\EventSet$ is characterized by an undirected graph $\DependencyGraph=([n],E)$, called a dependency graph of $\EventSet$, which satisfies that for any vertex $i$, $\Event_i$ is independent of $\{\Event_j: j\neq i, j\notin \Neighbor(i)\}$, where $\Neighbor(i)$ stands for the neighborhood of $i$ in $\DependencyGraph$. In this context, finding the conditions on $\EventSet$ is reduced to the fundamental challenge: Given a graph $\DependencyGraph$, determine its abstract interior $\mathcal{I}_a(\DependencyGraph)$ which is the set of vectors $\vec{p}$ such that $\Pr\left(\cap_{\Event\in \EventSet} \Neg{\Event} \right)>0$ for any event set $\EventSet$ with dependency graph $\DependencyGraph$ and probability vector $\vec{p}$. \textit{Local} solutions to this problem are collectively called abstract-LLL.
The most frequently used abstract-LLL is as follows:
\begin{theorem}[\cite{spencer1977asymptotic}]\label{thm:asymmetric}
Given a graph $\DependencyGraph=([n],E)$ and a vector $\vec{p}\in (0,1)^n$, if there exist real numbers $x_1,...,x_n \in (0,1)$ such that $p_i \leq x_i \prod_{j\in  \Neighbor(i)} (1-x_j)$ for any $i\in[n]$, then $\vec{p}\in \mathcal{I}_a(\DependencyGraph)$.
%$\Pr(A_i) \leq x_i \prod_{j\in  \Neighbor(i)} (1-x_j)$ that for any $i\in[n]$, 
%Let $\EventSet =\{\Event_1,\ldots,\Event_n\}$ be a finite set of events in the probability space $\Omega$. 
%Given a dependency graph $G$ of $\EventSet$, % let $\Neighbor_G(\Event)$ denote the neighborhood of $\Event$ in $G$ ($\Event_i$ and $\Event_j$ are neighbors iff $i,j$ are adjacent). If
%if there exist real numbers $x_1,...,x_n \in (0,1)$ such 
%$\Pr(A_i) \leq x_i \prod_{j\in  \Neighbor(i)} (1-x_j)$ that for any $i\in[n]$, 
%%$$\Pr(A_i) \leq x(\Event_i) \prod_{B \in \Neighbor(\Event_i)} (1-x(B))$$
%then it is possible to avoid all the events in $\EventSet$, i.e.
%$\Pr\left(\Neg{\Event_1} \wedge \ldots \wedge \Neg{\Event_n} \right) \geq \prod_{i\in [n]} (1-x(\Event_i))>0$. 
\end{theorem}
%This result also improved the condition for symmetric LLL to $ep(d+1)\le 1$. When $d>2$, this condition is tighter than $4pd\le 1$. 
%Given a graph $\DependencyGraph=([n],E)$, define its abstract interior $\mathcal{I}_a(\DependencyGraph)$ to be the set of vectors $\vec{p}\in (0,1]^n$ such that $\Pr\left(\cap_{\Event\in \EventSet} \Neg{\Event} \right)>0$ for any event set $\EventSet=\{\Event_1,...,\Event_n\}$ with dependency graph $\DependencyGraph$ and probability vector $\vec{p}$.
An exact characterization of $\mathcal{I}_a(\DependencyGraph)$ was presented by Shearer \cite{shearer1985problem} over 30 years ago. 
%It is based on \emph{alternating-sign independence polynomials} of the dependency graph. 
%\begin{definition}[alternating-sign independence polynomial]
%For events $\Event_1, \ldots , \Event_n$ with probabilities $p_1, \ldots , p_n$ and a
%dependency graph $\DependencyGraph$, let $Indep(\DependencyGraph)$ be the collection of independent sets of $\DependencyGraph$. the \emph{alternating-sign independence polynomial} is defined as
%$$q_S(p_1,\ldots , p_n) = \sum_{T\supseteq S, T\textrm{ is an independent set of }\DependencyGraph} (-1)^{|T|-|S|} \prod_{i\in T} p_i.$$
%\end{definition}
\begin{theorem}[\cite{shearer1985problem}]\label{thm:shearer}
Given a graph $\DependencyGraph=([n],E)$ and a vector $\vec{p}\in (0,1)^n$, $\vec{p}\in \mathcal{I}_a(\DependencyGraph)$ if and only if for any $S\in Ind(\DependencyGraph)$, $\sum_{T\supseteq S, T\in Ind(\DependencyGraph)} (-1)^{|T|-|S|} \prod_{i\in T} p_i>0$, where $Ind(\DependencyGraph)$ is the collection of independent sets of $\DependencyGraph$.
%If $q_S(p_1,\ldots, p_n) > 0$ for all independent sets $S \subseteq \EventSet$, then
%$\Pr(\cap_{i=1}^n \Neg{\Event_i}) 
%\ge q_{\emptyset}(p_1, \ldots , p_n)>0$. 
\end{theorem}
%Again, by the above theorem, the symmetric version of LLL is improved to $p d^d < (d-1)^{(d-1)}$ for $d>1$ and $d<1/2$ for $d=1$. 
%variable version

As in Theorem \ref{thm:asymmetric} and Theorem \ref{thm:shearer}, only dependency graphs and probabilities of events are involved in abstract-LLL. 
However, dependency graphs can only capture which events are dependent (more precisely, which events are independent), %with nothing to do with 
but not how they are dependent. 

A nice model of richer dependency structures is the variable-generated system $\EventSet$ of events, where each event is a constraint on a set $\VariableSet$ of independent random variables that can be continuous or discrete. Suppose $\EventSet=\{\Event_1,...,\Event_n\}$ and $\VariableSet=\{\Variable_1,...,\Variable_m\}$. 
Let $\VariableSet_i\subseteq \VariableSet$ be a set of variables that completely determines $\Event_i$ for each $i\in [n]$. 
The model can be characterized by an \EventVariable ~which is a bigraph $\BipartiteGraph=([n],[m],E)$ where each pair $(i,j)\in [n]\times [m]$ is an edge if and only if $\Variable_j\in \VariableSet_i$. Then the fundamental challenge of LLL becomes the VLLL problem as follows: Given a bigraph $\BipartiteGraph$, determine its interior $\mathcal{I}(\BipartiteGraph)$ which is the set of vectors $\vec{p}$ such that $\Pr\left(\cap_{\Event\in \EventSet} \Neg{\Event} \right)>0$ for any variable-generated event system $\EventSet$ with \EventVariable ~$\BipartiteGraph$ and probability vector $\vec{p}$. LLLs solving this problem are collectively called variable-LLL.。

The model of variable-generated event systems is important, mainly because most applications of LLL have natural underlying independent variables, e.g., hypergraph coloring \cite{mcdiarmid1997hypergraph}, satisfiability \cite{gebauer2016local,gebauer2009lovasz}, counting solutions to CNF formulas \cite{moitra2016approximate}, acyclic edge coloring \cite{giotis2017acyclic}, etc.
Besides, most results on the algorithmic aspects of LLL are based on this model (see Section \ref{sec:relatedwork}).
% and it has connections to some classical sampling algorithms such as the ``cycle-popping" algorithm
However, there are no special studies on the VLLL problem. A common approach for using LLL in the variable setting is ignoring the variable information and applying abstract-LLL to a dependency graph. This approach only produces results that cannot be better than Shearer's bound. Recently, Harris \cite{harris2016lopsidependency} presents a condition for \emph{lopsided version} \cite{erdHos1991lopsided} of variable-LLL which can go beyond Shearer's criterion, but his condition is based on more information than the \EventVariable ~(i.e., how events disagree on variables is needed). Thus, the VLLL problem remains open. 

Meanwhile, it is widely believed that Shearer's bound is generally not tight for variable-LLL. More precisely, given a bigraph $\BipartiteGraph=(U,V,E)$, its \emph{base graph} is defined as the graph $\DependencyGraph_\BipartiteGraph=(U,E')$  where two nodes $u_1, u_2 \in U$ are adjacent if and only if $u_1,u_2$ share some common neighbor in $\BipartiteGraph$. A property of base graph is that if $\BipartiteGraph$ is an \EventVariable ~of variable-generated event system $\EventSet$, then $\DependencyGraph_\BipartiteGraph$ is a dependency graph of $\EventSet$, which immediately implies that $\mathcal{I}_a(\DependencyGraph_\BipartiteGraph)\subseteq \mathcal{I}(\BipartiteGraph)$. When $\mathcal{I}_a(\DependencyGraph_\BipartiteGraph)\neq\mathcal{I}(\BipartiteGraph)$, we say that Shearer's bound is not tight for $\BipartiteGraph$, or $\BipartiteGraph$ has a gap. The only reported bigraph that has a gap is the $4$-cyclic one \cite{kolipaka2011moser}, namely a bigraph whose base graph is the $4$-cycle. An exact characterization of the conditions for gap existence is far from clear. 

Therefore, we try to solve two closely related peoblems:
\begin{enumerate}
\item VLLL problem: characterize the interior $\mathcal{I}(\BipartiteGraph)$ for any bigraph $\BipartiteGraph$. Kolipaka et al. \cite{kolipaka2011moser} have shown that the Moser-Tardos algorithm is efficient up to the Shearer's bound. However, it remains unknown whether the algorithm converges up to the tight bound of variable-LLL and whether it is efficient even beyond Shearer's bound. Moreover, it is widely believed that better bounds can be obtained through variable-LLL for many combinatorial problems, but how much better can it be? A prerequisite for answering these questions is to know what $\mathcal{I}(\BipartiteGraph)$ is since it tightly upper-bounds the range of variable-LLL. 

%what is exactly What is the extreme criterion for variable-LLL? Kolipaka et al. \cite{kolipaka2011moser} have shown that the Moser-Tardos algorithm is efficient up to the Shearer’s bound. However, what is the limit for Moser-Tardos algorithm? Whether is it efficient up to the tight bound of variable-LLL? It is widely believed that we can obtained better bounds through variable-LLL for many combinatorial problems. However, where is the limit of the power of variable-LLL? To answer these questions, we need to first know the exact criterion for variable-LLL. 

\item Gap problem: characterize the conditions for a bigraph to have a gap. The status in quo of variable-LLL is to ignore variable information and apply abstract-LLL. This over-simplification generally compromises the power of variable-LLL, but it is lossless and can be safely used when there is no gap. In addition, VLLL problem makes sense only when a gap exists,  otherwise it's solved by Shearer's theorem. All this calls for a solution to the gap problem.
\end{enumerate}

% \begin{enumerate}
% \item We know that Shearer's condition is not tight for variable LLL in general, but in which cases this condition is a necessary and sufficient condition? 
% \item What is the necessary and sufficient condition for the variable version of LLL? 
% \end{enumerate}

%In this paper, we mainly solve the above questions. 
%show a first necessary and sufficient condition under which the variable version local lemma holds, and show . 

%\begin{theorem}

%\end{theorem}

\subsection{Related Work} \label{sec:relatedwork} 
\LongVersion
LLL provides a powerful tool to show the existence of some complex combinatorial objects meeting a prescribed collection of requirements. 
\LongVersionEnd
The first result for abstract-LLL was proved by Erd{\H{o}}s and Lov\'{a}sz \cite{erdos1975problems} and the first asymmetric one (Theorem \ref{thm:asymmetric}) was presented in \cite{spencer1977asymptotic}. 
Though these results are useful, they are not tight in general. 
A tight, but not local, criterion (Theorem \ref{thm:shearer}) for abstract-LLL was proposed by Shearer \cite{shearer1985problem} over 30 years ago. 

\LongVersion
Shearer's criterion is hard to verify since it involves all possible independent sets, so efforts have been made to obtain simpler (hence weaker) forms. 
\LongVersionEnd
Pegden \cite{pegden2011highly}\cite{pegden2012lefthanded} introduced \emph{lefthanded-LLL} which does not hold on all dependency graphs, but it is generally tighter than the condition in Theorem \ref{thm:asymmetric} and provides a much simpler form of (tight) conditions on special classes of dependency graphs, e.g., chordal graphs. 
Instead of bounds only working for some dependency graphs, Bissacot et al. \cite{bissacot2011improvement} proposed to improve Theorem \ref{thm:asymmetric} by cluster expansion. 
Kolipaka \cite{kolipaka2012sharper} further introduced a hierarchy of bounds (e.g., the \emph{clique-LLL}) which can be applied to any dependency graph and are all tighter than the condition in Theorem \ref{thm:asymmetric}. 
\LongVersion
Note that almost all the bounds either lose applicability for some dependency graphs or are not tight in general. 
\LongVersionEnd % ******In addition, most of them need auxiliary values like $x:\EventSet \mapsto [0,1)$ defined in Theorem \ref{thm:asymmetric}, which are also needed in our condition. 

Erd\"{o}s and Spenser \cite{erdHos1991lopsided} introduced lopsided-LLL, which extends the results in \cite{erdos1975problems} to lopsidependency graphs. Scott and Sokal \cite{Scott05ondependency} proved that Shearer's condition is tight for lopsided dependency graphs. 

There are settings in which Shearer's bound are not tight in general. 
The best known one may be the variable-generated event systems, whose tight conditions are one of the main contributions of this paper. 
Harris \cite{harris2016lopsidependency} extended the concept of lopsidependency to variable-LLL, and proposed a condition which can go beyond Shearer's bound in some cases, but not so in general. Note that Harris' bound cannot be applied to standard variable-LLL, because the key concept of \emph{orderability} cannot be defined on event-variable graphs alone. 

To make LLL constructive, various sampling algorithms have been proposed so as to avoid all bad events.
%a lot of efforts have been put to design algorithms produce assignments such that all bad events are avoided. 
Algorithm design for LLL is closely related to different bounds mentioned above. 
% In order to solve this problem, 
Beck \cite{beck1991algorithmic} first showed that an \emph{algorithmic version LLL (algorithmic-LLL)} is possible and proposed an efficient deterministic sequential algorithm. 
In that paper, it was required that the degree of the dependency graph under consideration be upper bounded by $2^{n/48}$, which is a very strong restriction. 
Several work has been done to relax this requirement \cite{czumaj2000new,molloy1998further,radhakrishnan1998improved,salavatipoury}. 

Under the model of variable-generated event systems, Moser and Tardos \cite{moser2010constructive} proposed a simple sampling-based algorithm with expected polynomial runtime. Their algorithm is Las Vegas and outputs an assignment to the random variables so as to avoid all bad events. 
Though a strong model is used, the condition needed in their analysis is the same as Theorem \ref{thm:asymmetric} which is even not tight for the abstract-LLL. Pegden \cite{pegden2014extension} proved that Moser and Tardos's algorithm efficiently converges even under the condition of the cluster expansion local lemma.
Kolipaka and Szegedy \cite{kolipaka2011moser} further showed that under the same model, Moser-Tardos algorithm actually works efficiently up to Shearer's bound. 
%Read More: http://epubs.siam.org/doi/abs/10.1137/110828290
% They also devised an algorithm for constructive abstract-LLL. 
Harris \cite{harris2016lopsidependency} presented an algorithm for lopsided version of variable-LLL under the lopsided condition mentioned above. 
%That is, to use Harris' algorithm, the assumption of underlying mutually independent random variables is not sufficient, since it exploits properties of lopsided-LLL. 
It is still open what conditions are tight for an efficient constructive variable-LLL. 
Catarata et al.\ \cite{catarata2017moser} tried experimental methods to observe the possibilities. 

Moser-Tardos algorithm can be naturally parallelized because it is not harmful to do sampling for independent events at the same time. 
Moser and Tardos showed that this parallelization achieves a better expected runtime, but the condition required in their analysis is slightly stronger than that for the sequential case. %Moser-Tardos algorithm (Theorem \ref{thm:asymmetric}). $$\forall \Event_i \in \EventSet: p_i \leq (1 - \varepsilon) x(\Event_i) \prod_{B \in \Gamma(\Event_i)} (1-x(B)).$$
In fact, parallel algorithms for LLL has been considered much earlier than the invention of Moser-Tardos algorithm \cite{alon1991parallel}. 
Recently, there are new researches for parallel algorithms inspired by Moser-Tardos algorithm \cite{haeupler2017parallel,harris2017deterministic}. 
Besides, algorithmic-LLL has been studied using distributed computation models \cite{Brandt2016Lower,chung2014distributed,ghaffari2016improved}. 

Algorithms have also been devised for LLL with dependent variables and other conditions. 
Harris and Srinivasan \cite{harris2014constructive} first considered the space of permutations. 
Achlioptas and Iliopoulos \cite{achlioptas2016random} studied algorithms specified by certain multigraphs. 
Frameworks with resampling oracles are also investigated \cite{achlioptas2016focused,harvey2015algorithmic,kolmogorov2016commutativity}. 

Actually, variable-LLL has strong connection with sampling. Guo et al. \cite{guo2017uniform} proposed an algorithmic framework, called “partial rejection sampling”, which establishes this connection in scenarios such as uniform sampling. In a parallel work, Moitra \cite{moitra2017approximate} presended an algorithm to approximately sample solutions to general \emph{k}-CNF under Lov\'{a}sz Local Lemma-like conditions.

Apart from algorithms, LLL has affected (or has been affected by) many other disciplines, in particular physics. 
For example, alternating-sign independence polynomials of dependency graphs, which is a key element in Shearer's criterion, are also related to the concept of partition functions in statistical physics \cite{Scott2005,guttmann1987comment,todo1999transfer,wood1985exact}. Inspired by this connection, cluster expansion local lemma has been proposed \cite{bissacot2011improvement}, and the lower bound of a singularity point in the hard-core lattice gas model has been improved \cite{kolipaka2012sharper}. LLL has also been enriched by the concept of quantum in physics \cite{ambainis2012quantum,sattath2015constructive,Andras2017}.

%The difference between abstract and variable versions is studied not only for local lemma, but also for other combinatorial and probabilistic problems. 
%For example, Janson \cite{janson2004large} considers the problem of providing tight upper bounds such as Hoeffding-type or Bernstein-type inequalities for $\sup \Pr(\sum_{i=1}^n \Event_i \ge t)$. If the dependency among these $\Event_i$'s is modeled by a dependency graph (abstract version), Janson shows that these bounds depend on the \emph{fractional chromatic number} of the dependency graph. If the dependency is otherwise modeled by a biparitite graph (variable version), one can show that these bounds depend on the \emph{maximum degree} on the side of the ground variables \cite{wang2017learning}. 
%These results also tell us when there is a difference between these two concentration inequalities \cite{pelekis2017holder}. %
% : If the dependency graph has an induced odd circle (and the length is strictly greater than $3$) in it, or a clique in the dependency graph does not share a common variable, then there is a gap between the abstract version and the variable version Hoeffding-type bounds 

\section*{Notation}
\begin{itemize}
\item $[n]$: the set $\{1,2,...,n\}$ for positive integer $n$.
\item $\VariableSet,\VariableSetY$: sets of mutually independent random variables. %The $i$-th entry of $\mathcal{X}$ is denoted by $X_i$.
\item $\Variable,\VariableY$: random variables.  
\item $\vec{p},\vec{q},\vec{r}$: vectors of positive real numbers.%, called probability vectors.
\item $\phi(\cdot)$: given $\vec{p}\in (0,+\infty)^n$, $\phi(\vec{p})\in (0,1]^n$ is the vector whose $i$-th entry is $\min\{1, p_i\}$.
\item $\EventSet,\EventSetB$: sets of events, or sets of cylinders.
\item $\Event,\EventB$: events, or cylinders.
\item $\Neg{\Event}$: the complementary of the event/cylinder $A$.
%\item $\vbl(A)$: the set of independent random variables that determine the event $A$.
\item $\Pr(\Event)$: the probability of event $A$.
\item $\Pr(\EventSet)$: the vector whose $i$-th entry is the probability of the $i$-th event in $\EventSet$.
\item $\mu$: Lebesgue measure on Euclidean (sub)spaces.
\item $\DependencyGraph=(V,E)$: the undirected graph with vertex set $V$ and edge set $E$.
\item $\BipartiteGraph=(V_1,V_2,E)$: the bigraph with vertex set $V_1\cup V_2$ and edge set $E\subseteq V_1\times V_2$. $V_1$ and $V_2$ are called the left part and the right part of $\BipartiteGraph$, denoted by $L(\BipartiteGraph)$ and $R(\BipartiteGraph)$, respectively.
\item $\mathcal{N}_\DependencyGraph(v)$: the neighborhood of vertex $v$ in graph $\DependencyGraph$, or $\mathcal{N}(v)$ when  $\DependencyGraph$ is implicit.
\item $\mathbb{I}^{\{i\}}$: the unit interval in the $i$-th dimension of an Euclidean space, or simply $\mathbb{I}$ when $i$ is implicit.
\item $\mathbb{I}^{S}$: the unit cube $\prod_{i\in S}\mathbb{I}^{\{i\}}$, or simply $\mathbb{I}^m$ when $S=[m]$ for some integer $m$.
\end{itemize}

\section{Results and Discussion}
The main results of this paper are listed and discussed as follows. 
\paragraph{Tight condition for variable-LLL} As we mentioned, Shearer's condition is sufficient and necessary for abstract-LLL, but in general it is not tight for variable-LLL. Our first contribution is a sufficient and necessary condition for variable-LLL, namely an exact characterization of $\mathcal{I}(\BipartiteGraph)$ for any bigraph $\BipartiteGraph$. Characterizing $\mathcal{I}(\BipartiteGraph)$ is equivalent to delimiting its boundary, simply called the boundary of $\BipartiteGraph$ and denoted by $\partial(\BipartiteGraph)$, which consists of the vectors $\vec{p}$ such that $(1-\epsilon)\vec{p}\in \mathcal{I}(\BipartiteGraph)$ and $(1+\epsilon)\vec{p}\notin \mathcal{I}(\BipartiteGraph)$ for any $\epsilon\in(0,1)$.

% \begin{restatable}{theorem}{fill}
% \label{thm:boundaryfills}
% Given a bigraph $\BipartiteGraph=([n], [m], E)$ and $\vec{p}\in \partial_v(\BipartiteGraph)$, let $\vec{d}=(d_1,...,d_m)$ where $d_j$ is the degree of the vertex $j\in[m]$ in $\BipartiteGraph$. Then there is a $\vec{d}$-discrete cylinder set $\EventSet\sim\BipartiteGraph$ such that $ \mu(\EventSet)=\vec{p}$ and $\mu(\cup_{\Event\in \EventSet} \Event)=1$.
% \end{restatable}

%The condition actually defines a mathematical program to characterize the feasible region of $\vec{p}$ with which the variable-LLL holds. 

\begin{restatable}{theorem}{boundary}
\label{solutiontoboundary}
Given a bigraph $\BipartiteGraph=([n], [m], E)$, let $\vec{d}=(d_1,...,d_m)$ where $d_j$ is the degree of the vertex $j\in R(\BipartiteGraph)$. For any vector $\vec{q}\in (0,1)^n$, $\lambda\vec{q}$ lies on the boundary of $\BipartiteGraph$ if and only if $\lambda$ is the optimal solution to the program:
%Given a bigraph $\BipartiteGraph=([n], [m], E)$ and let $\vec{d}=(d_1,...,d_m)$ where $d_j$ is the degree of the vertex $j\in[m]$ in $\BipartiteGraph$, for any probability vector $\vec{q}$, then $\lambda\vec{q}\in \partial_v(\BipartiteGraph)$ if and only if $\lambda$ is the optimum solution to the problem:
% \begin{description}
% \item[ $\min \lambda$]
% \item[  s.t. ] $\sum_{i\in [n]}C_{i,k_1,k_2,...,k_m}\geq 1$ for any $k_j\in [d_j],j\in[m]$;
% \item[ \quad] For any $i\in [n]$ and $j\in[m]$, if $(i,j)\notin E$, then $C_{i,k_1,k_2,...,k_m}$ does not depend on $k_j$;
% \item[ \quad] $\sum_{k_1\in[\DiscreteDimen_1],...,k_m\in[\DiscreteDimen_m]}(\prod_{j\in [m]}x_{jk_j})C_{i,k_1,k_2,...,k_m}= \lambda q_i$ for $i\in [n]$;
% \item[ \quad] $\sum_{k\in \DiscreteDimen_j}x_{jk}=1$ for $j\in[m]$; 
% \item[ \quad] $x_{jk}\in[0,1]$ for $ j\in[m], k\in[d_j]$;
% \item[ \quad] $C_{i,k_1,k_2,...,k_m}\in\{0,1\}$ for $i\in [n], k_j\in [d_j],j\in[m]$. 
% \end{description}
\LongVersion
\begin{align*}
 \min & \quad \lambda\\
 \hbox{s.t.} & \quad \sum_{i\in [n]}C_{i,k_1,k_2,...,k_m}\geq 1 \hbox{ for any } k_j\in [d_j],j\in[m];\\
 & \quad  C_{i,k_1,k_2,...,k_m} \hbox{ does not depend on } k_j \hbox{ for any } (i,j)\in([n]\times[m])\setminus E;\\
 & \quad \sum_{k_1\in[\DiscreteDimen_1],...,k_m\in[\DiscreteDimen_m]}(\prod_{j\in [m]}x_{jk_j})C_{i,k_1,k_2,...,k_m}= \lambda q_i  \hbox{ for } i\in [n];\\
 &\quad \sum_{k\in [\DiscreteDimen_j]}x_{jk}=1 \hbox{ for } j\in[m]; \\
 &\quad x_{jk}\in[0,1] \hbox{ for } j\in[m], k\in[d_j];\\
 &\quad C_{i,k_1,k_2,...,k_m}\in\{0,1\} \hbox{ for } i\in [n], k_j\in [d_j],j\in[m].
\end{align*}
\LongVersionEnd

\ShortVersion
\begin{align*}
 \min & \quad \lambda\\
 \hbox{s.t.} & \quad \sum_{i\in [n]}C_{i,k_1,k_2,...,k_m}\geq 1 \hbox{ for any } k_j\in [d_j],j\in[m];\\
 & \quad \hbox{For any } (i,j)\in([n]\times[m])\setminus E, C_{i,k_1,k_2,...,k_m} \hbox{ does}\\
 & \quad \hbox{ not depend on } k_j;\\
 & \quad \sum_{k_1\in[\DiscreteDimen_1],...,k_m\in[\DiscreteDimen_m]}(\prod_{j\in [m]}x_{jk_j})C_{i,k_1,k_2,...,k_m}= \lambda q_i \\ 
 &\quad \hbox{ for } i\in [n];\\
 &\quad \sum_{k\in [\DiscreteDimen_j]}x_{jk}=1 \hbox{ for } j\in[m]; \\
 &\quad x_{jk}\in[0,1] \hbox{ for } j\in[m], k\in[d_j];\\
 &\quad C_{i,k_1,k_2,...,k_m}\in\{0,1\} \hbox{ for } i\in [n], k_j\in [d_j],j\in[m].
\end{align*}
\ShortVersionEnd
\end{restatable}

As far as we know, this is the first condition for general variable-LLL. It essentially means that the variables can be discretized. Namely, to determine the boundary vectors, it is enough to consider the discrete variables taking $d_j$ values. 
\LongVersion
Small finite domains of the variables enable to study the events by at least the method of exhaustion. 
\LongVersionEnd
In addition, the program facilitates to construct the ``worst-case" set of events, which means that the probability of the union of the events is maximized.

\LongVersion
This optimization problem looks like a geometric program, but it is not the case. Actually, it must be hard to solve, since we show that it is \#P-hard to decide the boundary of variable-LLL.
%Its significance lies in three aspects. First, we can solve the program to get the boundary vectors in any direction. Second, it shows that the worst-case events are discrete. Third,it provides a way to construct the event the probability of whose union is maximized.
\LongVersionEnd

\paragraph{Boundary of cyclic bigraphs}
Though the program above is hard to solve in general, its insight of discretization makes it possible to fully determine the boundary of any cyclic bigraph as in the following theorem. Here a bigraph is called $n$-cyclic if its base graph is a cycle of length $n$.
%can be easily solved for some special and important cases, and 
\LongVersion
We propose a method to calculate the boundary vectors of cycles. 
\LongVersionEnd

%\begin{restatable}[Boundary of cycles]{theorem}{Boundarycycle}
\begin{restatable}{theorem}{Boundarycycle}
\label{cycleboundary}
Given a vector $\vec{p}\in(0,1)^n$, for each $i\in [n]$, let $\lambda_i$ be the minimum positive solution to the equation system: $b_1=\lambda p_i, b_k=\frac{\lambda p_{k+i-1}}{1-b_{k-1}}$ for $2\leq k\leq n-1$, $b_{n-1}=1-\lambda p_{i-1}$. Let $\lambda_0=\min_{i\in[n]}\lambda_i$. Then $\lambda_0\vec{p}$ lies on the boundary of any $n$-cyclic bigraph.\end{restatable}

% Actually, the theorem shows thon the boundary of a cycle of length $n$ is same as the union boundary of a path of length $n$.
\LongVersion
In the literature, cyclic bigraphs are attractive as they are the only example showing a gap exists, i.e., 
\LongVersionEnd
\ShortVersion
In the literature, 
\ShortVersionEnd
only one vector on the boundary of 4-cyclic bigraphs has been identified. 
The above theorem shows that the whole boundary of any $n$-cyclic bigraph can be determined by solving an $(n-1)$-degree polynomial equation. 
\LongVersion
The method works for any cyclic bigraph, no matter whether the probability vector is symmetrical or not. 
\LongVersionEnd

Not only for cyclic bigraphs, we also give a procedure to exactly determine the boundary of treelike bigraphs. A bigraph is called treelike if its base graph is a tree.
%2. 证明对变量版本的边界的判定是#P-hard的，而且近似判定依然如此。不可近似到1/2^m之内，所以最终没写近似判定的困难性。
%\paragraph{Hardness of deciding membership.} Given a bigraph $\BipartiteGraph$ and a probability vector $\vec{p}$, we intuitively believe that it cannot be efficiently decided whether this $\vec{p}$ is in $\Interior(H)$. In this paper, we formally prove this result by a direct reduction in Theorem \ref{thm: hardness2}. Our reduction is from one kind of $3$-SAT problem, called Rtw-Opp-\#3SAT, where each variable appears twice oppositely. 

%3. 对抽象版本的Gap存在性给出了充要条件，即当且仅当依赖图是树，任何一个变量版本都不存在gap。
\paragraph{A sufficient and necessary condition for gap existence}
% Given a bigraph of variables and events, we can construct a corresponding dependency graph: If two events are not connected to any common variable in the original bigraph, then there does not exist an edge between them in the induced dependency graph. 
% Sometimes the distribution (of events) which can be modeled by a bigraph is a proper subset of which can be modeled by the corresponding induced dependency graph, but this is not in general true, i.e., sometimes the two sets are equal. 
% It is natural to ask in which cases the difference exists. 

%
Since a bigraph provides more information than its base graph, it is naturally expected to have a gap, namely Shearer's bound is not tight for bigrpahs. 
% However, many bounds obtained by the abstract version of LLLs are only used the dependency structure provided by dependency graph. A nature question is whether these bounds can be extended with extra variable information. 
We propose a necessary and sufficient condition to decide whether such a gap exist. For conciseness of presentation, we also call a bigraph gapful if it has a gap, and gapless otherwise.
%******Define the boundary and interior, as suggested by reviewer 2.
\begin{restatable}{theorem}{bipartite}
\label{Conj:GapGeom}
Given a bigraph $\BipartiteGraph$ and a vector $\vec{p}$ of positive reals, the following three conditions are equivalent:
\begin{enumerate}
\item For any $\lambda$ such that $\lambda\vec{p}\in\mathcal{I}(\BipartiteGraph)$, there is an exclusive variable-generated event system $\EventSet$ with \EventVariable ~$\BipartiteGraph$ and probability vector $\lambda\vec{p}$. \label{interiorex}
\item For the $\lambda$ such that $\lambda\vec{p}\in\partial(\BipartiteGraph)$, there is an exclusive variable-generated  event system $\EventSet$ with \EventVariable ~$\BipartiteGraph$ and probability vector $\lambda\vec{p}$. \label{boundaryex}
\item $\BipartiteGraph$ is gapless in the direction of $\vec{p}$. \label{gapless}
\end{enumerate}
\end{restatable}
Here the qualifier ``exclusive" means that the events in $\EventSet$ are either independent or disjoint, and ``gapless in the direction of $\vec{p}$" means that for any $\lambda$, $\lambda\vec{p}\in\mathcal{I}(\BipartiteGraph)$ if and only if $\lambda\vec{p}\in\mathcal{I}_a(\DependencyGraph_\BipartiteGraph)$.

By this criterion, one can check the existence of a gap just by examining the bigraph, without computing  Shearer's bound of its base graph.% A byproduct of this theorem is that only the exclusive sets of events, if any, are the worst cases.
% if there exists a set of events that are either exclusive or independent, then only such event sets are the worst-cases in terms of LLL. 
%ones. 
%This is complementary to the well-known insight of Shearer that exclusive sets of events are the worst. 

On this basis, we investigate gap existence for two families of bigraphs. 

\begin{restatable}{theorem}{graph}
\label{dependencygaplessdecision}
Treelike bigraphs are gapless.
\end{restatable}

Based on this theorem, we develop a simple algorithm to efficiently compute Shearer's bound for any dependency graph which is a tree.% As a straightforward corollary, the weak gap problem is completely solved, namely, a graph does not have a gap if and only if it is a tree.

In contrast, we obtain an opposite result for cyclic bigraphs, which considerably extends the only gap-existing example in literature \cite{kolipaka2011moser}.

\begin{restatable}{theorem}{cycle}
\label{cyclesaregapful}
Cyclic bigraphs are gapful.
\end{restatable}

\LongVersion
Another interesting perspective of gaps is dependency-graph-oriented: we say that a graph $\DependencyGraph$ is a-gapful if there is a gapful bigraph whose base graph is $\DependencyGraph$, otherwise it's called a-gapless; $\DependencyGraph$ is said to be strongly a-gapful if any bigraph with $\DependencyGraph$ as base graph is gapful, otherwise it's called strongly a-gapless. Six years ago Kolipaka et al. \cite{kolipaka2011moser} proposed to characterize strongly a-gapful graphs, but the problem remains open. We provide an exact characterization for both concepts.

\begin{restatable}{theorem}{treegraph}
\label{treegraphgapless}
A graph is a-gapless if and only if it is a tree.
\end{restatable}

\begin{restatable}{theorem}{chordal}
\label{stronglyagapfuliffchordal}
A graph is strongly a-gapful if and only if it is chordal.
\end{restatable}
\LongVersionEnd

\paragraph{Reduction method}  To discover more instances that have or have no gaps, we propose a set of reduction rules which allow us transforming a bigraph without changing the existence or nonexisence of a gap. 
We identify five basic operations. Three of them as well as their inverses preserve both gapful and gapless; the other two preserve gapful, while the inverses of the two preserve gapless. Applying these operations, we show that a bigraph is gapful if it contains a gapful one. This, together with Theorem \ref{cyclesaregapful}, intuitively means that Shearer's criterion is not tight for almost all cases of variable-LLL. Likewise, we show that combinatorial bigraphs $H_{n,m}$ are gapful if $m$ is small enough and are gapless if $m$ is large enough. 

%\section{Toy example}
\section{Probability Boundary of Variable-LLL}
This section aims at solving the \emph{VLLL problem}: given a bigraph $\BipartiteGraph$, determine all the vectors $\vec{p}$ such that $\Pr\left(\cap_{\Event\in \EventSet} \Neg{\Event} \right)>0$ for any variable-generated event system $\EventSet$ with \EventVariable ~$\BipartiteGraph$ and probability vector $\vec{p}$. 
Basically, we will transform the problem into a geometric one and solve it in the framework of Euclidean geometry. 

For conciseness of presentation, a variable-generated event system $\EventSet$ is said to conform with a bigraph $\BipartiteGraph$, denoted by $\EventSet\sim\BipartiteGraph$, if $\BipartiteGraph$ is an \EventVariable ~of $\EventSet$.

Throughout this section, we only consider bigraphs whose base graphs are connected. This restriction does not lose generality for the following reason. If a bigraph has disconnected  base graph, itself must also be disconnected and each component is again a bigraph. In this case, the interior of the original bigraph is exactly the direct product of the interiors of the component bigraphs.% Likewise, when studying abstract interiors, we will focus on connected graphs. 
% Likewise, we say that $\EventSet$ conforms with $\DependencyGraph$, denoted by $\EventSet\sim_a\BipartiteGraph$, if $\DependencyGraph$ is a dependency graph of $\EventSet$.

%is the condition of $\BipartiteGraph$ and $\vec{p}$ such that $\Pr(\cup_{\Event\in\EventSet}\Event)<1$ for any event set $\EventSet\sim \BipartiteGraph$ with $\Pr(A)=\vec{p}$?
%Suppose $\BipartiteGraph=([n],[m],E)$ is a bigraph and  $\vec{p}\in(0,1]^n$ is a vector. 
\subsection{A Geometric Counterpart}
Now we formulate a geometric counterpart of the VLLL problem, called the GLLL problem. 
Consider the $m$-dimensional Euclidean space $\mathbb{R}^{m}$ endowed with Lebesgue measure $\mu$. 
Let $\Variable_i$ be the coordinate variable of the $i$-th dimension, $i\in [m]$. 
% For any $S\subseteq [m]$, the $S$-\emph{unit cube}, denoted by $\mathbb{I}^S$, is defined to be the subset of $\mathbb{R}^m$ with $\Variable_i\in [0,1]$ for $i\in S$ and $\Variable_i=0$ for $i\notin S$. 
For any $S\subseteq [m]$, the $S$-\emph{unit cube}, denoted by $\mathbb{I}^S$, is defined to be the $|S|$-dimensional unit hypercube $[0,1]^{|S|}$ working as the domain range of the variables $\{X_i: i \in S\}$ such that for each $i \in S$, $X_i \in [0,1]$. %For notational convenience, we sometimes also denot the $S$-unit cube by $\mathbb{I}^{\mathcal{X}_S}$, where $\mathcal{X}_S=\{X_i: i\in S\}$. 
When $S=[k]$ for some $k\leq m$, we simply write $\mathbb{I}^k$ for $\mathbb{I}^{[k]}$. A \emph{cylinde}r $\Event$ in $\mathbb{I}^m$ is a subset of the form $\EventB\times\mathbb{I}^S$, where $\EventB\subseteq \mathbb{I}^{[m]\setminus S}$ is called a base of $\Event$;
%denoted by $\Xi(\Event)$, 
 define $\dim(\EventB)=[m]\setminus S$. % and $dim_\EventB(\Event)=S$ is called the base dimension set of $\Event$, denoted by $dim(\Event)$
Given a bigraph $\BipartiteGraph=([n], [m], E)$ and a set $\EventSet$ of cylinders $\Event_1,...,\Event_n$ in $\Interval^m$, we say that $\EventSet$ \emph{conforms with} $ \BipartiteGraph$, also denoted by $\EventSet\sim\BipartiteGraph$, if there are bases $\EventB_1,...,\EventB_n $ of $\Event_1,...,\Event_n$ such that $E=\{(i,j)\in[n]\times[m]:j\in \dim(\EventB_i)\}$. %$\vec{p}$ be the vector whose $i$-th entry is $\mu(\Event_i)$.
Now comes the \textbf{GLLL problem}: given bigraph $\BipartiteGraph$, determine all the vectors $\vec{p}$ such that $\mu(\cup_{\Event\in\EventSet}\Event)<1$ for any cylinder set $\EventSet\sim\BipartiteGraph$ with $\mu(\EventSet)=\vec{p}$.  
% what is the condition of $\BipartiteGraph$ and $\vec{p}$ such that $\mu(\cup_{\Event\in\EventSet}\Event)<1$  for any cylinder set $\EventSet\sim \BipartiteGraph$ with $\mu(A)=\vec{p}$?% $\mathbb{R}^m$ with $0\leq\Variable_i$The unit cube $\mathbb{I}^m$ is the subset of $\mathbb{R}^m$ with $0\leq\Variable_i$

One can easily see that the VLLL problem is equivalent to the GLLL problem in the sense that they have the same solutions. Hence, the rest of the paper will be presented in the context of the GLLL problem. For ease understanding, the terms ``event" and ``cylinder" will be used interchangeably, and so will ``probability" and ``Lebesgue measure". The complementary of a cylinder $\Event$ in $\mathbb{I}^{[m]}$ is defined to be the cylinder $\Neg{\Event}=\mathbb{I}^{[m]}\setminus \Event$.

\subsection{A Sufficient and Necessary Criterion}

\begin{definition}[Interior]
The \emph{interior} of a bigraph $\BipartiteGraph$, denoted by $\mathcal{I}(\BipartiteGraph)$, is the set of vectors $\vec{p}$ on $(0,1)$ such that $\mu\left(\cap_{\Event\in \EventSet} \Neg{\Event} \right)>0$ for any cylinder set $\EventSet\sim \BipartiteGraph$ with $\mu(\EventSet) = \vec{p}$.
\end{definition}

\begin{definition}[Exterior]
The \emph{exterior} of a bigraph $\BipartiteGraph=([n], [m], E)$, denoted by $\mathcal{E}(\BipartiteGraph)$, is the set $(0,1]^n\setminus \mathcal{I}(\BipartiteGraph)$.
\end{definition}

%\begin{definition}[Boundary Probability Vector]
%Given a bigraph $\BipartiteGraph=([n], [m], E)$, 
%a vector $\vec{p}\in (0,1]^n$ is called a \emph{boundary probability vector} of $\BipartiteGraph$, if for any $\epsilon\in(0,1)$, the following two conditions hold simultaneously:
%\begin{enumerate}
%\item $\Pr(\cap_{\Event \in \EventSet} \Neg{\Event})>0$ for any event set $\EventSet\sim \BipartiteGraph$ with $\Pr(\EventSet)=(1-\epsilon)\vec{p}$;
%\item There is an event set $\EventSet\sim \BipartiteGraph$ with $\Pr(\EventSet)=\phi((1+\epsilon)\vec{p})$ such that $\Pr(\cap_{\Event\in \EventSet} \Neg{\Event})=0$.
%\end{enumerate}
%\end{definition}

\begin{definition}[Boundary]
The \emph{boundary} of a bigraph $\BipartiteGraph$, denoted by $\partial(\BipartiteGraph)$, is the set of vectors $\vec{p}$ on $(0,1]$ such that $(1-\epsilon)\vec{p}\in \Interior(\BipartiteGraph)$ and $(1+\epsilon)\vec{p}\notin \Interior(\BipartiteGraph)$ for any $\epsilon\in (0,1)$. Any $\vec{p}\in \partial(\BipartiteGraph)$ is called a \emph{boundary vector} of $\BipartiteGraph$.
\end{definition}

We can show that there is a boundary vector in every direction.
\begin{lemma}\label{le:probabilityboundexist}
Given a bigraph $\BipartiteGraph=([n], [m], E)$, for any $\vec{p}\in (0,1]^n$, there exists a unique $\lambda>0$ such that $\lambda\vec{p}\in\partial(\BipartiteGraph)$. 
\end{lemma}
\LongVersion
\begin{proof}
Let $\Lambda\triangleq\{\lambda>0:\lambda\vec{p}\notin \Interior(\BipartiteGraph)\}$. 
If $\lambda$ is so large that $\lambda p_i\geq 1$ for some $i$, then $\lambda\in\Lambda$ since $\lambda\vec{p}\notin \Interior(\BipartiteGraph)$. 
If $\lambda$ is so small that $\lambda \sum_i p_i$ is smaller than $1$, then $\lambda\notin\Lambda$ because $\lambda\vec{p}\in \Interior(\BipartiteGraph)$.
Thus, $\Lambda$ is non-empty and its infimum, denoted by $\lambda_0$, must be positive. It is easy to see that $\lambda_0\vec{p}\in\partial(\mathcal{\BipartiteGraph})$. The uniqueness is trivial. 
\end{proof}
\LongVersionEnd

% * <liuxingwu@ict.ac.cn> 2017-03-27T06:37:41.498Z:
% 
% In the following two paragraphs we should set the geometrical view
% 
% ^ <liuxingwu@ict.ac.cn> 2017-03-27T06:38:53.654Z.
In the rest of this section, we propose a program to characterize boundary vectors. The cornerstone of the program is the observation that cylinders can be properly \textit{discretized} without changing the boundary.% of any bigraph. %Without loss of generality, we assume that every event is a Lebesgue-measurable subset in an $m$-dimensional unit cube, and every random variable corresponds to one dimension of the cube. 

Given an integer $d>0$, a cylinder $\Event\subseteq \mathbb{I}^m$ is said to be $d$-discrete in dimension $j$, if there is a partition of $\mathbb{I}^{\{j\}}$  into $d$ disjoint intervals $\Delta_1,...,\Delta_{d}$ such that $\Event=\cup_{k=1}^{d} S^\Event_k\times\Delta_k$ for some $S^\Event_k\subseteq\mathbb{I}^{[m]\setminus\{j\}}$, $k=1,\ldots, d$. A cylinder set $\EventSet$ is called $d$-discrete in dimension $j$, or discrete in dimension $j$ when $d$ is implicit, if so is every $\Event\in \EventSet$. 
Given a vector $\DiscreteVec=(d_1,...,d_m)$, a cylinder $\Event$ is called $\DiscreteVec$-discrete, if it is $d_j$-discrete in dimension $j$ for any $j\in[m]$.  Likewise, $\EventSet$ is called $\DiscreteVec$-discrete, or discrete when $\DiscreteVec$ is implicit, if so is every $\Event\in \EventSet$; then the vector $\DiscreteVec$ is called a discreteness degree of $\EventSet$.
%
%A set $\EventSet$ of cylinders in $\mathbb{I}^m$ is said to be $d_j$-discrete in dimension $j$, if there is a partition of $\mathbb{I}^{\{j\}}$  into $d_j$ disjoint intervals $\Delta_1,...,\Delta_{d_j}$, such that for every cylinder $\Event\in \EventSet$, there exists $S^\Event_k\subseteq\mathbb{I}^{[m]\setminus\{j\}}$, $k=1,\ldots, d_j$, satisfying $\Event=\cup_{k=1}^{d_j} S^\Event_k\times\Delta_k$.
%If $d_j$ is implicit, we simply say that $\EventSet$ is discrete in dimension $j$.  $\EventSet$ is called $(d_1,...,d_m)$-discrete, or simply called discrete, if it is $d_j$-discrete in dimension $j$ for any $j\in[m]$; then the vector $\DiscreteVec$ is called the discreteness degree of $\EventSet$. It is obvious that every cylinder in a discrete cylinder set is the union of finitely many cubes.

%We can assume that every random variable $\Variable_i$ is uniformly distributed on $[0,1]$. If variable $\Variable_i$ only takes $\DiscreteDimen_i$ values, then we partition $[0,1]$ into $\DiscreteDimen_i$ disjoint subintervals $\Interval_j$, $j=1,\ldots,\DiscreteDimen_i$ and the probability that $\Variable_i$ takes the $j$-th value equals to the probability of $\Variable_i\in \Interval_j$. 

%First of all, set a notation. 
Given two vectors $\vec{p}$ and $\vec{q}$, we say $\vec{p}\leq \vec{q}$ if the inequality holds entry-wise. Additionally, if the inequality is strict on at least one entry, we say that $\vec{p}< \vec{q}$. 

In the rest of this section, fix a bigraph $\BipartiteGraph=([n], [m], E)$ and a probability vector $\vec{p}\in \partial(\BipartiteGraph)$. Let $\vec{q}_\epsilon\triangleq \phi((1+\epsilon)\vec{p})$ for any real number $\epsilon>0$ and $\vec{d}\triangleq (d_1,...,d_m)$ with each $d_j$ being the degree of the vertex $j\in [m]$ in $\BipartiteGraph$.

%Let $\vec{q}\triangleq(1+\epsilon)\vec{p}$ be the probability vector whose $i$-th entry is $q_i=\min\{1,(1+\epsilon)p_i\}$. 

%\begin{definition}[Discrete Implementation]
%Consider a bigraph $G$ and probability vector $\vec{p}$. 
%Suppose that each random variable $\Variable_i$ takes finitely many values $\{\Interval^{\{i\}}_j:j\in D_i\}$ such that for every event $\Event\in\EventSet$  $\Pr[\Event_i] = p_i$. 
%is a finite union of the cubes $\Interval^{\{1\}}_{j_1}\times\cdots\times \Interval^{\{m\}}_{j_m}$. 
%This is called a \emph{$(D_1,\cdots D_m)$-discrete implementation}, or \emph{discrete} for short. If $D_1=\cdots =D_m=D$, it is also called \emph{$D$-discrete}.
%\end{definition}

The main results (Theorem \ref{thm:boundaryfills} and Theorem \ref{solutiontoboundary}) of this section present a discrete cylinder set for each probability vector on the boundary. As a byproduct, it is shown that the boundary lies in the exterior.
Following these theorems, there are two corollaries handling the discretization of interior and exterior respectively. 

The boundary is discretized in four steps, as shown in the coming four lemmas. First, we show that for any $\epsilon>0$, there is a discrete cylinder set whose measure vector lies in the exterior and is $\epsilon$-close to $\vec{p}$. Unfortunately, the discreteness degree of this cylinder set depends on $\epsilon$, and may be unbounded when $\epsilon$ tends to $0$. Second, we show that the set of cylinders can be chosen such that the discreteness degree is no more than $\DiscreteVec$. However, the measure vector may not be lower-bounded by $\vec{p}$, though it is still upper-bounded by $\vec{q}_\epsilon$. Third, with $\epsilon$ tending to $0$, a mathematic program and a calculus argument guarantee the existence of a $\DiscreteVec$-discrete cylinder set whose measure vector lies in the exterior and is upper-bounded by $\vec{p}$. Finally, we show that the measure vector of this cylinder set is exactly $\vec{p}$, which immediately leads to the main theorem.

%starting from the discrete cylinder set sequence indexed by $\epsilon$, we construct a degree uniformly bounded  discrete cylinder set sequence.  For any $\epsilon$, we have a uniform discreteness degree $\DiscreteVec$. Thirdly, we can show there exists a discrete limit cylinder set of this sequence. In these steps, we keep the cylinder sets in the exterior.

%However, the second step has a side effect, that it may reduce the probability.
%Hence, we do not know that the limit probability vector $\vec{r}$ is $\epsilon$ close to the boundary vector $\vec{p}$ for any positive $\epsilon$, but only that it is upper bounded by a probability vector which is arbitrarily close to the boundary vector $\vec{p}$. We get only $\vec{r} \leq \vec{p}$. The final step is to prove $\vec{r} =\vec{p}$, by showing that if $\vec{r} < \vec{p}$ then $\vec{r}$ falls into the interior.

%The intuition of the first step is to approximate continuous things using discrete ones, as integration is the limit of a finite summation. Here, we use Lebesgue integration.
The basic idea of proving the next lemma is to discretize cylinders dimension by dimension. To discretize the $j$-th dimension, the axis $\mathbb{I}^{\{j\}}$ is partitioned so that every cylinder varies little in each part, which naturally leads to an approximation (that is discrete in dimension $j$) to the origin cylinders. The partition is found by approximating an integral with a finite summation.

\begin{lemma}\label{le:finitediscrete}
For any $\epsilon>0$, there exists a discrete cylinder set $\EventSet\sim \BipartiteGraph$ such that $\vec{p}\leq\mu(\EventSet)\leq\vec{q}_\epsilon$ and $\mu(\cup_{\Event\in \EventSet} \Event)=1$.
\end{lemma}
%\marginpar{what are $q$ and $p$?}

\LongVersion
\begin{proof}
Since $\vec{p}\in\partial(\mathcal{E})$, there is a cylinder set $\EventSet'\sim \BipartiteGraph$ such that $\mu(\EventSet')=\vec{q}_{\epsilon/2}$ and $\mu(\cup_{\Event\in\EventSet'} \Event)=1$. 

We prove this lemma by showing the following claim. 

\textbf{Claim}: Suppose there is a cylinder set $\EventSetB\sim\BipartiteGraph$ such that $\mu(\cup_{\EventB\in\EventSetB} \EventB)=1$ and $\vec{q}_\sigma\leq\mu(\EventSetB)\leq\vec{q}_{\epsilon-\sigma}$ for some $0<\sigma<\epsilon/2$. Then there exists a discrete cylinder set $\EventSet\sim \BipartiteGraph$ such that $\vec{p}\leq\mu(\EventSet)\leq\vec{q}_\epsilon$ and $\mu(\cup_{\Event\in \EventSet} \Event)=1$.

\textbf{Proof of the claim}: Arbitrarily fix a cylinder set $\EventSetB=\{\EventB_1,...,\EventB_n\}$ satisfying the condition of the claim. Let $\mathcal{J}=\{j\in[m]: \EventSetB \textrm{ is discrete in dimension }j\}$. We prove the claim by induction on $|\mathcal{J}|$.

\textbf{Basis}: $|\mathcal{J}|=m$. The claim trivially holds.

\textbf{Hypothesis}: The claim holds when  $|\mathcal{J}|>l$.

\textbf{Induction}: Consider $|\mathcal{J}|=l<m$. 

Without loss of generality, assume that $1\notin \mathcal{J}$.% Let $\mathcal{I}=\{i\in [n]: (i,1)\in E\}$ where $E$ is the set of edges of $\BipartiteGraph$. Recall that any $(i,j)\notin E$ means that $\EventB_i$ does not depend on $\Variable_j$.
%and $\delta=\frac{\epsilon}{2m}p_0$. Yuyi:move to another place 

%Arbitrarily choose a variable $\Variable_1\in\VariableSet$. 
For each $i\in[n]$ and $x\in [0,1]$, let $\EventB_i^{(x)}= \EventB_i\cap  (\Variable_1=x)$. By Fubini's Theorem, $\EventB_i^{(x)}$ is Lebesgue measurable for almost all $x\in [0,1]$. Without loss of generality, assume that $\EventB_i^{(x)}$ is Lebesgue measurable for all $x\in [0,1]$. Let $f_i$ be the Lebesgue measurable function on $[0,1]$ such that $f_i(x)=\mu(\EventB_i^{(x)})$
%, where $\mu$ is the Lebesgue measure on the hyperplane $X_1=x$
. Then we have $\mu(\EventB_i)=\int_{[0,1]}f_i(x)d\mu$, where the integration is Lebesgue.
%, where $\mu$ is the Lebesgue measure on $[0,1]$

%Note that $i\in \mathcal{I}$ if and only if $1\in dim(\EventB_i)$.
Let $\delta=\frac{\sigma}{2}p_0$ where $p_0=\min_{i\in [n]}p_i$.
For any integer $1\leq k\leq \lceil\frac{1}{\delta}\rceil$, consider intervals 
\begin{equation}
\begin{split}
\Gamma_k\triangleq\left\{
\begin{array}{ll}
((k-1)\delta,\min\{k\delta,1\}] & \textrm{ if } k>1\\ 
\left[0,\delta\right] & \textrm{ if } k=1
\end{array}\right.
\end{split}
\end{equation}

For each list of integers $1\leq k_1,...,k_n\leq \lceil\frac{1}{\delta}\rceil$, define a set $\Delta_{k_1,...,k_n}=\cap_{i\in [n]} f_{i}^{-1}(\Gamma_{k_i})$. Arbitrarily re-number the $\Delta$'s with non-zero measure into $\Delta_1,\cdots,\Delta_K$, where $K\leq \lceil\frac{1}{\delta}\rceil^n$.  We observe that:
\begin{enumerate}
\item $\cup_{i\in [K]}\Delta_i\subseteq [0,1]$ and $\mu(\cup_{i\in [K]}\Delta_i)=1$;
\item $\Delta_1,\cdots,\Delta_K$ are pairwise disjoint;
\item For any $k\in [K]$, any $x,x'\in \Delta_k$, and any $i\in [n]$, it holds that $|f_{i}(x)-f_{i}(x')|\leq \delta$. 
\end{enumerate}

Since $\mu(\cup_{i\in [n]}\EventB_i)=1$, for any $k\in[K]$, we can choose $x_k\in\Delta_k$ such that $\mu(\cup_{i\in[n]}\EventB_i^{(x_k)})=1$. 

Partition $\mathbb{I}^{\{1\}}$ into disjoint intervals $\Delta'_1,...,\Delta'_K$  such that $\mu(\Delta'_k)=\mu(\Delta_k)$ for any $k\in[K]$.

For each $i\in[n]$, define $\EventB'_i\triangleq \cup_{k\in[K]}(\EventB_i^{(x_k)}\times\Delta'_k)$. One can easily check that for any $i\in[n]$ and $j\in[m]$, $\EventB'_i$ is independent of $\Variable_j$ if so is $\EventB_i$. 
Then the cylinder set $\EventSetB'=\{\EventB'_1,...,\EventB'_n\}$ satisfies:

\begin{enumerate}
\item $\EventSetB'$ conforms with $\BipartiteGraph$;
\item $|\mu(\EventB_i)-\mu(\EventB'_i)|\leq \delta$ for any $i\in[n]$, so $\vec{q}_{\sigma/2}\leq\mu(\EventSetB')\leq\vec{q}_{\epsilon-\sigma/2}$;
\item Since $\cup_{k\in[K]}((\cup_{i\in[n]}\EventB_i^{(x_k)})\times\Delta'_k)=\cup_{i\in[n]}(\cup_{k\in[K]}(\EventB_i^{(x_k)}\times\Delta'_k))=\cup_{i\in [n]}\EventB'_i$, it holds that $\mu(\cup_{i\in [n]}\EventB'_i)=\sum_{k\in[K]}\mu(\cup_{i\in[n]}\EventB_i^{(x_k)})\mu(\Delta'_k)=1$.
\end{enumerate}

Now consider the set $\mathcal{J}'=\{j\in[m]: \EventSetB' \textrm{ is discrete in dimension }j\}$. The construction of $\EventSetB'$ indicates that $\mathcal{J}\cup\{1\}\subseteq \mathcal{J}'$. Hence, $|\mathcal{J}'|\geq l+1$, applying the induction hypothesis to $\EventSetB'$ \textbf{finishes the proof of the Claim}.

The lemma follows immediately.
\end{proof}
\LongVersionEnd

The basic idea of proving the next lemma is as follows. By Lemma \ref{le:finitediscrete}, we have a discrete cylinder set. The vector of the measures of the cylinders that depend on a common variable $\Variable_j$ turns out to be a convex combination of $d_j$-dimensional vectors. A simple combinatorial argument indicates that at most $d_j$ out of the latter vectors are enough to generate (also by convex combination) former one, which immediately implies the desired discreteness degree.

\begin{lemma}\label{le:boundeddiscrete}
%Let $\vec{q'}$ be a probability vector as in Lemma \ref{le:finitediscrete}. 
For any $\epsilon>0$, there exists a $\vec{d}$-discrete cylinder set $\EventSet\sim\BipartiteGraph$ such that $ \mu(\EventSet)\leq\vec{q}_\epsilon$ and $\mu(\cup_{\Event\in \EventSet} \Event)=1$. 
\end{lemma}

\LongVersion
\begin{proof}
By Lemma \ref{le:finitediscrete}, there is a discrete cylinder set $\EventSet'=\{\Event'_1,...,\Event'_n\}\sim\BipartiteGraph$ such that $\mu(\EventSet')\leq\vec{q}_\epsilon$ and $\mu(\cup_{i\in[n]} \Event'_i)=1$. Let $\vec{q'}=\mu(\EventSet')$ and the discreteness degree of $\EventSet'$ be $(d'_1,...,d'_m)$. Now by induction on $l=|\{j\in [m]: \DiscreteDimen'_j > \DiscreteDimen_j\}|$, we show that the existence of such an $\EventSet'$ implies the existence of a desired $\EventSet$. 

\textbf{Basis}: If $l=0$, the lemma holds by letting $\EventSet=\EventSet'$.

\textbf{Hypothesis}: The lemma holds if $l\leq L$. 

\textbf{Induction}:  Consider the case $l=L+1$. Without loss of generality, assume $\DiscreteDimen'_1>\DiscreteDimen_1$. 

By the definition of discreteness, there is a partition of $\mathbb{I}^{\{1\}}$ into $d'_1$ disjoint measurable sets $\Delta_1,...,\Delta_{d'_1}$ such that $\Event'_i=\cup_{k=1}^{d'_1} S_{i,k}\times\Delta_k$ for any $i\in[n]$, where each $S_{i,k}\subseteq\mathbb{I}^{[m]\setminus\{1\}}$. Let $\mathcal{I}=\{i\in [n]: (i,1)\in E\}$. We know that $|\mathcal{I}|=d_1$. Since $\mu(\cup_{i\in[n]} \Event'_i)=1$, we have $\cup_{i\in[n]}S_{i,k}=\mathbb{I}^{[m]\setminus\{1\}}$ up to a set of measure zero, for any $1\leq k\leq d'_1$. 

%Namely, 
%% in the event system $(\EventSetB, \mathcal{X}, \mathbb{X},\vec{q'})$, 
%the first dimensional unit interval $\Interval^{\{1\}}$ is partitioned into $\DiscreteDimen'_1$ subintervals $\Interval^{\{1\}}_1,...,\Interval^{\{1\}}_{\DiscreteDimen'_1}$. 
%For each $1\leq i\leq \DiscreteDimen'_1$, let $\delta_i$ be the length of $\Interval^{\{1\}}_i$. 
%Obviously, $\sum_{1\leq i\leq \DiscreteDimen'_1}\delta_i=1$.
%
%Notation: For any $x\in \Interval^{\{1\}}$ and event $\EventB\in\EventSetB$, we denote $\EventB\cap (\Variable_1=x)$ by $\EventB_{(x)}$. Let $\tilde{\EventSetB}=\{\EventB\in\EventSetB: \Variable_1\in \vbl(\EventB)\}$. By the definition of $\DiscreteDimen_1$, $\DiscreteDimen_1=|\tilde{\EventSetB}|$
%
%For each  $1\leq i\leq \DiscreteDimen'_1$, arbitrarily choose an $x_i\in\Interval^{\{1\}}_i$. Then we have:
%\begin{itemize}
%\item $\cup_{\EventB\in \tilde{\EventSetB}}\EventB_{x_i} \supseteq \overline{\cup_{\EventB\in\EventSetB\setminus\tilde{\EventSetB}}\EventB_{x_i}}$ and $\cup_{\EventB\in\EventSetB\setminus\tilde{\EventSetB}}\EventB_x$ is independent of $x\in \Interval^{\{1\}}$;
%\item For any $\EventB\in \tilde{\EventSetB}$, $\Pr(\EventB)=\sum_{1\leq i\leq \DiscreteDimen'_1}\mathcal{M}(\EventB_{x_i})\delta_i$, where $\mathcal{M}$ is the Lebesgue measure on the hyperplane $\Variable_1=x$. 
%\end{itemize}

Consider $\vec{\pi}=\vec{q'}|_{\mathcal{I}}$, which is a $d_1$-dimensional vector. Note that $\mu(\Event'_i)=\sum_{1\leq k\leq d'_1}\mu(S_{i,k})\delta_k$ for any $i\in\mathcal{I}$, where $\delta_k=\mu(\Delta_k)$. Hence $\vec{\pi}=\sum_{1\leq k\leq d'_1}\delta_k\vec{v}_k$ with each $\vec{v}_k=(\mu(S_{i,k}): i\in \mathcal{I})$ being a vector in the $d_1$-dimensional Euclidean space $\mathbb{R}^{\mathcal{I}}$. Since each $\delta_i\geq 0$ and $\sum_{1\leq i\leq \DiscreteDimen'_1}\delta_i=1$, from the perspective of geometry, $\vec{\pi}$ lies in the convex hull of $\vec{v}_1,\cdots \vec{v}_{\DiscreteDimen'_1}$. The segment between the origin and $\vec{\pi}$ must intersect with the boundary of the convex hull; let $\vec{u}$ be an intersection point. The boundary of the convex hull has a natural triangulation of dimension at most $\DiscreteDimen_1-1$. As a result, $\vec{u}$ must be located inside a simplex spanned by $K\leq\DiscreteDimen_1$ points among $\vec{v}_1,\cdots \vec{v}_{\DiscreteDimen'_1}$. 
Without loss of generality, assume that the $K$ points are $\vec{v}_1,\cdots \vec{v}_K$. Hence, there are $\lambda_1,...,\lambda_K>0$ such that $\vec{u}=\sum_{1\leq k\leq K}\lambda_k\vec{v}_k$ and $\sum_{1\leq k\leq K}\lambda_k=1$. 

For $i\in[n]$, define $\Event''_i=\cup_{1\leq k\leq K}S_{i,k}\times \Delta'_k$, where the disjoint intervals $\{\Delta'_1,...,\Delta'_K\}$ is an partition of $\Interval^{\{1\}}$ and $\mu(\Delta'_k)=\lambda_k$ for $1\leq k\leq K$. For $i\in [n]\setminus\mathcal{I}$, since $\Event'_i$  is independent of $\Variable_1$, $S_{i,k}$ does not depend on $k$, which in turn implies that $\Event''_i=S_{i,1}\times\Interval^{\{1\}}=\Event'_i$. Moreover, one can easily check that for any $i\in \mathcal{I}$ and $j\in[m]$, $\Event''_i$ is independent of $\Variable_j$ if so is $\Event'_i$.

Let  $\EventSet''=\{\Event''_1,...,\Event''_n\}$.
%Denote the resulting set of cylinders by  $\EventSetB'$, and the probability vector of $\EventSetB'$ by $\vec{p'}$. 
We have the following observations:
\begin{enumerate}
\item $\EventSet''$ conforms with $\BipartiteGraph$;
\item $\mu(\cup_{i\in[n]}\Event''_i)=\mu(\cup_{k\in[K]}(\cup_{i\in[n]}S_{i,k})\times\Delta'_k)=\sum_{k\in[K]}\mu(\cup_{i\in[n]}S_{i,k})\lambda_k=1$;
\item $\mu(\EventSet'')\leq \mu(\EventSet')\leq \vec{q}_\epsilon$.
\end{enumerate}
Denote by $(d''_1,...,d''_m)$ the discreteness degree of $\EventSet''$. The construction of $\EventSet''$ indicates that $d''_j\leq d'_j$ for $j>1$, and $d''_1=K\leq d_1$. It holds that $|\{j\in [m]: \DiscreteDimen''_i > \DiscreteDimen_i\}|\leq (L+1)-1=L$. Applying the induction hypothesis to $\EventSet''$ immediately finishes the proof.
\end{proof}
\LongVersionEnd

By Lemma \ref{le:boundeddiscrete}, for any small $\epsilon>0$, there is a $\DiscreteVec$-discrete cylinder set $\EventSet_\epsilon$ whose measure is upper bounded by $\vec{q}_\epsilon$. 
The next lemma claims that this is the case even if $\epsilon=0$. The basic idea is to show that as $\epsilon$ tends to 0, $\EventSet_\epsilon$ converges in some sense and the limit is a $\DiscreteVec$-discrete cylinder set. For this end, we establish an equivalence between the existence of a $\DiscreteVec$-discrete cylinder set and a mathematical program consisting of polynomial constraints. This equivalence, together with an argument based on the continuity of the constraints, ensures that a sequence of $\EventSet_\epsilon$ converges and the limit cylinder set is as desired. 

\begin{lemma}\label{le:boundaryfills}
%Let $\DiscreteVec = (\DiscreteDimen_1,\ldots,\DiscreteDimen_m)$ be as defined in Lemma \ref{le:boundeddiscrete}. 
There is a $\DiscreteVec$-discrete cylinder set $\EventSet\sim\BipartiteGraph$ such that $ \mu(\EventSet)\leq\vec{p}$ and $\mu(\cup_{\Event\in \EventSet} \Event)=1$.
\end{lemma}
\LongVersion
\begin{proof}
Arbitrarily choose a sequence of positive real numbers $\epsilon_l$ such that $\lim\limits_{l\rightarrow\infty}\epsilon_l=0$. 

Now arbitrarily fix an $l> 0$. Define the vector $\vec{q}^{(l)}=(q^{(l)}_1,...,q^{(l)}_n)\triangleq \vec{q}_{\epsilon_l}$. 
By Lemma \ref{le:boundeddiscrete}, there exists a $\vec{d}$-discrete cylinder set $\EventSet^{(l)}=\{\Event^{(l)}_1,...,\Event^{(l)}_n\}\sim\BipartiteGraph$ such that $ \mu(\EventSet^{(l)})\leq\vec{q}^{(l)}$ and $\mu(\cup_{i\in [n]} \Event^{(l)}_i)=1$. Let $\vec{r}^{(l)}=(r^{(l)}_1,...,r^{(l)}_n)\triangleq\mu(\EventSet^{(l)})$. 
The existence of $\EventSet^{(l)}$ is equivalent to the following condition $Q$.

\textbf{Condition} $Q$: there are $x^{(l)}_{jk}\in[0,1]$ for $ j\in[m], k\in[d_j]$ and $C^{(l)}_{i,k_1,k_2,...,k_m}\in\{0,1\}$ for $i\in [n], k_j\in [d_j],j\in[m]$ such that 
%\begin{itemize}
%\item $\sum_{1\leq k\leq n}C_{k,j_1,j_2,...,j_m}\geq 1$ for any $1\leq j_i\leq D_i, 1\leq i\leq m$;
%\item $C_{k,j_1,...,j_i, ...,j_m}$ is independent of $j_i$ if $X_i\notin \mathbb{X}_k$, for any $1\leq k\leq n,1\leq j_r\leq D_r$;
%\item $q'_k=\sum_{\forall i\in [m],j_i\in[D_i]}C_{k,j_1,j_2,...,j_m}\prod_{i\in [m]}x_{ij_i}\leq q_k$ for $1\leq k\leq n$;
%\item $\sum_{1\leq j\leq D_i}x_{ij}=1$ for $1\leq i\leq m$.
%\end{itemize}
%
%
%Condition 2 in Lemma \ref{le:program} is satisfied. Consequently, for each $l\geq 1$, there are $x^{(l)}_{ij}\in[0,1]$ for $1\leq i\leq m, 1\leq j\leq \DiscreteDimen_i$ and $C^{(l)}_{k,j_1,j_2,...,j_m}\in\{0,1\}$ for $1\leq k\leq n$ and $1\leq j_i\leq \DiscreteDimen_i$ such that 
\begin{enumerate}
\item $\sum_{i\in [n]}C^{(l)}_{i,k_1,k_2,...,k_m}\geq 1$ for any $k_j\in [d_j],j\in[m]$;
\item For any $i\in [n]$ and $j\in[m]$, if $(i,j)\notin E$, then $C^{(l)}_{i,k_1,k_2,...,k_m}$ is independent of $k_j$;
\item $\sum_{k_1\in[\DiscreteDimen_1],...,k_m\in[\DiscreteDimen_m]}(\prod_{j\in [m]}x^{(l)}_{jk_j})C^{(l)}_{i,k_1,k_2,...,k_m}=r^{(l)}_i\leq q^{(l)}_i$ for $i\in [n]$, and
\item $\sum_{k\in \DiscreteDimen_j}x^{(l)}_{jk}=1$ for $j\in[m]$.
\end{enumerate}
Intuitively, $\vec{d}$-discreteness means that each dimension $j$ is partitioned into $d_j$ segments, with $x^{(l)}_{jk}$ standing for the length of the $k$-th segment. This leads to a partition of the unit cube $\Interval^m$ into sub-cubes, where the $(k_1,k_2,...,k_m)$-th subcube has measure  $\prod_{j\in [m]}x^{(l)}_{jk_j}$. The variable $C^{(l)}_{i,k_1,k_2,...,k_m}$ indicates whether the $(k_1,k_2,...,k_m)$-th subcube is in the cylinder $\Event^{(l)}_i$. Then the equivalence trivially holds.

Note that each $C^{(l)}_{i,k_1,k_2,...,k_m}$ is binary and $i,k_1,k_2,...,k_m$ all range on finite sets that do not depend on $l$. Hence, there is a subsequence of $l$ such that for any fixed $i,k_1,k_2,...,k_m$, $C^{(l)}_{i,k_1,k_2,...,k_m}$ is a constant denoted by $C_{i,k_1,k_2,...,k_m}$. Without loss of generality, assume that the subsequence is the whole sequence. 

Arbitrarily fix $ j\in[m]$ and $k\in[d_j]$. Then the sequence $\{x^{(l)}_{jk}\}_{l\geq 1}$ must have a convergent subsequence, because the interval $[0,1]$ is a compact topological space. Again without loss of generality, assume that the whole sequence $\{x^{(l)}_{jk}\}_{l\geq 1}$ converges. Denote the limit by $x_{jk}$.

Likewise, without loss of generality, we can assume that the sequence $\{r^{(l)}_i\}_{l\geq 1}$ converges  for any $i\in [n]$. Let $r_i=\lim\limits_{l\rightarrow \infty} r^{(l)}_i$. Obviously, $r_i\leq \lim\limits_{l\rightarrow \infty} q^{(l)}_i=p_i$ for any $i\in [n]$.

Letting $l$ approaches infinity, we can see that $x_{jk}$ with $ j\in[m], k\in[d_j]$ and $C_{i,k_1,k_2,...,k_m}$ with $i\in [n], k_j\in [d_j],j\in[m]$ satisfy the condition $Q$. As a result, there is a $\DiscreteVec$-discrete cylinder set $\EventSet\sim\BipartiteGraph$ such that $ \mu(\EventSet)=(r_1,...,r_n)\leq\vec{p}$ and $\mu(\cup_{\Event\in \EventSet} \Event)=1$. 
\end{proof}

\begin{remark}
The equivalence mentioned in the proof of Lemma \ref{le:boundaryfills} implies a necessary and sufficient condition for deciding the interior of $\BipartiteGraph$. Namely, 
%given a bigraph $\BipartiteGraph=([n], [m], E)$ and the vector $\DiscreteVec $ as defined in Lemma \ref{le:boundeddiscrete}, 
a vector $\vec{q}=(q_1,...,q_n)\in \mathcal{E}(\BipartiteGraph)$ if and only if there are $x_{jk}\in[0,1]$ for $ j\in[m], k\in[d_j]$ and $C_{i,k_1,k_2,...,k_m}\in\{0,1\}$ for $i\in [n], k_j\in [d_j],j\in[m]$ such that 

\begin{enumerate}
\item $\sum_{i\in [n]}C_{i,k_1,k_2,...,k_m}\geq 1$ for any $k_j\in [d_j],j\in[m]$;
\item For any $i\in [n]$ and $j\in[m]$, if $(i,j)\notin E$, then $C_{i,k_1,k_2,...,k_m}$ is independent of $k_j$;
\item $\sum_{k_1\in[\DiscreteDimen_1],...,k_m\in[\DiscreteDimen_m]}(\prod_{j\in [m]}x_{jk_j})C_{i,k_1,k_2,...,k_m}\leq q_i$ for $i\in [n]$, and
\item $\sum_{k\in [\DiscreteDimen_j]}x_{jk}=1$ for $j\in[m]$.
\end{enumerate}
\end{remark}
\LongVersionEnd

For the cylinder set $\EventSet$ obtained in Lemma \ref{le:boundaryfills}, the next lemma claims that $\mu(\EventSet)= \vec{p}$. 
Roughly speaking, if there are $\Event_i$ and $\Event_j$ both depending on $\Variable_l$ and satisfying that $\mu(\Event_i)<p_i$ and $\mu(\Event_j)=p_j$, we can remove a thin slice (perpendicular to the axis $\Variable_l$) from $\Event_j$ and attach it to $\Event_i$. 
After this operation, both $\mu(\Event_i)<p_i$ and $\mu(\Event_j)<p_j$, no extra dependency is brought about, and the whole cube remains been filled up. Iteratively, we can finally get $\mu(\Event_k)<p_k$ for any $k$, which is contradictory to the assumption that $\vec{p}$ is a boundary vector.

\begin{lemma}\label{le:boundaryisunique}
If there is a cylinder set $\EventSet\sim\BipartiteGraph$ such that $ \mu(\EventSet)\leq\vec{p}$ and $\mu(\cup_{\Event\in \EventSet} \Event)=1$, then $\mu(\EventSet)= \vec{p}$.
\end{lemma}

\LongVersion
\begin{proof}
%Arbitrarily fix an event system $(\EventSet, \mathcal{X}, \mathbb{X},\vec{\pi}')$ such that $\Pr(\cup_{A\in\EventSet}A)=1$. 
First of all, we prove the following claim:

\textbf{Claim}: Suppose there exists a cylinder set $\EventSet\sim\BipartiteGraph$ such that $ \mu(\EventSet)<\vec{p}$ and $\mu(\cup_{\Event\in \EventSet} \Event)=1$. Then there are $\epsilon>0$ and a cylinder set $\EventSetB\sim \BipartiteGraph$ satisfying $\mu(\EventSetB)\leq (1-\epsilon)\vec{p}$ and $\mu(\cup_{\EventB\in\EventSetB} \EventB)=1$.

\textbf{Proof of the Claim}: Arbitrarily choose $\EventSet\sim\BipartiteGraph$ such that $ \vec{r}\triangleq\mu(\EventSet)<\vec{p}$ and $\mu(\cup_{i\in[n]} \Event_i)=1$. Assume $\EventSet=\{\Event_1,...,\Event_n\}$, $\vec{p}=(p_1,...,p_n)$, $\vec{r}=(r_1,...,r_n)$. Let $\Delta(\vec{r},\vec{p})=|\{i\in [n]: r_i<p_i\}|$. We proceed by induction on $\Delta(\vec{r},\vec{p})$.

\textbf{Basis}: $\Delta(\vec{r},\vec{p})=n$. Choose $\epsilon>0$ such that $\vec{r}\leq (1-\epsilon)\vec{p}$. The claim trivially holds by letting $\EventSetB=\EventSet$.

\textbf{Hypothesis}: The claim holds for any $\Delta(\vec{r},\vec{p})>K$.

\textbf{Induction}: Consider the case $\Delta(\vec{r},\vec{p})=K<n$. Choose $i,j\in L(\BipartiteGraph)$ such that $r_i<p_i$, $r_j=p_j$, and $\Neighbor_\BipartiteGraph(i)\cap \Neighbor_\BipartiteGraph(j)\neq \emptyset$.
%\begin{itemize}
%\item $r_i<p_i$ and $r_j=p_j$, and
%\item $\vbl(\Event_i)\cap \vbl(\Event_j)\neq\emptyset$.
%\end{itemize}
Such $i,j$ exist due to the assumption that the base graph of $\BipartiteGraph$ is connected. 
%event structure induces a connected dependency graph.

Let $\delta=\min\{p_i-r_i,p_j\}$. Arbitrarily choose $l_0\in \Neighbor_\BipartiteGraph(i)\cap \Neighbor_\BipartiteGraph(j)\subseteq [m]$. Let $D_{x}$ be the cylinder $[x,x+\frac{\delta}{2}]\times\Interval^{[m]\setminus\{l_0\}}\subset \Interval^m$, where $0\leq x\leq 1-\frac{\delta}{2}$. Since $\delta\leq p_j$, there must be some $x$ such that $0<\mu(D_{x}\cap \Event_j)<p_j$. Fix such an $x$. Define $\Event'_j = \Event_j\setminus D_x$ and $\Event'_i=\Event_i\cup D_x$. Consider $\EventSet'=\{\Event'_1,...,\Event'_n\}$ where $\Event'_k=\Event_k$ for $k\in[n]\setminus \{i,j\}$. Let $\vec{r}'=(r'_1,...,r'_n)\triangleq \mu(\EventSet')$. We observe that:

\begin{enumerate}
\item $r'_i\leq r_i+\frac{\delta}{2}<p_i$ and $0< p_j-\frac{\delta}{2}\leq r'_j=p_j-\mu(D_{x}\cap \Event_j)<p_j$, so $\vec{r}'<\vec{p}$;
\item $\EventSet'\sim \BipartiteGraph$, since $\EventSet$ conforms with $\BipartiteGraph$ and $\Event'_k$ does not depend on $\Variable_l$ for any $(k,l)\notin E$% and $dim(\Event'_i)=dim(\Event_i)$ for $i\in[n]$
, and
\item $\mu(\cup_{i\in[n]}\Event'_i)=1$ because $\mu(\cup_{i\in[n]}\Event'_i)\geq \mu(\cup_{i\in[n]}\Event_i))=1$.
\end{enumerate}

Note that $\Delta(\vec{r}',\vec{p})=\Delta(\vec{r},\vec{p})+1>K$. Applying the induction hypothesis to $\EventSet'$, we \textbf{finish the proof of the Claim}. 

Now we get back to prove the lemma. Suppose for contradiction that there is a cylinder set $\EventSet\sim\BipartiteGraph$ such that $ \mu(\EventSet)<\vec{p}$ and $\mu(\cup_{\Event\in \EventSet} \Event)=1$. By the Claim,  there are $\epsilon>0$ and a cylinder set $\EventSetB\sim \BipartiteGraph$ satisfying $\mu(\EventSetB)\leq (1-\epsilon)\vec{p}$ and $\mu(\cup_{\EventB\in\EventSetB} \EventB)=1$. We reach a contradiction since $\vec{p}\in\partial(\BipartiteGraph)$.
\end{proof}
\LongVersionEnd

% \fill*
\ShortVersion
Our main theorem immediately follows from Lemma \ref{le:boundaryfills} and Lemma \ref{le:boundaryisunique}. 
\ShortVersionEnd

\begin{theorem}\label{thm:boundaryfills}
Given a bigraph $\BipartiteGraph=([n], [m], E)$ and $\vec{p}\in \partial(\BipartiteGraph)$, let $\vec{d}=(d_1,...,d_m)$ where $d_j$ is the degree of the vertex $j\in R(\BipartiteGraph)$. Then there is a $\vec{d}$-discrete cylinder set $\EventSet\sim\BipartiteGraph$ such that $ \mu(\EventSet)=\vec{p}$ and $\mu(\cup_{\Event\in \EventSet} \Event)=1$.
\end{theorem}

\LongVersion
\begin{proof}
This immediately follows from Lemma \ref{le:boundaryfills} and Lemma \ref{le:boundaryisunique}.
\end{proof}
\LongVersionEnd

Theorem \ref{thm:boundaryfills} and Lemma \ref{le:boundaryisunique} essentially give a \textbf{necessary and sufficient condition for deciding the boundary}: $\vec{p}$ is a boundary vector if and only if it is a minimal probability vector that allows a cylinder set as in Theorem \ref{thm:boundaryfills}. Due to discreteness, such cylinders have only finitely many forms, so their existence can be checked at least by the exhaustive method. In this sense, not only can we decide boundary vectors, but also \textbf{constructively} find the ``worst-case" cylinders (i.e., the measure of the union is maximized). The method is as in Theorem \ref{solutiontoboundary}. 

\boundary*

Given a solution to the program, $\mathbb{I}^m$ is partitioned into subcubes by cutting every axis $X_j$ into $d_j$ intervals of length $x_{jk_j}$, $k_j\in [d_j]$. For each $i\in[n]$, let $\Event_i$ be the union of the subcubes numbered by $(k_1,k_2,...,k_m)$ with $C_{i,k_1,k_2,...,k_m}=1$. Then $\EventSet=\{\Event_1,...,\Event_n\}$ satisfies the requirement of Theorem \ref{thm:boundaryfills}. 
%each dimension $j$ of $\mathbb{I}^m$ is partitioned into $d_j$ subintervals, and we have a set of $C_{i,k_1,k_2,...,k_m}$'s. For each $i\in[n]$, let $\Event_i$ be the union of the subcubes in $\mathbb{I}^m$ numbered by $(k_1,k_2,...,k_m)$ such that $C_{i,k_1,k_2,...,k_m}=1$. Then $\EventSet=\{\Event_1,...,\Event_n\}$ satisfies the requirement of Theorem \ref{thm:boundaryfills}.

By Theorem \ref{thm:boundaryfills}, for $\vec{p}\in \partial(\BipartiteGraph)$, the worst set of cylinders can be $\vec{d}$-discrete. We will generalize the result to non-boundary vectors. When $\vec{p}$ is in the interior of $\BipartiteGraph$, the basic idea of the next corollary is to add an extra cylinder to the original set of cylinders so that their union has measure $1$. By minimizing the extra cylinder, the union of the original cylinders should be maximized. % This means that the vector of the measures of the cylinders is at boundary. 
Then the discreteness degree follows from Theorem \ref{thm:boundaryfills}.% the worst set of cylinders so that: 1. the cubed is filled, and 2. the vector of the measures of the cylinders is at boundary. Then the discreteness degree 

\begin{corollary}\label{cor:Interiorhasworstcase}
Given a bigraph $\BipartiteGraph=([n], [m], E)$ and $\vec{p}\in \mathcal{I}(\BipartiteGraph)$, define $\vec{d}=(d_1,...,d_m)$ where $d_j$ is the degree of the vertex $j\in R(\BipartiteGraph)$. Let $\vec{d}'=(d_1 +1,...,d_m +1)$. Then there is a $\vec{d}'$-discrete cylinder set $\EventSetB=\argmax_{\EventSet\sim\BipartiteGraph,\mu(\EventSet)=\vec{p}} \mu(\cup_{\Event\in \EventSet } \Event)$.
\end{corollary}

\LongVersion
\begin{proof}
Let $\xi= \sup_{\EventSet\sim\BipartiteGraph,\mu(\EventSet)=\vec{p}}\mu(\cup_{\Event\in \EventSet}\Event)$. Suppose $\xi<1$. Define a bigraph $\BipartiteGraph'=([n+1], [m], E')$ where $E'=E\cup\{(n+1,j):j\in[m]\}$. Let $\vec{p}'\in(0,1]^{n+1}$ be such that $p'_i=p_i$ for $1\leq i\leq n$ and $p'_{n+1}=1-\xi$. 

Arbitrarily choose $\epsilon>0$ and $0<\delta<\epsilon (1-\xi)$. We have two facts:

\begin{enumerate}
\item There is a cylinder set $\EventSet'\sim\BipartiteGraph'$ such that $ \mu(\EventSet')\leq(1+\epsilon)\vec{p}'$ and $\mu(\cup_{\Event\in \EventSet'} \Event)=1$. The reason lies in two aspects. On the one hand, by the definition of $\xi$, there is a cylinder set $\EventSet\sim\BipartiteGraph$ satisfying $ \mu(\EventSet)=\vec{p}$ and $\mu(\cup_{\Event\in \EventSet} \Event)\geq \xi-\delta$. On the other hand, let $\Event_{n+1}$ be an arbitrary cylinder such that $\mu(\Event_{n+1})=(1+\epsilon)(1-\xi)$ and $\cap_{\Event\in \EventSet} \Neg{\Event}\subseteq\Event_{n+1}$. It is easy to check that $\EventSet'=\EventSet\cup \{\Event_{n+1}\}$ is the desired cylinder set.
%
%There is an event system $(\EventSet, \mathcal{X}, \mathbb{X},\vec{p})$ such that $\Pr(\cap_{A\in \EventSet} \bar{A})\leq r+\delta$. For each $A_i\\EventSet$ with $i\in [n]$, arbitrarily choose event $A'_i$ satisfying
%	\begin{itemize}
%	\item $\vbl(A'_i)=\vbl(A_i)$,
%	\item $\Pr(A'_i)=(1+\epsilon)\Pr(A_i)$,
%	\item $A'_i\supset A_i$.
%	\end{itemize}
%Obviously, $\Pr(\cap_{i\in [n]} \overline{A'_i})\leq r+\delta$. Let $A'_{n+1}$ be an arbitrary event in the cube such that $\Pr(A'_{n+1})=(1+\epsilon)r$ and $A'_{n+1}\supseteq \cap_{i\in [n]} \overline{A'_i}$. Define $\EventSet'=\EventSet\cup\{A_{n+1}\}$. One can easily check that $(\EventSet', \mathcal{X}, \mathbb{X}',(1+\epsilon)\vec{p}')$ is an event system and that $\Pr(\cup_{A\in \EventSet'} A)=1$.

\item $\mu(\cup_{\Event\in \EventSet'} \Event)<1$ for any cylinder set $\EventSet'\sim\BipartiteGraph'$ with $\mu(\EventSet')=(1-\epsilon)\vec{p}'$. To show this, arbitrarily choose 
$\EventSet'=\{\Event_1,...,\Event_{n+1}\}\sim\BipartiteGraph'$ with $\mu(\EventSet')=(1-\epsilon)\vec{p}'$. Then $\EventSet=\EventSet'\setminus \{\Event_{n+1}\}$ conforms with $\BipartiteGraph$ and $\mu(\EventSet)=(1-\epsilon)\vec{p}$. By the definition of $\xi$, $\mu(\cup_{\Event\in \EventSet} \Event)\leq \xi$. We further have $\mu(\cup_{\Event\in \EventSet'} \Event)\leq \mu(\cup_{\Event\in \EventSet} \Event)+\mu(\Event_{n+1})=\xi+(1-\epsilon)(1-\xi)<1$.
\end{enumerate}

As a result, $\vec{p}'\in\partial(\BipartiteGraph')$. By Theorem \ref{thm:boundaryfills}, there is a $\vec{d}'$-discrete cylinder set $\EventSet'=\{\Event_1,...,\Event_{n+1}\}\sim\BipartiteGraph'$ such that  $\mu(\EventSet')=\vec{p}'$ and $\mu(\cup_{\Event\in \EventSet'} \Event)=1$. Again, $\EventSet=\EventSet'\setminus \{\Event_{n+1}\}$ conforms with $\BipartiteGraph$ and $\mu(\EventSet)=\vec{p}$. Note that $1=\mu(\cup_{\Event\in \EventSet'} \Event)\leq \mu(\cup_{\Event\in \EventSet} \Event)+\mu(\Event_{n+1})\leq \xi+1-\xi=1$, so 
$\mu(\cup_{\Event\in \EventSet} \Event)=1-\mu(\Event_{n+1})=\xi$.

Now deal with the case $\xi=1$. Let $\BipartiteGraph'$ be as defined above. Define vector $\vec{p}^{(\epsilon)}\triangleq (p_1,...,p_n, \epsilon)$ for any $\epsilon>0$. Using an argument like in the first fact mentioned above, we know that there is a cylinder set $\EventSet^{(\epsilon)}\sim\BipartiteGraph'$ such that $ \mu(\EventSet^{(\epsilon)})\leq \vec{p}^{(\epsilon)}$ and $\mu(\cup_{\Event\in \EventSet^{(\epsilon)}} \Event)=1$. Following the proof of Lemma \ref{le:boundaryfills}, we know that $\vec{p}=\lim_{\epsilon\rightarrow 0}\vec{p}^{(\epsilon)}$ lies in the exterior of $\BipartiteGraph$, contradictory to the assumption that $\vec{p}\in \mathcal{I}(\BipartiteGraph)$. As a result, it is impossible that $\xi=1$. The proof ends.
\end{proof}
\LongVersionEnd

The next corollary indicates that for $\vec{p}\in \mathcal{E}(\BipartiteGraph)$, the discreteness degree is also small. The basic idea is opposite to that of proving Corollary \ref{cor:Interiorhasworstcase}. Some events and/or a part of one are \emph{removed} so that the remaining events exactly fill the cube. Then the rest events are discretized according to Theorem \ref{thm:boundaryfills}. Finally, a slight refinement of the discretization also discretizes the removed events.
\begin{corollary}\label{cor:exteriorfewcuts}
Given a bigraph $\BipartiteGraph=([n], [m], E)$ and $\vec{p}\in \mathcal{E}(\BipartiteGraph)$, define $\vec{d}=(d_1,...,d_m)$ where $d_j$ is the degree of the vertex $j\in R(\BipartiteGraph)$. There is a $\tilde{\vec{d}}$-discrete cylinder set $\EventSet\sim\BipartiteGraph$ such that $ \mu(\EventSet)=\vec{p}$ and $\mu(\cup_{\Event\in \EventSet} \Event)=1$, where $\tilde{d}_{j_0}=d_{j_0}+1$ for some $j_0\in [m]$ and $\tilde{d}_j=d_j$ for $j\neq j_0$.
\end{corollary}

\LongVersion
\begin{proof}
We prove by induction on $n$. 

\textbf{Basis: }$n=1$. It trivially holds.

\textbf{Hypothesis: } The lemma holds whenever $n<N$.

\textbf{Induction: } Consider $n=N$. Define $\BipartiteGraph'=([n-1], [m], E')$ to be an induced subgraph of $\BipartiteGraph$ , and $\vec{p}'=(p_1,...,p_{n-1})$. Now we proceed case by case.

\textbf{Case 1}: $\vec{p}'\in \mathcal{E}(\BipartiteGraph')$. By the induction hypothesis, there is a $\tilde{\vec{d}}'$-discrete cylinder set $\EventSet'\sim\BipartiteGraph'$ such that $ \mu(\EventSet')=\vec{p}'$ and $\mu(\cup_{\Event\in \EventSet'} \Event)=1$, where $\tilde{d}'_{j_0}=d'_{j_0}+1$ for some $j_0\in [m]$, $\tilde{d}'_j=d'_j$ for $j\neq j_0$, and each $d'_j=|\{i\in[n-1]: (i,j)\in E'\}|$. 

Without loss of generality, assume that $(n,m)\in E$. The discreteness of $\EventSet'$ in dimension $m$ means that $\mathbb{I}^{\{m\}}$ is partitioned into $\tilde{d}'_m$ disjoint intervals. Now refine the partition into $\tilde{d}'_m +1$ intervals such that the union of some intervals is $[0,p_n]$. Let $\EventSet=\EventSet'\cup\{\Event_n\}$, where $A_n=\Interval^{m-1}\times [0,p_n]$. Then $\mu(\EventSet)=\vec{p}$ and $\mu(\cup_{\Event\in \EventSet} \Event)=1$.

As to the discreteness, obviously $\EventSet$ is $(\tilde{d}'_1,...,\tilde{d}'_{m-1},\tilde{d}'_m +1)$-discrete. If $j_0=m$, then $\tilde{d}'_j=d'_j\leq d_j$ for $j\leq m-1$, and $\tilde{d}'_m+1=(d'_m+1)+1= d_m+1$%, so $\EventSet$ is $\vec{d}$-discrete
. If $j_0\neq m$, then $\tilde{d}'_j=d'_j\leq d_j$ for $j\notin\{j_0,m\}$, $\tilde{d}'_{j_0}=d'_{j_0}+1\leq d_{j_0} +1$, and $\tilde{d}'_m +1=d'_m +1= d_m$. As a result, let $\tilde{\vec{d}}=(d_1,...,d_{j_0-1},d_{j_0}+1,d_{j_0+1},...,d_m)$, and we always have that $\EventSet$ is $\tilde{\vec{d}}$-discrete.

%there is $i_0\in [m]$ and a $(\tilde{D}'_1,...,\tilde{D}'_m)$-discrete event system $(\EventSet', \mathcal{X}, \mathbb{X}',\vec{p}')$ such that $\Pr(\cap_{A\in \EventSet'} \bar{A})=0$, where $\tilde{D}'_{i_0}=D'_{i_0}+1$, $\tilde{D}'_i=D'_i$ for $i\neq i_0$, and each $D'_i=|\{j\in [n-1]: X_i\in \mathbb{X}_j\}|$. Arbitrarily choose $j\in [m]$ such that $X_j\in \mathbb{X}_n$. By definition of discrete event system, the $j$th axis $\Interval^{\{j\}}$ is partitioned in at most $\tilde{D}'_i$ intervals. Now refine the partition into at most $\tilde{D}'_i+1$ intervals the union of some of which is exactly $[0,p_n]$. Let $\EventSet=\EventSet'\cup\{A_n\}$ where $A_n=\Interval^{[m]\setminus \{j\}}\times [0,p_n]$. We have that $(\EventSet, \mathcal{X}, \mathbb{X},\vec{p})$ is a $(\tilde{D}_1,...,\tilde{D}_m)$-discrete event system and $\Pr(\cap_{A\in \EventSet} \bar{A})=0$.

\textbf{Case 2}: $\vec{p}'\in\mathcal{I}(\BipartiteGraph')$. Define $\vec{p}''=(p_1,...,p_{n-1},p''_n)\in (0,1]^n$, where $0<p''_n\leq p_n$ is chosen such that $\vec{p}''\in \partial(\BipartiteGraph)$. By Theorem \ref{thm:boundaryfills}, there is a $\vec{d}$-discrete cylinder set $\EventSet''\sim\BipartiteGraph$ such that $\mu(\EventSet'')=\vec{p}''$ and $\mu(\cup_{\Event\in \EventSet''} \Event)=1$. Again, without loss of generality, assume that $(n,m)\in E$. Then as in Case 1, the discreteness of $\EventSet''$ implies a partition of $\mathbb{I}^{\{m\}}$. We likewise refine that partition and construct the desired cylinder set $\EventSet$. The detail is omitted.
\end{proof}
\LongVersionEnd

%\begin{corollary}\label{co:innercannotfill}
%For any event structure $\mathcal{E}=(\mathcal{X}, \mathbb{X})$ and $\vec{p}\in(0,1]^n$, $\vec{p}\in\mathcal{I}_v(\mathcal{E})$ if and only if $\Pr(\cap_{A\in \EventSet} \bar{A} )>0$ for all event systems $(\EventSet, \mathcal{X}, \mathbb{X},\vec{p})$.
%\end{corollary}

\begin{remark}
The above theorems and corollaries mean that given a bigraph and a vector in $(0,1]^n$, the worst case cylinders can be discretized. More importantly, the discreteness degree is determined by the bigraph only.
\end{remark}

The discreteness degrees mentioned in Theorems \ref{thm:boundaryfills}, \ref{solutiontoboundary} and Corollaries \ref{cor:Interiorhasworstcase}, \ref{cor:exteriorfewcuts} are tight in general.  For example, consider the complete bigraph $\BipartiteGraph=([n], [1], E)$. For any $\vec{p}\in (0,1]^n$,  $\vec{p}\in\mathcal{I}(\BipartiteGraph)$ if and only if $\sum_{i\in[n]} p_i<1$, while $\vec{p}\in\partial(\BipartiteGraph)$ if and only if $\sum_{i\in[n]} p_i=1$. One can easily check that the discreteness degrees in the Theorems and Corollaries are the smallest possible for this example.

\section{Breaking cycles}\label{sec:cyclicboundary}
In this section, we compute the boundary of cyclic bigraphs. Roughly speaking, a cyclic bigraph models the variable-generated system of events where events are located on a cycle and neighbors (and only neighbors) depend on common variables. 
% On the contrary, the bigraph is linear if the cycle is replaced by a path. 
Note that the only gapful bigraph reported in the literature is $4$-cyclic \cite{kolipaka2011moser}. 
% Cyclic bigraphs are interesting because the only gapful bigraph reported in the literature is cyclic.

%\begin{definition}[Induced dependency graph]
%Given a bigraph $\BipartiteGraph=([n], [m], E)$, the \textit{induced dependency graph}, denoted by $G_{\BipartiteGraph}$, is an undirected graph whose vertex set is $[n]$ and an edge between $i,j\in [n]$ exists if and only if $i,j$ share a neighbor in $\BipartiteGraph$.
%\end{definition}

\begin{definition}[Cyclic bigraph]
A bigraph $\BipartiteGraph$ is said to be $n$-cyclic if the base graph $G_{\BipartiteGraph}$ is a cycle of length $n$. When $n=3$, it additionally requires $\cap_{i\in L(\BipartiteGraph)}\Neighbor_\BipartiteGraph(i)=\emptyset$. In case of no ambiguity, an $n$-cyclic bigraph is simply called a cyclic bigraph.
\end{definition}

As far as the GLLL problem is concerned, an $n$-cyclic bigraph is always equivalent to the canonical one $\BipartiteGraph_n=([n], [n], E_n)$ where $E_n=\{(i,i), (i,(i+1)(\overline{\textrm{mod}} ~n)): i\in [n]\}$. 
Here the value $k(\overline{\textrm{mod}} ~n)$ is defined to be $(k-1)(\textrm{mod} ~n)+1$. 
Hence, we will focus on $\BipartiteGraph_n$ in the rest of this section.

To simplify notation,  the operator ``$(\overline{\textrm{mod}} ~n)$" will be omitted whenever clear from context.

A concept that is opposite to cyclic bigraphs is as follows.

\begin{definition}[Linear bigraph]
A bigraph $\BipartiteGraph$ is said to be $n$-linear if the base graph $G_{\BipartiteGraph}$ is a path of length $n$. In case of no ambiguity, an $n$-linear bigraph is simply called linear. 
\end{definition}

A rather surprising phenomenon of cyclic bigraphs is that they can be reduced to linear bigraphs in the following sense: Any boundary vector of an $n$-cyclic bigraph is also that of an $n$-linear one. 
That is, to find the boundary vector in a certain direction, some pair of neighboring events can be decoupled (i.e., become independent of each other) by ignoring their shared variables. In this sense we say that the cycle is broken. The result is stated in the next theorem.

% Similarly, we have the canonical $n$-linear bigraph $P_n=([n], [n-1], E'_n)$ where $E'_n=E_n-\{(n,n),(n-1,n)\}$.
\begin{theorem} \label{thm: break cycle}
For any vector $\vec{p}\in \partial(\BipartiteGraph_n)$, there is a $\vec{d}$-discrete cylinder set $\EventSet\sim\BipartiteGraph_n$ such that $\mu(\EventSet)=\vec{p}$, $\mu(\cup_{\Event\in \EventSet} \Event)=1$, and $\vec{d}< (2,2,...,2)$.
%For any vector $\vec{p}\in \partial(\BipartiteGraph_n)$,there is a set of  can be realized by a discrete event set  $\EventSet=\{\Event_1,...,\Event_n\}$ conforming with $\BipartiteGraph_n$, such that there is cyclic permutation $\pi$ over $[n]$,  $\pi(\EventSet)$ is $(2,2,\ldots,2,1)$-discrete, equivalently $\pi(\EventSet) \sim P_n$, where $\pi(\EventSet)=\{\Event_{\pi(1)},...,\Event_{\pi(n)}\}$.
\end{theorem}
\begin{remark}
$\vec{d}< (2,2,...,2)$ means that $d_j=1$ for some $j\in [n]$. Then all the cylinders (especially $\Event_j$ and $\Event_{j+1}$) are independent of $X_j$. As a result, $\EventSet$ also conforms with $\BipartiteGraph_n^{\{j\}}$, the $n$-linear bigraph obtained by removing the vertex $j\in R(\BipartiteGraph_n)$, meaning that $\vec{p}\in \mathcal{E}(\BipartiteGraph_n^{\{j\}})$. Due to the assumption that $\vec{p}\in \partial(\BipartiteGraph_n)$ and the easy fact that $\mathcal{E}(\BipartiteGraph_n^{\{j\}}) \subseteq \mathcal{E}(\BipartiteGraph_n)$, $\vec{p}$ must also lie on the boundary of $\BipartiteGraph_n^{\{j\}}$.%the and mplying that is  there is cyclic permutation $\pi$ over $[n]$,  $\pi(\EventSet)$ is $(2,2,\ldots,2,1)$-discrete, equivalently $\pi(\EventSet) \sim P_n$, where $\pi(\EventSet)=\{\Event_{\pi(1)},...,\Event_{\pi(n)}\}$.
\end{remark}

\ShortVersion
Now we give a sketchy proof of Theorem \ref{thm: break cycle}.
\ShortVersionEnd

\LongVersion
To prove Theorem \ref{thm: break cycle}, 
\LongVersionEnd
\ShortVersion
A\ShortVersionEnd
\LongVersion
first a\LongVersionEnd
rbitrarily 
fix $\vec{p}\in \partial(\BipartiteGraph_n)$. By Theorem \ref{thm:boundaryfills}, there is a $(2,2,...,2)$-discrete cylinder set $\EventSet\sim\BipartiteGraph_n$ such that $\mu(\EventSet)=\vec{p}$ and $\mu(\cup_{\Event\in \EventSet} \Event)=1$. 
Arbitrarily choose such a cylinder set $\EventSet=\{\Event_1,...,\Event_n\}$. For each $i\in [n]$, let $\EventB_i$ be the base of $\Event_i$ such that $\dim(\EventB_i)=\{i,i+1\}$.

\LongVersion

Define function $F$ as follows. For any $S_1 \subseteq \Interval^{\{i,j\}}$ and $S_2 \subseteq \Interval^{\{j,k\}}$ where $i \not \eq j \not \eq k$, let $F(S_1,S_2)$ be the largest set $S\subseteq\Interval^{\{i,k\}}$ such that $S\times\Interval^{[n]\setminus\{i,k\}}\subseteq S_1\times\Interval^{[n]\setminus\{i,j\}} \cup S_2 \times\Interval^{[n]\setminus\{j,k\}}$. Let $F(S_1,S_2,...,S_l)= F(S_1,F(S_2,S_3,...,S_l))$ for any $2<l\leq n$. For any $i\in[n]$, there is a set $B_i\subseteq \Interval^{\{i,i+1\}}$ such that $\Event_i=\EventB_i\times \Interval^{[n]\setminus \{i,i+1\}}$. For any $1\leq i, j\leq n$, let $\EventSetB_{ij}=\{B_i,...,B_j\}$ if $i\leq j$, otherwise $\EventSetB_{ij}=\{B_i,...,B_n,B_1,...,B_j\}$.

Note that $F(\EventSetB_{i,j}) = F(B_i,F(\EventSetB_{i+1,j})) = F(F(\EventSetB_{i,j-1}),\EventSetB_j)$ for any $i \not \eq j \in [n]$. For simplify, let $F(\EventSetB_{i,i})=\EventB_i$. To emphasize this important definition $F(\EventSetB_{i,j})$, we leave it to readers to verify that $\mu(\cup_{\Event\in \EventSet} \Event)=1$ if and only if $\mu(F(\EventSetB_{i,j}) \cup F(\EventSetB_{j+1,i-1}))=1$ for any $i\leq j$.

Due to the discreteness of $\EventSet$, each $\Interval^{\{i,j\}}$ is partitioned into four rectangles as in Figure \ref{fig:rectangularpartition} and only unions of some of the rectangles make sense. Especially interesting is the 14 types of non-trivial unions, namely $T_{11}$ through to $T_{44}$, grouped into the four categories $T_1,...,T_4$, as shown in Figure \ref{fig:typesoftheevents}. For any $i,j\in [n]$, $B_i$ and $F(\EventSetB_{i,j})$ must have one of the 14 types in $\Interval^{\{i,i+1\}}$ and in $\Interval^{\{i,j+1\}}$, respectively.
%We sometimes call them $\langle i,j\rangle$-types, in order to make clear the dimensions and strictly keep the meaning and subscripts of ``$a,b$" as in Figure \ref{fig:typesoftheevents}.

%
%Arbitrarily fix $\vec{p}\in \partial(\BipartiteGraph_n)$.
%By Theorem \ref{thm:boundaryfills}, there is a cylinder set $\EventSet\sim\BipartiteGraph_n$ such that $\mu(\EventSet)=\vec{p}$, $\mu(\cup_{\Event\in \EventSet} \Event)=1$, and $\EventSet$ is $2$-discrete in every dimension.
%
%Arbitrarily choose such an $\EventSet=\{\Event_1,...\Event_n\}$.
%\begin{lemma} \label{le: fakj}
%For any $i\in [n]$, $\EventB_i=\Interval^{\{i,i+1\}}\setminus F(\EventSet_{i+1,i-1})$.
%\end{lemma}
%\begin{proof}
%Since $\mu(\cup_{\Event\in \EventSet} \Event)=1$, we know that $\mu(\EventB_i \cup F(\EventSet_{i+1,i-1}))=1$, which is equivalent to  $\EventB_i\supseteq \Interval^{\{i,i+1\}}\setminus F(\EventSet_{i+1,i-1})$. By Lemma \ref{le:boundaryisunique}, since $\vec{p}$ is a boundary vector, the measure of $\EventB_i$ must be minimized given the other cylinders. As a result, $\EventB_i=\Interval^{\{i,i+1\}}\setminus F(\EventSet_{i+1,i-1})$.
%\end{proof}

\begin{lemma} \label{le: fakj}
For any $i,j\in [n]$, $F(\EventSetB_{j+1,i-1})=\Interval^{\{i,j\}}\setminus F(\EventSetB_{i,j})$.
\end{lemma}
\begin{proof}
Since $\mu(\cup_{\Event\in \EventSet} \Event)=1$, we know that $\mu(F(\EventSetB_{j+1,i-1}) \cup F(\EventSetB_{i,j}))=1$, which is equivalent to  $F(\EventSetB_{j+1,i-1})\supseteq \Interval^{\{i,j\}}\setminus F(\EventSetB_{i,j})$. By the definition of $F$, each of $F(\EventSetB_{i,j})$ and $F(\EventSetB_{j+1,i-1})$ also has one of the 14 types in Figure \ref{fig:typesoftheevents}. 

Suppose for contradiction that $F(\EventSetB_{i,j}) \cap F(\EventSetB_{j+1,i-1})) \not \eq \emptyset$. Then one of the rectangles in $F(\EventSetB_{i,j})$ can be removed so as to preserve the property $\mu(F(\EventSetB_{j+1,i-1}) \cup F(\EventSetB_{i,j}))=1$. It is straightforward to see that this can be achieved by removing a rectangle from either $B_i$ or $F(\EventSetB_{i+1,j})$. Iteratively, we see that $\mu(F(\EventSetB_{j+1,i-1}) \cup F(\EventSetB_{i,j}))=1$ remains true event if one element in $\EventSetB_{i,j}$ gets smaller. Considering Lemma \ref{le:boundaryisunique} and the assumption that $\vec{p}$ is a boundary vector, we reach a contradiction.
\end{proof}
\LongVersionEnd

Now, we explore how $\EventB_i, \EventB_{i+1}$ are correlated in terms of their types. 

If some $\EventB_i$ has type $T_2$, then $\Event_i$ is independent of either $\Variable_i$ or $\Variable_{i+1}$. It is easy to see that $\EventSet$ is 1-discrete either in dimension $i$ or in dimension $i+1$. Hence we have

\begin{lemma} \label{le: type2enough}
If $\EventB_i$ has type $T_2$ for some $i\in [n]$, then $\EventSet$ has a discreteness degree smaller than $(2,2,...,2)$.
\end{lemma}
\LongVersion
\begin{proof}
Arbitrarily choose $i\in [n]$ such that $\EventB_i$ has type $T_2$. Without loss  of generality, assume that $\EventB_i$ is independent of $X_i$. This means that all events except $\EventB_{i-1}$ are independent of $X_i$. By Lemma \ref{le: fakj}, $\EventB_{i-1}=\Interval^{\{i-1,i\}}\setminus F(\EventSetB_{i,i-2})$ which is also independent of $X_i$. As a result, $\EventSet$ do not depend on $X_i$, namely, it is 1-discrete in dimension $i$.
\end{proof}
\LongVersionEnd

As a result, in the rest of this section, it is assumed  that no bases have type $T_2$.
\LongVersion

\begin{lemma} \label{le: for break cycle}
For any $i,j\in [n]$ such that $i<j$ or $i>j+1$, we have the following observations.
\begin{enumerate}
\item  If $F(\EventSetB_{i,j})$ has type $T_1$, then  both $\EventB_i$ and $F(\EventSetB_{i+1,j})$ has type $T_1$.
\item  If $F(\EventSetB_{i,j})$ has type $T_3$, then both $\EventB_i$ and $F(\EventSetB_{i+1,j})$ has type $T_3$.
\item  Suppose that none of $\EventB_i,...,\EventB_j$ has type $T_2$. If $F(\EventSetB_{i,j})$ has type $T_4$, then one of $\EventB_i$ and $F(\EventSetB_{i+1,j})$ has type $T_1$, and the other has type $T_4$.
\end{enumerate}
%If $F(\EventSet_{23})$ has type 4, then  the types of  $\EventB_2$ and $\EventB_3$ are 1 and 4.
\end{lemma}
\begin{proof}
In each case, it is straightforward to check two facts. First, the claimed combination of types of $\EventB_i$ and $F(\EventSetB_{i+1,j})$ is feasible, namely, it can produce the given type of $F(\EventSetB_{i,j})$. Second, this combination is minimum in the following sense: if $\EventB_i$ and $F(\EventSetB_{i+1,j})$ have other feasible types, then at least one of them can be reduced without changing $F(\EventSetB_{i,j})$. Thus, similar to the proof of Lemma \ref{le: fakj}, at least one element in $\EventSetB_{i,j}$ can be reduced without changing $F(\EventSetB_{i,j})$.
The detailed proofs of the two facts are omitted.

By the second fact and Lemma \ref{le:boundaryisunique}, since $\vec{p}$ is a boundary vector, we know that the other feasible combinations are impossible.
\end{proof}

\LongVersionEnd

\begin{figure}
\centering
\includegraphics[scale=0.4]{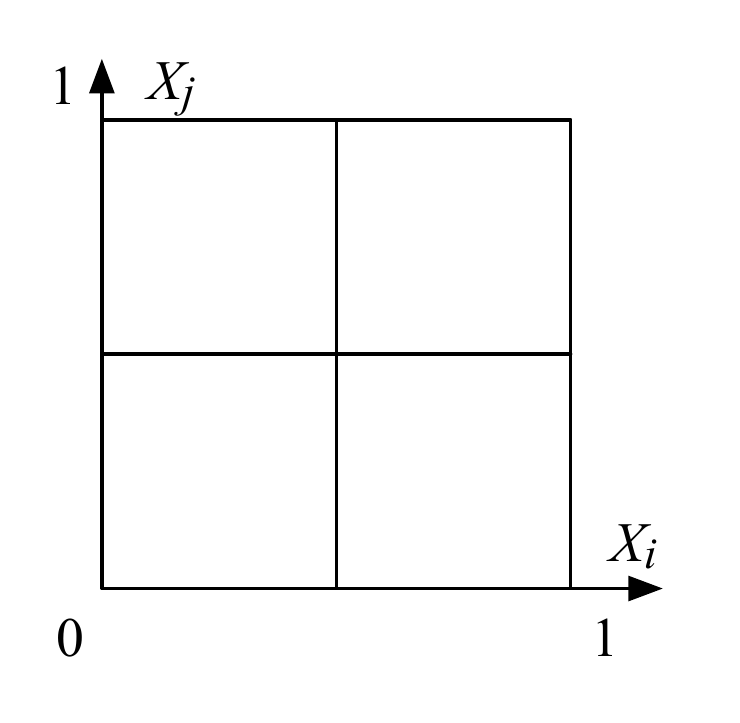}
\caption{Partitioning $\Interval^{\{i,j\}}$  into four rectangles}\label{fig:rectangularpartition}
\end{figure}
\begin{figure}
\centering
\includegraphics[scale=0.5]{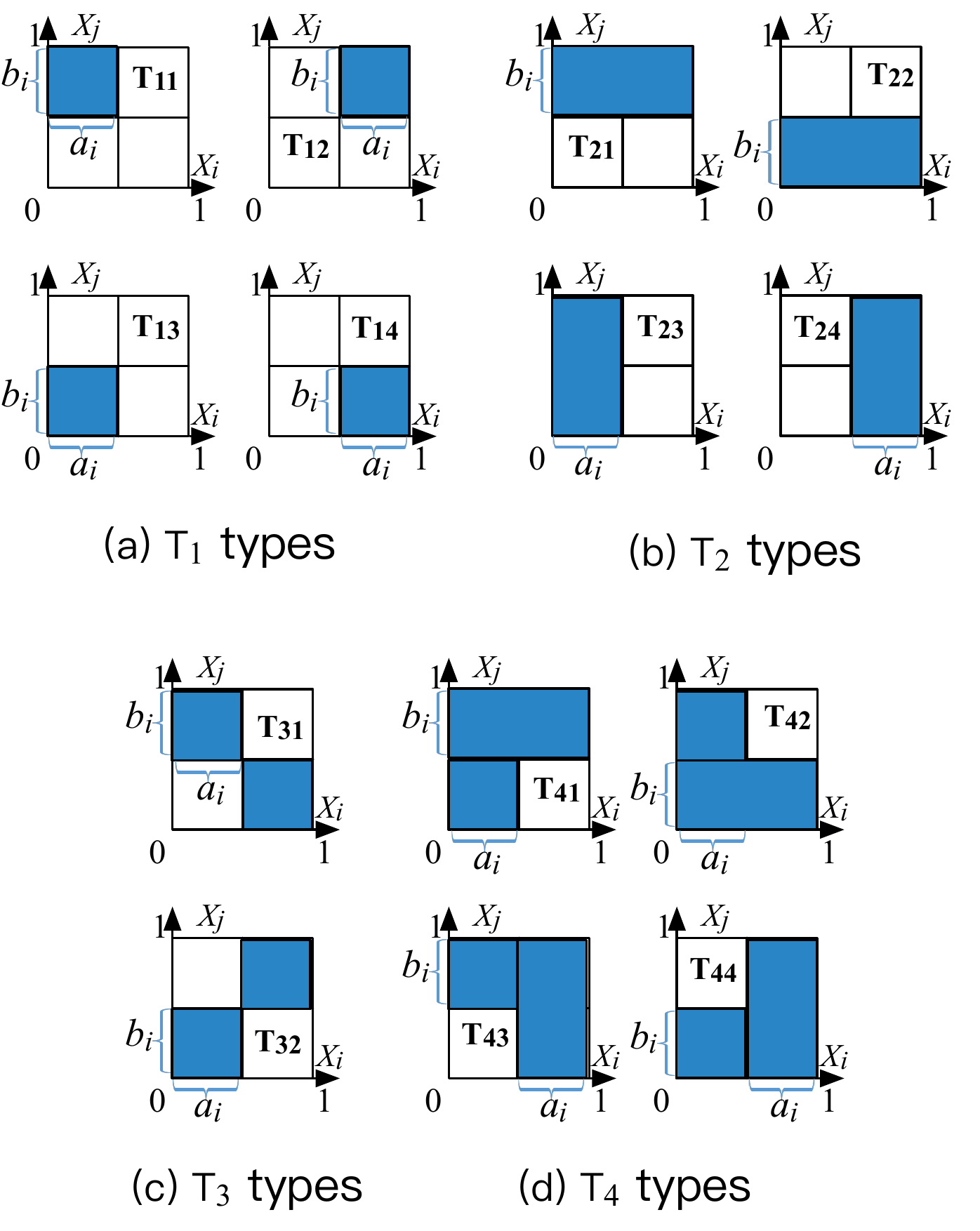}
\caption{Types $T_{11}$ through to $T_{44}$  in $\Interval^{\{i,j\}}$, indicated by the shaded areas}\label{fig:typesoftheevents}
\end{figure}
Now we can show that there are at most two essentially different possibilities of the types of $\EventB_1,...,\EventB_n$, as indicated in the following lemma.
\begin{lemma} \label{le: possibletypes}
There are at most two possible combinations of the types of the bases.
\begin{enumerate}
\item $T_1$-dominant: all bases have type $T_1$ except one has type $T_4$.
\item $T_3$-dominant: all bases have type $T_3$.
\end{enumerate}
\end{lemma}
\LongVersion
\begin{proof}
We do a case by case analysis.

\textbf{Case 1:} $B_1$ has type $T_1$. By Lemma \ref{le: fakj}, $F(\EventSetB_{2,n})$ has type $T_4$. Applying Lemma \ref{le: for break cycle} to $F(\EventSetB_{2,n})$ results in two possibilities. One is that $B_2$ has type $T_4$ and $F(\EventSetB_{3,n})$ has type $T_1$, and the other is that $B_2$ has type $T_1$ and $F(\EventSetB_{3,n})$ has type $T_4$. Then iteratively apply Lemma \ref{le: for break cycle} to $F(\EventSetB_{3,n})$. Altogether, we see that all $\EventB_i$ have type $T_1$ except one has type $T_4$.

\textbf{Case 2:} $B_1$ has type $T_3$. By Lemma \ref{le: fakj}, $F(\EventSetB_{2,n})$ has type $T_3$. Applying Lemma \ref{le: for break cycle} to $F(\EventSetB_{2,n})$ shows that both $B_2$ and $F(\EventSetB_{3,n})$ have type $T_3$. Repeat this process, and one can observe that all $\EventB_i$ have type $T_3$.

\textbf{Case 3:} $B_1$ has type $T_4$. Then $F(\EventSetB_{2,n})$ has type $T_1$ due to Lemma \ref{le: fakj}. Iteratively applying Lemma \ref{le: for break cycle} to $F(\EventSetB_{2,n})$ indicates that all the other $\EventB_i$ have type $T_1$.
\end{proof}

However, the two possibilities are ruled out by the following two lemmas, respectively.
\LongVersionEnd

\ShortVersion
However, the two possibilities are ruled out by the following lemma.

\begin{lemma}\label{le:all4_115areimpossible}
Neither $T_1$-dominant nor $T_3$-dominant combination is possible.% that all $B_i$'s have type $T_3$.
\end{lemma}  

\ShortVersionEnd

\LongVersion
\begin{lemma}\label{le:all4isimpossible}
The $T_3$-dominant combination is impossible.% that all $B_i$'s have type $T_3$.
\end{lemma}  

\begin{proof}
\begin{figure}
\centering
\includegraphics[scale=0.6]{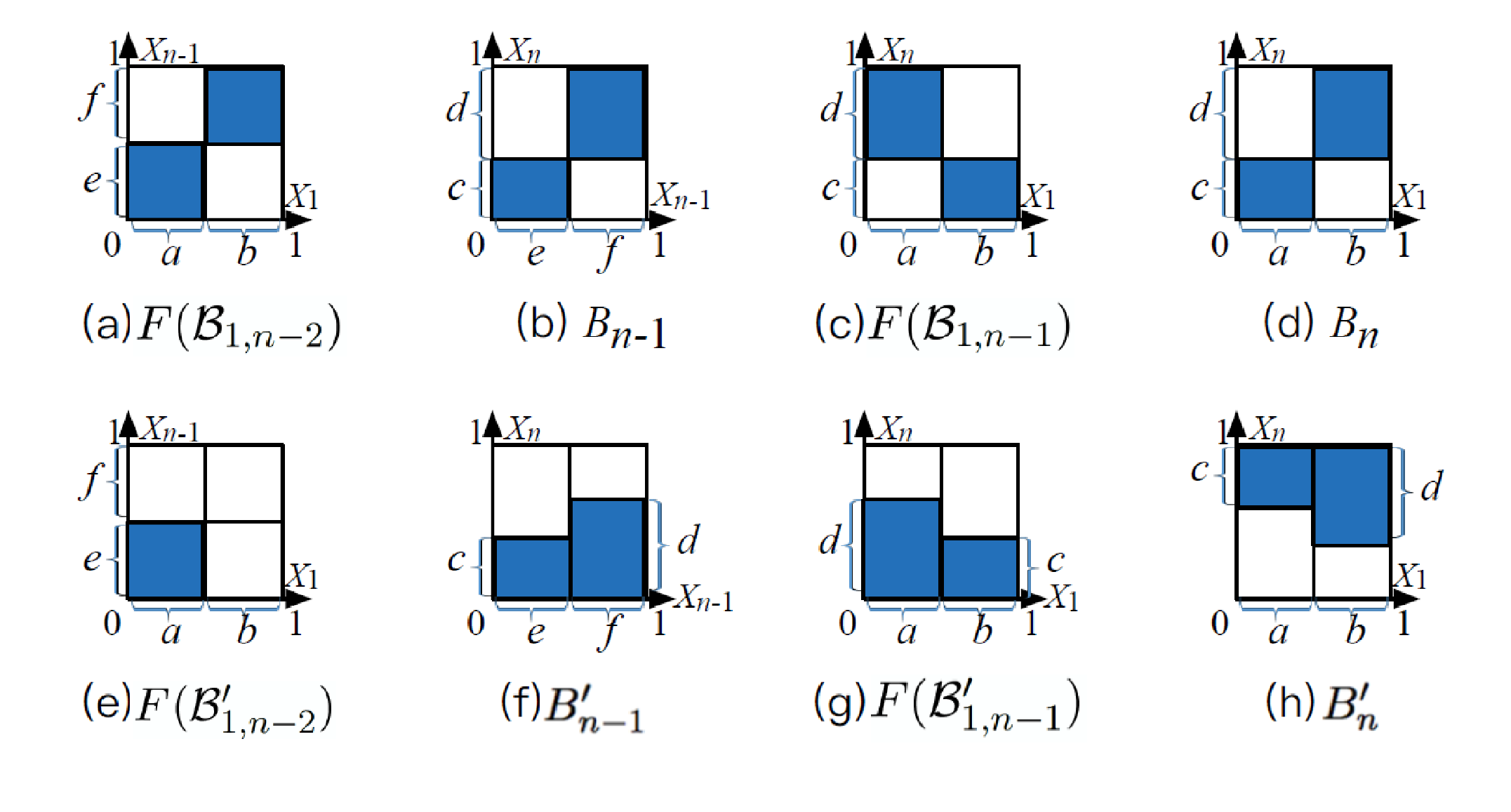}
\caption{The all-$T_3$-case is impossible} \label{fig:all4isimpossible}
\end{figure}
For contradiction, suppose that all $B_i$ have type $T_3$. Then $F(\EventSetB_{1,n-2})$ also has type $T_3$. Without loss of generality, assume that $F(\EventSetB_{1,n-2})$ and $B_{n-1}$ are as illustrated in Figure \ref{fig:all4isimpossible}(a) and \ref{fig:all4isimpossible}(b). Then $\mu(\Event_{n-1})=ce+df$. 
Again without loss of generality, we assume $c\leq d$. 
%We just consider the case $c\leq d$. The other case can be handled symmetrically.

On the one hand, $F(\EventSetB_{1,n-1})$ is as illustrated in Figure \ref{fig:all4isimpossible}(c). By Lemma \ref{le: fakj}, $B_n$ must be as in Figure \ref{fig:all4isimpossible}(d). We have $\mu(A_n)=ac+bd$.

On the other hand, construct $B'_i$ by properly removing one of the two rectangles from $B_i$ for each $1\leq i\leq n-2$ such that $F(\EventSetB'_{1,n-2})$ is as shown in Figure  \ref{fig:all4isimpossible}(e). Let $B'_{n-1}$ be as illustrated in Figure  \ref{fig:all4isimpossible}(f).  One can check that $F(\EventSetB'_{1,n-1})$ must be as illustrated in Figure  \ref{fig:all4isimpossible}(g). Choose $B'_n$ as in Figure  \ref{fig:all4isimpossible}(h). Let $\EventSet'$ is the set of cylinders whose bases are $B'_1,...,B'_{n}$, respectively. We know that $\mu(\Event'_i)<\mu(\Event_i)$ for $1\leq i\leq n-2$, $\mu(\Event'_i)=\mu(\Event_i)$ for $i\in\{n,n-1\}$, $\Pr(\cup_{i\in [n]}\Event'_i)=1$, and $\EventSet'$ conforms with $\BipartiteGraph_n$. Since $\vec{p}$ is a boundary vector, by Lemma \ref{le:boundaryisunique}, we reach a contradiction. 

\end{proof}

%Now the discrete analysis has done its work, and we will use mathematical analysis to rule out the last case. 

\begin{lemma}\label{le:115isimpossible}
The $T_1$-dominant combination is impossible.
\end{lemma}

\begin{proof}
For contradiction, suppose without loss of generality that $\EventB_n$ has type $T_4$ while the others $\EventB_i$ have $T_1$.
% Again without loss of generality,
We can further assume that $B_i$ has type $T_{14}$ as in Figure \ref{fig:typesoftheevents}(a) for every $1\leq i\leq n - 1$. Let $(p_1,...,p_n)=\vec{p}$. We have $a_ib_i=p_i$ for $1\leq i\leq n-1$ and $b_j+a_{j+1}=1$ for $1\leq j\leq n-2$. The parameters $a_i$'s and $b_i$'s should be chosen to maximize the measure of $F(\EventSetB_{1,n-1})$, namely $a_1b_{n-1}$.

Let $x=b_{n-2}\in (0,1)$. Then $b_{n-1}=\frac{p_{n-1}}{1-x},a_{n-2}=f_0(x)\triangleq \frac{p_{n-2}}{x},a_{n-2-k}=f_k(x)\triangleq \frac{p_{n-2-k}}{1-f_{k-1}(x)}, 1\leq k\leq n-3$. Define $g(x)\triangleq a_1b_{n-1}=f_{n-3}(x)\frac{p_{n-1}}{1-x}$.

To maximize $g(x)$, we consider its derivative $\frac{dg(x)}{dx}=f_{n-3}(x)\frac{p_{n-1}}{(1-x)^2}+\frac{df_{n-3}(x)}{dx}\frac{p_{n-1}}{1-x}$. Since only the sign of $\frac{dg(x)}{dx}$ matters, let $h(x)\triangleq \frac{dg(x)}{dx}\frac{(1-x)^2}{p_{n-1}}=f_{n-3}(x)+\frac{df_{n-3}(x)}{dx}(1-x)$. One can check that $\frac{dh(x)}{dx}=\frac{d^2 f_{n-3}(x)}{d^2 x}(1-x)$.
On the other hand, for any $0\leq k\leq n-4$, $\frac{d^2 f_{k+1}(x)}{d^2 x}=\frac{p_{n-3-k} }{(1-f_k(x))^2}\frac{d^2 f_k(x)}{d^2 x}+2\frac{p_{n-3-k} }{(1-f_k(x))^3}\left(\frac{d f_k(x)}{d x}\right)^2$. Since $\frac{d^2 f_0(x)}{d^2 x}>0$, by induction we know that $\frac{d^2 f_k(x)}{d^2 x}> 0$ for any $1\leq k\leq n-3$.

Altogether, $\frac{dh(x)}{dx} > 0$ when $x\in (0,1)$. This implies three possible cases:
\begin{enumerate}
\item $h(x)> 0$ for all $x\in (0,1)$;
\item $h(x)< 0$ for all $x\in (0,1)$;
\item There is $x_0\in (0,1)$ such that $h(x)< 0$ for all $x\in (0,x_0)$ and $h(x)> 0$ for all $x\in (x_0,1)$ and $h(x_0) = 0$ .
\end{enumerate}

Since $h(x)$ and $\frac{dg(x)}{dx}$ always have the same sign, $g(x)$ is either strictly monotonic on $(0,1)$ or decreasing on $(0,x_0)$ and increasing on $(x_0,1)$. 

On the other hand, we show that $x=b_{n-2}$ ranges over a closed interval in $(0,1)$. First, due to $1\geq b_{n-1}=\frac{p_{n-1}}{1-x}$, $x$ increases as $b_{n-1}$ increases, so it reaches upper bound when $\EventB_{n-1}$ has type $T_2$. Second, it is easy to see that $x$ decreases as $a_1$ increases, so $x$ reaches lower bound when $\EventB_1$ has type $T_2$.

Note that $g(x)$ gets maximized either at the lower bound or upper bound of $x$. When $g(x)$ is maximized, either $\EventB_{n-1}$ or $\EventB_1$ has type $T_2$, contradictory to our assumption that the measure of $F(\EventSetB_{1,n-1})$ is maximized in the case $\EventB_n$ has type $T_4$ while the others $\EventB_i$ have $T_1$.
\end{proof}

It is time to prove Theorem \ref{thm: break cycle}.

\begin{proof}
Theorem \ref{thm: break cycle} immediately follows from Lemmas \ref{le: type2enough}, \ref{le: possibletypes}, \ref{le:all4isimpossible}, \ref{le:115isimpossible}.
\end{proof}
\LongVersionEnd

\ShortVersion
By Lemmas \ref{le: type2enough} through to \ref{le:all4_115areimpossible}, Theorem \ref{thm: break cycle} immediately follows, which in turn enables us to find out all boundary vectors of any cyclic bigraph. 
\ShortVersionEnd

%We are ready to find out all boundary vectors of any cycle.

\Boundarycycle*

% \begin{theorem}[Boundary of cycles]\label{cycleboundary}
% Given an arbitrary probability vector $\vec{p}\in(0,1)^n$, for each $i\in [n]$, let $\lambda_i$ be the minimum positive solution to the equation system: $b_1=\lambda p_i, b_k=\frac{\lambda p_{k+i-1}}{1-b_{i-1}}$ for $2\leq k\leq n-1$, $b_{n-1}=1-\lambda p_{i-1}$. Let $\lambda=\min_{i\in[n]}\lambda_i$. Then $\lambda\vec{p}\in\partial_v(\BipartiteGraph_n)$.
% \end{theorem}
\LongVersion
\begin{proof}
Arbitrarily choose a vector $\vec{p}\in(0,1]^n$. By Lemma \ref{le:probabilityboundexist}, there is a unique $\lambda>0$ such that $\lambda\vec{p}\in\partial(\BipartiteGraph_n)$. By Theorem \ref{thm: break cycle}, there is a $\vec{d}$-discrete cylinder set $\EventSet\sim\BipartiteGraph_n$ such that $\mu(\EventSet)=\lambda\vec{p}$, $\mu(\cup_{\Event\in \EventSet} \Event)=1$, and $\vec{d}< (2,2,...,2)$. Then each $\Event_i\in\EventSet$ has a base $\EventB_i\in \Interval^{\{i,i+1\}}$. Arbitrarily choose $i\in [n]$ such that $d_i=1$, which means that both $\EventB_i$ and $\EventB_{i-1}$ have type $T_2$ and $F(\EventSetB_{i-1,i})$ have type $T_4$. More precisely, the type of $\EventB_{i-1}$ is $T_{23}$ or $T_{24}$, and that of $\EventB_i$ is $T_{21}$ or $T_{22}$. By Lemma \ref{le: fakj}, $F(\EventSetB_{i+1,i-2})$ has type $T_1$. By Lemma \ref{le: for break cycle}, we can show that for any $j\notin \{i-1,i\}$, $\EventB_j$ has type $T_1$. The types are illustrated in Figure \ref{fig:typesoftheevents}. Using the notation as in Figure \ref{fig:typesoftheevents}, it is easy to check that $a_{i-1}=\lambda p_{i-1}$, $b_i=\lambda p_i$, $b_j+a_{j+1}=1$ for any $j\neq i-1$, and $b_k a_k=\lambda p_k$ for any $k\notin \{i-1,i\}$. Eliminate all $a$'s, properly re-number the $b$'s, and we get the desired equation. As a result, the unique $\lambda$, say $\lambda_0$, is a solution to that equation.

For any $0<\lambda'<\lambda_0$, let $b'_1=\lambda' p_i, b'_k=\frac{\lambda' p_{k+i-1}}{1-b'_{k-1}}$ for $2\leq k\leq n-1$. By an argument of monotonicity, we know that $0<b'_k< b_k$ for $1\leq k\leq n-1$. On the other hand, if it also holds that $b'_{n-1}=1-\lambda' p_{i-1}$, then $b'_{n-1}>b_{n-1}$, which is a contradiction. Therefore, $\lambda'$ cannot be a solution to the equation system.
Altogether, $\lambda_0$ is the minimum positive solution. The proof ends.
\end{proof}
\LongVersionEnd
%Therefore,  up to a cyclic renumbering of the cylinders, $(T_2,T_2,T_1,.....T_1)$ is the unique possible sequence of types of $B_i$'s. Now we are ready to find out all boundary probabilities of any cycles. Given an arbitrary probability vector $\vec{p}\in(0,1)^n$, by Lemma \ref{le:probabilityboundexist}, there is a unique $\lambda>0$ such that $\lambda\vec{p}\in\partial_v(\BipartiteGraph_n)$. To find out this $\lambda$, for each $1\leq i\leq n$, consider the equation system: $b_1=\lambda p_i, b_k=\frac{\lambda p_{k+i-1 (\textrm{mod }n)}}{1-b_{i-1}}$ for $2\leq k\leq n-1$, $b_{n-1}=1-\lambda p_{i-1 (\textrm{mod }n)}$. The equation system is obtained by letting $B_{i-1}$ and $B_i$ have type $T_2$. Let the solution be $\lambda_i$. Then $\min_{1\leq i\leq n} \lambda_i$ is the desired $\lambda$.

As an application of Theorem \ref{cycleboundary}, we explicitly characterize the boundary of the 3-cyclic bigraph $\BipartiteGraph_3$.

\begin{example}
For $\BipartiteGraph_3$, consider an arbitrary $\vec{p}\in(0,1)^3$ with $p_1+p_2+p_3=1$. For $i\in\{1,2,3\}$, we have $\lambda_i=\frac{1-\sqrt{1-4 p_i p_{i-1}}}{2 p_i p_{i-1}}$. Since the function $\frac{1-\sqrt{1-2x}}{x}$ is increasing with $x>0$, the final $\lambda_0 $ is the $\lambda_i$ with $i$ minimizing $p_i p_{i-1}$. For example, if $p_1\ge p_2$ and $p_1\ge p_3$, then $\lambda_3 \vec{p}=\frac{1-\sqrt{1-4 p_2 p_3}}{2 p_2 p_3}\vec{p}$ is a boundary vector. 
\end{example}

\section{Gap between Abstract- and Variable-LLL}
In this section, we investigate conditions under which Shearer's bound remains tight for Variable-LLL.
\subsection{A Theorem for Gap Decision}

% \begin{definition}[Dependency graph]
% Given a graph $G=([n],E)$ and a set of events $\EventSet=\{\Event_1,...,\Event_n\}$, if for any $i\in[n]$, $\Event_i$ is independent of $\{A_j: j\neq i, (i,j)\notin E\}$, then $G$ is called a dependency graph of $\EventSet$, and we say that $\EventSet$ conforms with $G$, denoted by $\EventSet\sim_a G$.
% \end{definition}

\begin{definition}[Exclusiveness]
An event set $\EventSet$ is said to be \emph{exclusive} with respect to a graph $\DependencyGraph$, if $\DependencyGraph$ is a dependency graph of $\EventSet$ and $\mu(\Event_i\cap \Event_j)=0$ for any $i,j$ such that $i\in \Neighbor_\DependencyGraph(j)$. A cylinder set $\EventSet$ is called exclusive with respect to a bigraph $\BipartiteGraph$, if $\EventSet$ conforms with $\BipartiteGraph$ and $\EventSet$ is exclusive with respect to  $\DependencyGraph_\BipartiteGraph$. We do not mention ``with respect to $\DependencyGraph$ or $\BipartiteGraph$" if it is clear from context. 
\end{definition}

The next lemma claims that exclusive cylinder sets always exist if all the probabilities are small enough. 

\begin{lemma}\label{le:smalleventsdisjoint}
For any bigraph $\BipartiteGraph$, there is $\epsilon>0$ such that for any vector $\vec{p}$ on $(0,\epsilon)$, there exists a cylinder set $\EventSet$ that is exclusive with respect to $\BipartiteGraph$ and $\mu(\EventSet)=\vec{p}$.
\end{lemma}
\LongVersion
\begin{proof}
Let $\BipartiteGraph=([n],[m],E)$. For each $i\in L(\BipartiteGraph)$, define a cylinder $\Event_i=\{(x_1,...,x_m): \frac{i-1}{n}\leq x_j<\frac{i}{n}, \forall j\in \Neighbor(i)\}$. Obviously, $\{\Event_1,...,\Event_n\}$ is exclusive with respect to $\BipartiteGraph$. The lemma holds with $\epsilon=\min_{i\in L(\BipartiteGraph)} \mu(\Event_i)$. 
\end{proof}
\LongVersionEnd

\begin{definition}[Abstract Interior]
The \emph{abstract interior} of a graph $\DependencyGraph=([n],E)$, 
denoted by $\Interior_a(\DependencyGraph)$, is the set $\{\vec{p}\in (0,1)^n: \Pr\left(\cap_{\Event\in \EventSet} \Neg{\Event} \right)>0$  for any event set $\EventSet\sim_a \DependencyGraph$ with $\Pr(\EventSet) = \vec{p}\}$, where ``$\EventSet\sim_a \DependencyGraph$" means that $\DependencyGraph$ is a dependency graph of $\EventSet$. Given a bigraph $\BipartiteGraph$, we simply write $\Interior_a(\BipartiteGraph)$ for $\Interior_a(\DependencyGraph_{\BipartiteGraph})$.
\end{definition} 

It is obvious that $\Interior_a(\BipartiteGraph)\subseteq \Interior(\BipartiteGraph)$ for any bigraph $\BipartiteGraph$. 

\begin{definition}[Abstract Boundary]
The \emph{abstract boundary} of a graph $\DependencyGraph=([n],E)$, denoted by $\partial_a(\DependencyGraph)$, is the set $\{\vec{p}\in (0,1]^n: (1-\epsilon)\vec{p}\in \Interior_a(\DependencyGraph) \textrm{ and } (1+\epsilon)\vec{p}\notin \Interior_a(\DependencyGraph) \textrm{ for any }\epsilon\in (0,1)\}$. Any $\vec{p}\in \partial_a(\DependencyGraph)$ is called an \emph{abstract boundary vector} of $\DependencyGraph$.
\end{definition}

\LongVersion
Here is an interesting property of exclusive event sets.
\LongVersionEnd

\ShortVersion
The proof of \cite[Theorem 1]{shearer1985problem} presents an interesting property of exclusive event sets as in Lemma \ref{le:exclusionisworst}.
\ShortVersionEnd

\begin{lemma}\label{le:exclusionisworst}
Given a graph $\DependencyGraph$ and $\vec{p}\in\mathcal{I}_a(\DependencyGraph)\cup\partial_a(\DependencyGraph)$. Among all event sets $\EventSet\sim_a \DependencyGraph$ with $\Pr(\EventSet)=\vec{p}$, there is an exclusive one such that $\Pr(\cup_{\Event\in \EventSet} \Event )$ is maximized.
\end{lemma}  
\LongVersion
\begin{proof}
It is a byproduct of the proof of \cite[Theorem 1]{shearer1985problem}.
\end{proof} 
\LongVersionEnd

\begin{definition}[Gap]
A bigraph $\BipartiteGraph$ is called \emph{gapful in the direction} of $\vec{p}\in (0,1)^n$, if there is $\lambda>0$ such that $\lambda \vec{p}\in \Interior(\BipartiteGraph)\setminus\Interior_a(\BipartiteGraph)$, otherwise it is called gapless in this direction. $\BipartiteGraph$ is said to be \emph{gapful} if it is gapful in some direction, otherwise it is gapless.
\end{definition} 

For convenience, ``being gapful" will be used interchangeably with ``having a gap".

The main result of this section, namely  Theorem \ref{Conj:GapGeom}, is a necessary and sufficient condition for deciding whether a bigraph is gapful. 
Intuitively, it bridges gaplessness and exclusiveness both in the interior and on the boundary. 
At the first glance, the connection between gaplessness and exclusiveness seems to be an immediate corollary of the well-known Lemma \ref{le:exclusionisworst} by Shearer. However, this is not the case. The main difficulty lies in boundary vectors. Suppose the bigraph is gapless. On the one hand, for a vector on its boundary, there is an exclusive event set whose union has probability $1$, by Lemma \ref{le:exclusionisworst}. These events are not necessarily cylinders, so we cannot claim the existence of an exclusive cylinder set. On the other hand, there indeed is a cylinder set whose union has measure $1$. Such a cylinder set must be exclusive as desired, if the union of non-exclusive events always has smaller probability than that of exclusive ones. But Lemma \ref{le:exclusionisworst} just claims that the union of non-exclusive events cannot have bigger probability, not precluding the possibility that the probabilities are equal. Our proof essentially distills down to ruling out this possibility, as in Lemma \ref{WorstPlacement}.

\LongVersion
The next lemma will be used in proving Lemma \ref{WorstPlacement}. It claims that every individual event in an exclusive event set contributes to the overall probability.
\begin{lemma}\label{thelessthesmaller}
Given an exclusive set $\EventSet$ of events, $\Pr(\cup_{\Event\in \EventSet'} \Event)<\Pr(\cup_{\Event\in \EventSet} \Event)$ for any $\EventSet'\subsetneq \EventSet$.
\end{lemma}
\begin{proof}
The proof is by induction on $|\EventSet|$, the size of $\EventSet$. 

\textbf{Basis:} $|\EventSet|=1$. It trivially holds.

\textbf{Hypothesis:} The lemma whenever $|\EventSet|<n$.

\textbf{Induction:} Consider $|\EventSet|=n$.  Let $\EventSet=\{\Event_1,...,\Event_n\}$, and $\DependencyGraph=([n],E)$ be a dependency graph with respect to which $\EventSet$ is exclusive. For contradiction, suppose that there is
$\EventB\in \EventSet$ such that $\Pr(\cup_{\Event\in \EventSet\setminus\{ \EventB\}} \Event)=\Pr(\cup_{\Event\in \EventSet} \Event)$.  We try to reach a contradiction. Without loss of generality, assume $\EventB=\Event_n$.

Since $\Pr(\cup_{i\in [n-1]} \Event_i)=\Pr(\cup_{i\in [n]} \Event_i)$, we have $\Event_n\subseteq \cup_{i\in [n-1]} \Event_i$. Recall that $\EventSet$ is exclusive, so $\Event_n\cap (\cup_{i\in \Neighbor(n)} \Event_i)=\emptyset$. %, where $\Neighbor(n)$ is the exclusive neighborhood of the vertex $n$ in $G$. 
As a result, $\Event_n\subseteq \cup_{i\notin \Neighbor^+(n)} \Event_i$, where $\Neighbor^+(n)\triangleq \Neighbor(n)\cup\{n\}$. Note that $\Event_n$ and $\{\Event_i:i\notin \Neighbor^+(n)\}$ are independent, so $\Pr(\Event_n)=\Pr(\Event_n\cap \cup_{i\notin \Neighbor^+(n)} \Event_i)=\Pr(\Event_n)\Pr(\cup_{i\notin \Neighbor^+(n)}\Event_i)$, which implies that $\Pr(\cup_{i\notin \Neighbor^+(n)}\Event_i)=1$.

Consider the connected components of $\DependencyGraph$ after $\Neighbor^+(n)$ has been removed. There must be a component $\Gamma$ such that $\Pr(\cup_{i\in \Gamma}\Event_i)=1$, because of two facts. First, event sets on different connected components are independent. Second, if the union of independent events has probability $1$, at least one of them has probability $1$.

Because the vertex $n$ is isolated from $\Gamma$ and $G$ is connected, there must be some vertex $k\in \Neighbor(n)$ that is adjacent to $\Gamma$. Let $\Gamma'=\Gamma\cup \{k\}$, $\DependencyGraph'$ be the induced subgraph of $\DependencyGraph$ on $\Gamma'$, $\vec{p}'=\vec{p}|_{\Gamma'}$, $\EventSet'=\{\Event_i: i\in\Gamma'\}$. Then $\EventSet'$ is exclusive with respect to $\DependencyGraph'$, and $\Pr(\EventSet')=\vec{p}'$. Since $G'$ has less than $n$ vertices, by induction hypothesis, $\Pr(\cup_{i\in \Gamma'} \Event_i)>\Pr(\cup_{i\in \Gamma} \Event_i)=1$ which is a contrdiction.
\end{proof}

The following corollary means that the probability vector of any exclusive cylinder set must lie in the interior or on the boundary. It can be regarded as the converse of Lemma \ref{le:exclusionisworst}. 

\begin{corollary}\label{exclusivenonexterior}
Given a bigraph $\BipartiteGraph$ and a vector $\vec{p}$ on $(0,1]$, if there is a cylinder set $\EventSet$ that is exclusive with respect to $\BipartiteGraph$ and $\mu(\EventSet)=\vec{p}$, then $\vec{p}\in\mathcal{I}(\DependencyGraph)\cup\partial(\DependencyGraph)$.
\end{corollary}
\begin{proof} Assume $\BipartiteGraph=([n],[m],E)$, $\EventSet=\{\Event_1,...,\Event_n\}$, and $\vec{p}=(p_1,...,p_n)$. Consider the vertex $n\in L(\BipartiteGraph)$, and let $B_n$ be the base of $\Event_n$ that lies in $\Interval^{\Neighbor_\BipartiteGraph(n)}$. Arbitrarily fix $0<\epsilon<p_n$. Choose a subset $B'_n\subset B_n$ with $\mu(B'_n)=p_n-\epsilon$. Let $B''_n=B_n\setminus B'_n$. Define a bigraph $\BipartiteGraph'=([n+1],[m],E)$ such that $\Neighbor_{\BipartiteGraph'}(i)=\Neighbor_{\BipartiteGraph}(i)$ for any $i\in L(\BipartiteGraph)\subset L(\BipartiteGraph')$ and $\Neighbor_{\BipartiteGraph'}(n+1)=\Neighbor_{\BipartiteGraph}(n)$ for $n+1\in L(\BipartiteGraph')$. Let $\EventSet'=\{\Event_1,...,\Event_{n-1}, \Event'_n, \Event'_{n+1}\}$ where $\Event'_n$ and $\Event'_{n+1}$ are the cylinders with bases $B'_n$ and $B''_n$, respectively. It is easy to see that $\EventSet'$ is exclusive with respect to $\BipartiteGraph'$ and $\mu(\EventSet')=(p_1,...,p_{n-1},p_n-\epsilon,\epsilon)$. By Lemma \ref{thelessthesmaller}, $\mu(\cup_{\Event\in\EventSet''} \Event)< \mu(\cup_{\Event\in\EventSet'} \Event)\leq 1$, where $\EventSet''=\{\Event_1,...,\Event_{n-1}, \Event'_n\}$. One can check that $\EventSet''$ is exclusive with respect to $\BipartiteGraph$ and $\mu(\EventSet'')=\vec{p}_\epsilon\triangleq (p_1,...,p_{n-1},p_n-\epsilon)$. Hence $\vec{p}_\epsilon\in\Interior(\BipartiteGraph)$. Since $\epsilon$ can be arbitrarily small, we know that $\vec{p}\in\mathcal{I}(\DependencyGraph)\cup\partial(\DependencyGraph)$.
\end{proof}

The following lemma is key to the proof of Theorem \ref{Conj:GapGeom}. Intuitively, it claims that the overall probability is maximized by and only by an exclusive set of event. The ``by'' part was proved in \cite[Theorem 1]{shearer1985problem}, and the ``only by'' part will be proved here. The proof is inspired by that of \cite[Theorem 1]{shearer1985problem}.
\LongVersionEnd

\ShortVersion
The following lemma is key to the proof of Theorem \ref{Conj:GapGeom}. Intuitively, it claims that the overall probability is maximized by and only by an exclusive set of event. Note that the ``by'' part was proved in \cite[Theorem 1]{shearer1985problem}, but it appears here to make this paper self-contained.
\ShortVersionEnd

%The main part (case 2) of the proof is by contradiction. Intuitively, if $\EventSet$ is not exclusive, $\EventSetB$ is exclusive, and $\Pr(\cup_{\Event\in \EventSet} \Event)=\Pr(\cup_{B\in \EventSetB} B)$, then there is an index set $S$ such that  $\EventSet_S,\EventSetB_S$ have three properties ($Q_1,Q_2,Q_3$ in the proof), where $\EventSet_S=\{\Event_i\in \EventSet: i\in S\}$. The properties scale down in the sense that if $\EventSet_S,\EventSetB_S$ have the properties, there is $S'\subsetneq S$ such that $\EventSet_{S'},\EventSetB_{S'}$ also have the properties. Then we get a contradiction since the properties can't hold simultaneously when $S$ is small enough. Some techniques of the proof come from \cite{shearer1985problem}.

\begin{lemma}\label{WorstPlacement}
Suppose that $\DependencyGraph$ is a dependency graph of event sets $\EventSet$ and $\EventSetB$, $\Pr(\EventSet)=\Pr(\EventSetB)$, and $\EventSetB$ is exclusive. Then $\Pr(\cup_{\Event\in \EventSet} \Event)\leq\Pr(\cup_{B\in \EventSetB} B)$, and the equality holds if and only if $\EventSet$ is exclusive.
\end{lemma}
\LongVersion
\begin{proof}
Shearer proved $\Pr(\cup_{A\in \EventSet} A)\leq\Pr(\cup_{B\in \EventSetB} B)$ in \cite[Theorem 1]{shearer1985problem}, so we focus on the other part. 

Assume $\DependencyGraph=([n],E)$, $\EventSet=\{\Event_1,...,\Event_n\}$, and $\Pr(\EventSet)=(p_1,...,p_n)$.
Let's borrow the notation from the proof of \cite[Theorem 1]{shearer1985problem}. %For each $i\in[n]$, let $A_i\in\EventSet$ be the event in $\EventSet$ corresponding to vertex $i$ in $G$, likewise $B_i\in\EventSetB$ is defined. 
For any $S\subseteq[n]$, define $\alpha(S)=\Pr(\cap_{i\in S}\overline{A_i})$ and $\beta(S)=\Pr(\cap_{i\in S}\overline{B_i})$. We proceed case by case.

\textbf{Case 1}: $\beta([n])>0$. Suppose that $\EventSet$ is not exclusive.

We first prove by induction on $|S|$ that $\alpha(S)/\beta(S)$ increases with inclusion. The base case holds since $\alpha(\emptyset)=\beta(\emptyset)$ and $\alpha(S)=\beta(S)$ for any singleton $S$. For induction, given $S_1\subset [n]$ and $j\in[n]\setminus S_1$, let $S_2=S_1\cup\{j\}$, $T_2=S_1\cap \Neighbor(j)$, and $T_1=S_1\setminus T_2$. We have 
\begin{eqnarray}\label{equ:quotientmonotonic}
\frac{\alpha(S_2)}{\beta(S_2)}- \frac{\alpha(S_1)}{\beta(S_1)}\geq \frac{\alpha(S_1)-p_{j}\alpha(T_1)}{\beta(S_1)-p_{j}\beta(T_1)}- \frac{\alpha(S_1)}{\beta(S_1)}=\frac{p_{j}\beta(T_1)}{\beta(S_1)-p_{j}\beta(T_1)}\left[\frac{\alpha(S_1)}{\beta(S_1)}-\frac{\alpha(T_1)}{\beta(T_1)}\right]\geq 0. 
\end{eqnarray}
%$\frac{\alpha(S'_2)}{\beta(S'_2)}- \frac{\alpha(S'_1)}{\beta(S'_1)}\geq \frac{\alpha(S'_1)-p_{j'}\alpha(T'_1)}{\beta(S'_1)-p_{j'}\beta(T'_1)}- \frac{\alpha(S'_1)}{\beta(S'_1)}=\frac{p_{j'}\beta(T'_1)}{\beta(S'_1)-p_{j'}\beta(T'_1)}\left[\frac{\alpha(S'_1)}{\beta(S'_1)}-\frac{\alpha(T'_1)}{\beta(T'_1)}\right]\geq 0$. 

The last inequality is by induction, and the first one holds because on the one hand
\begin{eqnarray}\label{equ:S2}
\begin{array}{rl}
    \alpha(S_2)=&\Pr(\cap_{i\in S_2}\overline{A_i})= \Pr(\cap_{i\in S_1}\overline{A_i})-\Pr(\cap_{i\in S_1}\overline{A_i}\cap A_j)\\
    = &  \Pr(\cap_{i\in S_1}\overline{A_i})-\Pr(\cap_{i\in T_1}\overline{A_i}\cap A_j)+\Pr(\cap_{i\in T_1}\overline{A_i}\cap A_j\cap(\cup_{i\in T_2}A_i))\\
    \geq&  \Pr(\cap_{i\in S_1}\overline{A_i})-\Pr(\cap_{i\in T_1}\overline{A_i}\cap A_j)=\alpha(S_1) - p_{j}\alpha(T_1),
\end{array}
\end{eqnarray}
and on the other hand, $\beta(S_2)=\beta(S_1)-p_{j}\beta(T_1)$ due to a similar process like formula \ref{equ:S2} and the assumption that $\EventSetB$ is exclusive.  Hence, $\alpha(S)/\beta(S)$ is increasing.

As a special case, choose $i,j\in[n]$ such that $j\in \Neighbor(i)$ and $\Pr(\Event_i\cap\Event_j)>0$. Such a pair of $i,j$ exists because of the assumption that $\EventSet$ is not exclusive. Apply (\ref{equ:quotientmonotonic}) to $S_1=\{i\}, S_2=\{i,j\}$, $T_2=\{i\}, T_1=\emptyset$. Since the inequality in (\ref{equ:S2}) turns out to be ``$>$", the first inequality in (\ref{equ:quotientmonotonic}) is also ``$>$". Thus $\frac{\alpha(S_2)}{\beta(S_2)}>\frac{\alpha(S_1)}{\beta(S_1)}=1$, which, together with the monotonicity of $\alpha(S)/\beta(S)$, implies that $\frac{\alpha([n])}{\beta([n])}>1$. As a result, $\alpha([n])>\beta([n])$.

%Note that (\ref{equ:quotientmonotonic}) implies the following monotonicity: 
%\begin{eqnarray}\label{equ:increasing}
%\textrm{For any }U\subseteq V\subseteq [n]\textrm{ with }\beta(V)>0, \frac{\alpha(V)}{\beta(V)}\geq\frac{\alpha(U)}{\beta(U)}.
%\end{eqnarray}

%By monotonicity of $\alpha(S)/\beta(S)$ and $\frac{\alpha(S_2)}{\beta(S_2)}>1$, we have $\frac{\alpha([n])}{\beta([n])}>1$, namely, $\alpha([n])>\beta([n])$.

\textbf{Case 2}: $\beta([n])=0$.
%Shearer proved $1\leq\frac{\alpha(S)}{\beta(S)}\leq \frac{\alpha(S')}{\beta(S')}$ for any $S\subset S'$ with $\beta(S')>0$. And if $\beta([n])>0$, that proof also leads to our lemma, up to a minor revision. As a result, it is enough to investigate the case $\beta([n])=0$, and we assume that $\Pr(\cup_{A\in \EventSet} A)=1$. 
Assume that $\Pr(\cup_{A\in \EventSet} A)=\Pr(\cup_{B\in \EventSetB} B)$ while $\EventSet$ is NOT exclusive. We try to reach a contradiction.

Let $S_2 =[n]$. Since $\EventSet$ is NOT exclusive, there is $j\in[n]$ such that $\Pr(A_j\cap(\cup_{i\in \Neighbor(j)}A_i))>0$. Let $S_1 =S_2 \setminus \{j\}, T_1= S_1\setminus \Neighbor(j), T_2=S_1\setminus T_1=\Neighbor(j)$. The property $Q_1$ holds immediately:

\quad $Q_1$: $T_2\neq\emptyset$ and $\Pr(A_j\cap(\cup_{i\in T_2}A_i))>0$.
%\begin{description}
%\item[$Q_1$: ] $T_2\neq\emptyset$ and $\Pr(A_j\cap(\cup_{i\in T_2}A_i))>0$.
%\end{description}

Then note that 
\begin{eqnarray}\label{equ:allareequal}
\begin{array}{rl}
    0=\alpha(S_2)\geq & \alpha(S_1) - p_{j}\alpha(T_1)=  \frac{\alpha(S_1)}{\beta(S_1)}(\beta(S_1) - p_{j}\alpha(T_1) \frac{\beta(S_1)}{\alpha(S_1)}) \\
    \geq &\frac{\alpha(S_1)}{\beta(S_1)}(\beta(S_1) - p_{j}\beta(T_1))= \frac{\alpha(S_1)}{\beta(S_1)}\beta(S_2)=0
\end{array}
\end{eqnarray}
where $\beta(S_1)>0$ due to Lemma \ref{thelessthesmaller}. 
The first inequality in (\ref{equ:allareequal}) is due to (\ref{equ:S2}). 
%\begin{eqnarray}\label{equ:S2}
%\begin{array}{rl}
%    \alpha(S_2)=&\Pr(\cap_{i\in S_2}\overline{A_i})= \Pr(\cap_{i\in S_1}\overline{A_i})-\Pr(\cap_{i\in S_1}\overline{A_i}\cap A_j)\\
%    = &  \Pr(\cap_{i\in S_1}\overline{A_i})-\Pr(\cap_{i\in T_1}\overline{A_i}\cap A_j)+\Pr(\cap_{i\in T_1}\overline{A_i}\cap A_j\cap(\cup_{i\in T_2}A_i))\\
%    \geq&  \Pr(\cap_{i\in S_1}\overline{A_i})-\Pr(\cap_{i\in T_1}\overline{A_i}\cap A_j)=\alpha(S_1) - p_{j}\alpha(T_1)
%\end{array}
%\end{eqnarray}
The second follows from $\frac{\alpha(S_1)}{\beta(S_1)}\geq \frac{\alpha(T_1)}{\beta(T_1)}$, by the monotonicity of $\frac{\alpha(S)}{\beta(S)}$.
Since both inequalities turns out to be equal, we get the properties $Q_2, Q_3$:

\quad $Q_2$: $\frac{\alpha(S_1)}{\beta(S_1)} = \frac{\alpha(T_1)}{\beta(T_1)}.$

\quad $Q_3$: $\Pr(\cap_{i\in T_1}\overline{A_i}\cap A_j\cap(\cup_{i\in T_2}A_i))=0$.
%\begin{description}
%\item[$Q_2$: ] $\frac{\alpha(S_1)}{\beta(S_1)} = \frac{\alpha(T_1)}{\beta(T_1)}.$
%\item[$Q_3$: ] $\Pr(\cap_{i\in T_1}\overline{A_i}\cap A_j\cap(\cup_{i\in T_2}A_i))=0$.
%\end{description}

Consequently, the proof is reduced to proving the following claim.

\textbf{Claim:} For any $S_2 \subseteq[n]$ and $j\in S_2$, let $S_1 =S_2 \setminus \{j\}, T_1= S_1\setminus \Neighbor(j), T_2=S_1\setminus T_1$. It is impossible that the properties $Q_1,Q_2,Q_3$ hold simultaneously.

\textbf{Proof of the Claim}: The proof is by induction on the size of $T_1$. 

\textbf{Basis:} $T_1=\emptyset$. By $Q_3$, $0=\Pr(\cap_{i\in T_1}\overline{A_i}\cap A_j\cap(\cup_{i\in T_2}A_i))=\Pr(A_j\cap(\cup_{i\in T_2}A_i))$, which is contradictory to $Q_1$.

\textbf{Hypothesis:} The claim holds if $|T_1|<t$.

\textbf{Induction:} Consider the case where $|T_1|=t$. Assume for contradiction that $Q_1,Q_2,Q_3$ hold simultaneously.

By $Q_1$, one can choose $j'\in T_2$ such that $\Pr(A_j\cap A_{j'})>0$.

We first show that $T_1\cap \Neighbor(j')\neq \emptyset$. This is because if $T_1\cap \Neighbor(j')= \emptyset$, then

\begin{eqnarray}\label{equ:contradiction}
\begin{array}{rcl}
0&\stackrel {\textrm{by }Q_3}{=}&  \Pr(\cap_{i\in T_1}\overline{A_i}\cap A_j\cap(\cup_{i\in T_2}A_i))\geq \Pr(\cap_{i\in T_1}\overline{A_i}\cap A_j\cap A_{j'})\\
&\stackrel {\textrm{by }T_1\cap \Neighbor(j')= \emptyset}{=}&  \Pr(\cap_{i\in T_1}\overline{A_i})\Pr(A_j\cap A_{j'})>0   \quad \textrm{A contradiction!}
\end{array}
\end{eqnarray}

A byproduct of (\ref{equ:contradiction}) is that $Q_3$ implies
\begin{eqnarray}\label{byproduct}
\Pr(\cap_{i\in T_1}\overline{A_i}\cap A_j\cap A_{j'})=0.
\end{eqnarray}

%The last inequality follows from the choice of $j'$ and $$\Pr(\cap_{i\in T_1}\overline{A_i})=1-\Pr(\cup_{i\in T_1}A_i)\stackrel{\cite[\textrm{Theorem 1}]{shearer1985problem}}{\geq} 1-\Pr(\cup_{i\in T_1}B_i)\stackrel{\textrm{by Lemma }\ref{thelessthesmaller}}{>}0.$$
Then we prove that $\Pr(A_{j'}\cap(\cup_{i\in T_1\cap \Neighbor(j')}A_i))>0$. This is due to 
\begin{eqnarray*}
\begin{array}{rcl}
&&\Pr(A_{j'}\cap(\cup_{i\in T_1\cap \Neighbor(j')}A_i))\\
&\geq &\Pr(A_{j'}\cap(\cup_{i\in T_1\cap \Neighbor(j')}A_i)\cap A_j\cap_{k\in T_1\setminus \Neighbor(j')}\overline{A_k})\\
&\stackrel{\textrm{by } (\ref{byproduct})}{=}&\Pr(A_{j'}\cap(\cup_{i\in T_1\cap \Neighbor(j')}A_i)\cap A_j\cap_{k\in T_1\setminus \Neighbor(j')}\overline{A_k})\\
&&+\Pr(\cap_{i\in T_1}\overline{A_i}\cap A_j\cap A_{j'}) \\
&=&\Pr(A_j\cap A_{j'}\cap_{k\in T_1\setminus \Neighbor(j')}\overline{A_k})\\
&=&\Pr(A_j\cap A_{j'})\Pr(\cap_{k\in T_1\setminus \Neighbor(j')}\overline{A_k})>0
\end{array}
\end{eqnarray*}
%
%$\Pr(X_{j'}=1, \exists l\in T_1\cap N(j'), X_l=1)\geq \Pr(\forall k\in T_1\setminus N(j'), X_k=0, \exists l\in T_1\cap N(j'), X_l=1, X_{j'}=1, X_{j}=1)=\Pr(\forall k\in T_1\setminus N(j'), X_k=0, \exists l\in T_1\cap N(j'), X_l=1, X_{j'}=1, X_{j}=1)+\Pr(\forall k\in T_1, X_k=0, X_{j'}=1, X_{j}=1)=\Pr(\forall k\in T_1\setminus N(j'), X_k=0, X_{j'}=1, X_{j}=1)=\Pr(\forall k\in T_1\setminus N(j'), X_k=0)\Pr(X_{j'}=1, X_{j}=1)>0$.
%%Condition 1 is $\frac{\alpha(S^{(0)}_1)}{\beta(S^{(0)}_1)} = \frac{\alpha(T^{(0)}_1)}{\beta(T^{(0)}_1)}.$ And condition 2 is $Pr(X$ is $0$ on $S^{(0)}_1$ and $1$ on $j)$ = $Pr(X$ is $0$ on $T^{(0)}_1$ and $1$ on $j)$.
%
%
%Since $\Pr(\forall k\in T_1, X_k=0, X_{j}=1)=\Pr(\forall k\in S_1, X_k=0,X_{j}=1)+\Pr(\forall k\in T_1, X_k=0,\exists l\in T_2, X_l=1, X_{j}=1)$, $Q_3$ implies that $\Pr(\forall k\in T_1, X_k=0,\exists l\in T_2, X_l=1, X_{j}=1)=0$. Considering that  $\Pr(\forall k\in T_1, X_k=0,\exists l\in T_2, X_l=1, X_{j}=1)\geq \Pr(\forall k\in T_1, X_k=0, X_{j'}=1, X_{j}=1)$, we have  $\Pr(\forall k\in T_1, X_k=0, X_{j'}=1, X_{j}=1)=0$.
%%By the choice of ${j^{(0)}}$, there is $l\in T^{(0)}_2$ such that $\Pr(X_{j^{(0)}}=1,X_l=1)>0$. Arbitrarily fix such an $l$ and denote it by $j^{(1)}$.
%
%Now we show that in the dependency graph $G$, $T_1\cap N(j')\neq \emptyset$. Suppose $T_1\cap N(j')= \emptyset$. Then we have $0=\Pr(\forall k\in T_1, X_k=0, X_{j'}=1, X_{j}=1)=\Pr(\forall k\in T_1,X_k=0)\Pr(X_{j'}=1, X_{j}=1)>0$, which is a contradiction. Hence $T_1\cap N(j')\neq \emptyset$.

Now let $S'_2 \triangleq T_1\cup\{j'\}, S'_1 \triangleq T_1, T'_1\triangleq S'_1\setminus \Neighbor(j'), T'_2\triangleq S_1'\setminus T'_1= T_1\cap \Neighbor(j')$. We have shown 

\quad $Q'_1$: $T'_2\neq\emptyset$ and $\Pr(A_{j'}\cap(\cup_{i\in T'_2}A_i))>0$.
%\begin{description}
%\item[$Q'_1$: ] $T'_2\neq\emptyset$ and $\Pr(A_{j'}\cap(\cup_{i\in T'_2}A_i))>0$.
%\end{description}

%$\frac{\alpha(T_1)}{\beta(T_1)}\stackrel {\textrm{by }Q_2}{=}\frac{\alpha(S_1)}{\beta(S_1)}\stackrel {\cite[\textrm{Theorem 1}]{shearer1985problem}}{\geq}\frac{\alpha(S'_2)}{\beta(S'_2)}\stackrel {\cite[\textrm{Theorem 1}]{shearer1985problem}}{\geq} \frac{\alpha(S'_1)}{\beta(S'_1)}=\frac{\alpha(T_1)}{\beta(T_1)}.$ Hence, $\frac{\alpha(S'_2)}{\beta(S'_2)}= \frac{\alpha(S'_1)}{\beta(S'_1)}$.
Now we show other properties. 

On the one hand, $\frac{\alpha(S'_2)}{\beta(S'_2)}- \frac{\alpha(S'_1)}{\beta(S'_1)}=0$ due to four facts: $S'_1=T_1$, $S'_2\subseteq S_1$, $\frac{\alpha(S)}{\beta(S)}$ is monotone, and $Q_2$ holds.

One the other hand,  $\frac{\alpha(S'_2)}{\beta(S'_2)}- \frac{\alpha(S'_1)}{\beta(S'_1)}\geq \frac{\alpha(S'_1)-p_{j'}\alpha(T'_1)}{\beta(S'_1)-p_{j'}\beta(T'_1)}- \frac{\alpha(S'_1)}{\beta(S'_1)}=\frac{p_{j'}\beta(T'_1)}{\beta(S'_1)-p_{j'}\beta(T'_1)}\left[\frac{\alpha(S'_1)}{\beta(S'_1)}-\frac{\alpha(T'_1)}{\beta(T'_1)}\right]\geq 0$ by formula (\ref{equ:quotientmonotonic}). Since both inequalities should be equality, we have $\alpha(S'_2)=\alpha(S'_1)-p_{j'}\alpha(T'_1)$ and 

\quad $Q'_2$: $\frac{\alpha(S'_1)}{\beta(S'_1)} = \frac{\alpha(T'_1)}{\beta(T'_1)}.$

%\begin{description}
%\item[$Q'_2$: ] $\frac{\alpha(S'_1)}{\beta(S'_1)} = \frac{\alpha(T'_1)}{\beta(T'_1)}.$
%\end{description}
%
Applying formula (\ref{equ:S2}) to $S'_2$, the equality $\alpha(S'_2)=\alpha(S'_1)-p_{j'}\alpha(T'_1)$ implies  

\quad $Q'_3$: $\Pr(\cap_{i\in T'_1}\overline{A_i}\cap A_{j'}\cap(\cup_{i\in T'_2}A_i))=0$.
%\begin{description}
%\item[$Q'_3$: ] $\Pr(\cap_{i\in T'_1}\overline{A_i}\cap A_{j'}\cap(\cup_{i\in T'_2}A_i))=0$.
%\end{description}

Altogether, the properties $Q'_1,Q'_2,Q'_3$ remains true for $S'_2,j',S'_1,T'_1,T'_2$. 

However, since $|T'_1|<|T_1|=t$, by the induction hypothesis, the properties $Q'_1,Q'_2,Q'_3$ can't holds simultaneously. We reach a contradiction. \textbf{The Claim is proven}. 
\end{proof}
\LongVersionEnd

%\begin{eqnarray*}
%\begin{array}{rcl}
%\frac{\alpha(T_1)}{\beta(T_1)}&\stackrel {\textrm{by }Q_2}{=}&\frac{\alpha(S_1)}{\beta(S_1)}\\
%&=&\frac{\alpha(S'_2)}{\beta(S'_2)}\geq \frac{\alpha(S'_1)}{\beta(S'_1)}\\
%&\stackrel {\cite[\textrm{Theorem 1}]{shearer1985problem}}{\geq}&\frac{\alpha(T_1)}{\beta(T_1)}
%\end{array}
%\end{eqnarray*} 

%In a similar way as the above proof, we have $\Pr(X_{j^{(2)}}=1, \exists l\in T'_1\cap N(j^{(2)}), X_l=1)>0$. With this fact and the three properties $Q^{(1)}_1,Q^{(1)}_2,Q^{(1)}_3$, we can iterate the process until either $T_1^{(t)}= \emptyset$ or $T_2^{(t)}= \emptyset$ at some $t$. However, $T_2^{(t)}= \emptyset$ is impossible due to the property $Q^{(t)}_1$. So, $T_1^{(t)}= \emptyset$.
%
%Then, on the one hand, $Q^{(t)}_3$ implies that $\Pr(\forall k\in T^{(t)}_1, X_k=0,\exists l\in T^{(t)}_2, X_l=1, X_{j^{(t)}}=1)=0$. On the other hand, by $T_1^{(t)}= \emptyset$, $\Pr(\forall k\in T^{(t)}_1, X_k=0,\exists l\in T^{(t)}_2, X_l=1, X_{j^{(t)}}=1)=\Pr(\exists l\in T^{(t)}_2, X_l=1, X_{j^{(t)}}=1)>0$ which follows from $Q^{(t)}_1$. A contradiction is reached.

%The following is a technical lemma which claims that under some conditions, ex 

\LongVersion
Intuitively, the next lemma shows that any set of events can be reduced proportionally so that the dependency graph and exclusiveness are preserved and the probability of the union decreases at most linearly. Basically, in order to reduce an event $\Event$, construct cylinders with height $1$ whose bases are the events, respective. Then adjust the height of $\Event$-based cylinder to $\lambda$. Regard the cylinders as new events and repeat this process until each original event has been handled.
\begin{lemma}\label{DGboundaryregular}
Given a graph $\DependencyGraph=([n],E)$ and a vector $\vec{p}\in(0,1]^n$, suppose that event set $\EventSet\sim_a \DependencyGraph$ and $\Pr(\EventSet)=\vec{p}$. For any $\lambda\in (0,1)$, there is an event set $\EventSetB\sim_a \DependencyGraph$ with $\Pr(\EventSetB)=\lambda\vec{p}$ such that 
\begin{enumerate}
\item If $\EventSet$ is exclusive, so is $\EventSetB$;
\item $\Pr(\cup_{A\in\EventSet}A)-(1-\lambda)\sum_{i\in[n]}p_i\leq\Pr(\cup_{B\in\EventSetB}B)\leq \Pr(\cup_{A\in\EventSet}A)$.
\end{enumerate}

\end{lemma}

\begin{proof}
Assume $\EventSet=\{A_1,...,A_n\}$. Let $\mathcal{S}^{(0)}$ be the probability space from which the events in $\EventSet$ come. Define probability space $\mathcal{S}^{(1)}=\mathcal{S}^{(0)}\times \Interval$ where $\Interval$ is the unit interval $[0,1]$ endowed with Lebesgue measure. Let $\EventSet^{(1)}$ be the set of events in $\mathcal{S}^{(1)}$ defined as $A^{(1)}_1=A_1\times [0,\lambda]$ and $A^{(1)}_k=A_k\times \Interval$ for $k\neq 1$. Let $\vec{p}^{(1)}=(\lambda p_1,p_2,...,p_n)$. It is easy to see that $\EventSet^{(1)}\sim_a\DependencyGraph$, $\Pr(\EventSet^{(1)})=\vec{p}^{(1)}$, and $\Pr(\cup_{i\in[n]} A^{(1)}_i)\geq \Pr(\cup_{i\in[n]}A_i)-(1-\lambda)  p_1$. 

Likewise, we define probability space $\mathcal{S}^{(2)}=\mathcal{S}^{(1)}\times \Interval$, and event set $\EventSet^{(2)}$ in $\mathcal{S}^{(2)}$ with $A^{(2)}_2=A^{(1)}_2\times [0,\lambda]$ and $A^{(2)}_k=A^{(1)}_k\times \Interval$ for $k\neq 2$. Let $\vec{p}^{(2)}=(\lambda p_1,\lambda p_2,p_3,...,p_n)$. We have that $\EventSet^{(2)}\sim_a\DependencyGraph$, $\Pr(\EventSet^{(2)})=\vec{p}^{(2)}$, and $\Pr(\cup_{i\in[n]} A^{(2)}_i)\geq \Pr(\cup_{i\in[n]}A_i)-(1-\lambda)  (p_1+p_2)$.

Iterate until we get $\EventSet^{(n)}=(A^{(n)}_1,...,A^{(n)}_n)\sim_a\DependencyGraph$, $\Pr(\EventSet^{(n)})=\vec{p}^{(n)}=\lambda\vec{p}$, and $\Pr(\cup_{i\in[n]} A^{(n)}_i)\geq \Pr(\cup_{i\in[n]}A_i)-(1-\lambda)\sum_{i\in[n]}p_i$.

One can check that 
\begin{enumerate}
\item If $\EventSet$ is exclusive, so is $\EventSet^{(i)}$ for any $i\in[n]$;
\item $\Pr(\cup_{A\in\EventSet^{(i)}}A)\leq \Pr(\cup_{A\in\EventSet}A)$ for any $i\in[n]$.
\end{enumerate}

Let $\EventSetB=\EventSet^{(n)}$. The proof ends.
\end{proof}

Now we are ready to present a counterpart of Lemma \ref{le:probabilityboundexist}.
\begin{lemma}\label{le:abstractprobabilityboundexist}
For any graph $\DependencyGraph=([n],E)$ and $\vec{p}\in (0,1)^n$, there is a unique $\lambda>0$ such that $\lambda\vec{p}\in\partial_a(\DependencyGraph)$.
\end{lemma}     
\begin{proof}
Arbitrarily fix a graph $\DependencyGraph=([n],E)$ and $\vec{p}\in (0,1)^n$. Let $\Lambda\triangleq\{\lambda>0:\lambda\vec{p}\notin\Interior_a(\DependencyGraph)\}$.

It is easy to see that  

\begin{enumerate}
\item If $\lambda$ is so big that an entry of $\lambda\vec{p}$ equals $1$,  $\Pr(\cup_{\Event\in \EventSet} \Event )=1$ for any event set $\EventSet\sim_a \DependencyGraph$ such that $\Pr(\EventSet)=\phi(\lambda\vec{p})$.
\item If $\lambda$ is so small that $l_1$-norm of $\lambda\vec{p}$ is smaller than $1$,  $\Pr(\cup_{\Event\in \EventSet} \Event )<1$ for any event set $\EventSet\sim_a \DependencyGraph$ such that $\Pr(\EventSet)=\lambda\vec{p}$.
\end{enumerate}
Thus, $\Lambda$ is non-empty and its infimum, denoted by $\lambda_0$, must be positive. Let $\vec{q}=\lambda_0\vec{p}$. In order to show that $\vec{q}\in \partial_a(\DependencyGraph)$, consider an arbitrary real number $\epsilon>0$.

On the one hand, because $\lambda_0=\inf \Lambda$, we have $(1-\epsilon)\vec{q}\in \Interior_a(\DependencyGraph)$. 

On the other hand, assume for contradiction that $(1+\epsilon)\vec{q}\in \Interior_a(\DependencyGraph)$. By Lemma \ref{le:exclusionisworst}, we can choose an exclusive event set $\EventSet\sim_a \DependencyGraph$ such that $\Pr(\EventSet)=(1+\epsilon)\vec{q}$ and $\Pr(\cup_{\Event\in \EventSet} \Event )<1$. By Lemma \ref{DGboundaryregular}, for any $0<\delta<1$, there is an exclusive event set $\EventSet_\delta\sim_a \DependencyGraph$ such that $\Pr(\EventSet_\delta)=\delta(1+\epsilon)\vec{q}$ and $\Pr(\cup_{\Event\in \EventSet_\delta} \Event )<\Pr(\cup_{\Event\in \EventSet} \Event )<1$. 
By Lemma \ref{WorstPlacement}, $\Pr(\cup_{\Event\in \EventSet} \Event )<1$ for any event set $\EventSet\sim_a \DependencyGraph$ with $\Pr(\EventSet)=\delta(1+\epsilon)\vec{q}$, so $\delta(1+\epsilon)\vec{q}\in \Interior_a(\DependencyGraph)$, which means $\delta(1+\epsilon)\lambda_0\notin \Lambda$.
Since $\delta$ ranges over $(0,1)$, we have $(0,(1+\epsilon)\lambda_0)\cap\Lambda=\emptyset$, contradictory to the fact that $\lambda_0=\inf \Lambda$. As a result, $(1+\epsilon)\vec{q}\notin \Interior_a(\DependencyGraph)$.

Altogether, $\lambda_0\vec{p}\in\partial_a(\DependencyGraph)$. The uniqueness immediately follows from the definition of abstract boundary vectors.
%
%
%$\Pr(\cap_{A\in \EventSet} \bar{A} )$ is minimized by any exclusive abstract event system $(\EventSet, G,(1+\epsilon)\vec{q})$. Arbitrarily fix an exclusive abstract event system $(\EventSet, G,(1+\epsilon)\vec{q})$. By Lemma \ref{DGboundaryregular}, for any $0<\delta<1$, there is an exclusive event set $\EventSet_\delta$ $(\EventSet_\delta, G,\delta(1+\epsilon)\vec{q})$ such that $\Pr(\cap_{A\in \EventSet_\delta} \bar{A} )>\Pr(\cap_{A\in \EventSet} \bar{A} )>0$. Again by Lemma \ref{le:exclusionisworst}, $\Pr(\cap_{A\in \EventSet_\delta} \bar{A} )>0$ for any abstract event system $(\EventSet, G,\delta(1+\epsilon)\vec{q})$. This means that $\delta(1+\epsilon)\lambda_0\notin \Lambda$ for any $\delta>0$, contradictory to the fact that $\lambda_0=\inf \Lambda$.
%
%As a result, for any real number $\epsilon>0$. there is an abstract event system $(\EventSet, G,(1+\epsilon)\vec{q})$ such that $\Pr(\cap_{A\in \EventSet} \bar{A} )=0$, meaning that $\vec{q}=\lambda_0\vec{p}$ is an abstract boundary probability vector of $G$.                
\end{proof}
\LongVersionEnd

Now we are ready to prove the main theorem of this section.
\bipartite*

% \begin{theorem}\label{Conj:GapGeom}
% Given a bigraph $\BipartiteGraph=([n],[m],E)$ and a vector $\vec{p}\in (0,1]^n$, the following three conditions are equivalent:
% \begin{enumerate}
% \item For any $\lambda$ such that $\lambda\vec{p}\in\mathcal{I}_v(\BipartiteGraph)$, there is an exclusive cylinder set $\EventSet\sim \BipartiteGraph$ with $\mu(\EventSet)=\lambda\vec{p}$. \label{interiorex}
% \item For the $\lambda$ such that $\lambda\vec{p}\in\partial_v(\BipartiteGraph)$, there is an exclusive cylinder set $\EventSet\sim \BipartiteGraph$ with $\mu(\EventSet)=\lambda\vec{p}$. \label{boundaryex}
% \item $\BipartiteGraph$ is gapless in the direction of $\vec{p}$. \label{gapless}
% \end{enumerate}
% \end{theorem}

\begin{proof}
($\ref{interiorex}\Rightarrow\ref{gapless}$): Arbitrarily fix $\lambda>0$ such that $\vec{q}\triangleq\lambda\vec{p}\in\mathcal{I}(\BipartiteGraph)$. Let $\EventSet\sim \BipartiteGraph$ be an exclusive cylinder set such that $\mu(\EventSet)=\vec{q}$ and $\mu(\cup_{\Event\in\EventSet}\Event)<1$. It also holds that $\EventSet$ is exclusive with respect to the base graph $G_\BipartiteGraph$. Since $\mu(\cup_{\Event\in\EventSet}\Event)<1$, by Lemma \ref{WorstPlacement}, $\mu(\cup_{\EventB\in\EventSetB}\EventB)<1$ for any event set $\EventSetB\sim_a G_\BipartiteGraph$ with $\Pr(\EventSetB)=\vec{q}$. As a result, $\vec{q}\in\mathcal{I}_a(\BipartiteGraph)$. Altogether, $\BipartiteGraph$ is gapless in the direction of $\vec{p}$.

($\ref{gapless}\Rightarrow\ref{boundaryex}$): Assume that $\BipartiteGraph$ is gapless in the direction of $\vec{p}$. Let $\lambda$ be such that $\vec{q}\triangleq\lambda\vec{p}\in\partial(\BipartiteGraph)$. By Theorem \ref{thm:boundaryfills}, there is a cylinder set $\EventSet\sim\BipartiteGraph$ such that $ \mu(\EventSet)=\vec{q}$ and $\mu(\cup_{\Event\in \EventSet} \Event)=1$. On the other hand, $\vec{q}\in\partial_a(\BipartiteGraph)$ due to the assumption that $\BipartiteGraph$ is gapless in the direction of $\vec{p}$. By Lemma \ref{le:exclusionisworst}, there is an exclusive event set $\EventSetB\sim_a G_\BipartiteGraph$ such that $ \mu(\EventSetB)=\vec{q}$ and $\Pr(\cup_{\EventB\in \EventSetB} \EventB)=1$. Because $\EventSet$ also conforms with $G_\BipartiteGraph$ and $\Pr(\cup_{\EventB\in \EventSetB} \EventB)=\Pr(\cup_{\Event\in \EventSet} \Event)=1$, by Lemma \ref{WorstPlacement}, $\EventSet$ must be exclusive with respect to $G_\BipartiteGraph$, hence exclusive with respect to $\BipartiteGraph$. 

($\ref{boundaryex}\Rightarrow\ref{interiorex}$): Arbitrarily fix $\lambda>0$ such that $\vec{q}\triangleq\lambda\vec{p}\in\mathcal{I}(\BipartiteGraph)$. Let $\delta>1$ be such that $\delta\lambda\vec{p}\in\partial(\BipartiteGraph)$. Arbitrarily choose an exclusive cylinder set $\EventSet\sim\BipartiteGraph$ which satisfies $ \mu(\EventSet)=\delta\lambda\vec{p}$. 
Let $\EventSet=\{\Event_1,...,\Event_n\}$. For each $i\in L(\BipartiteGraph)$, there is a base $\EventB_i$ of $\Event_i$ such that $\dim(\EventB_i)=\Neighbor_\BipartiteGraph(i)$. Arbitrarily choose a subset $\EventB'_i\subset \EventB_i$ with $\mu(\EventB'_i)=\mu(\EventB_i)/\delta$. Let $\EventSet'=\{\Event'_1,...,\Event'_n\}$ where each $\Event'_i$ is the cylinder with base $\EventB'_i$. It is easy to check that $\EventSet'\sim\BipartiteGraph$, $\mu(\EventSet')=\vec{q}$, and $\EventSet'$ is exclusive.
%Assume $\EventSet=\{\Event_1,...,\Event_n\}$. For any $i\in[n]$, let $\EventB_i$ be the base of cylinder $\Event_i$, and arbitrarily choose a subset $\EventB'_i\subset \EventB_i$ with $\mu(\EventB'_i)=\mu(\EventB_i)/\delta$. Let $\Event'_i$ be the cylinder with $\EventB'_i$ as base, and $\EventSet'=\{\Event'_1,...,\Event'_n\}$. It is easy to check that $\EventSet'$ conforms with $\BipartiteGraph$, $\mu(\EventSet')=\vec{q}$, and $\EventSet'$ is exclusive.
\end{proof}

The significance of Theorem \ref{Conj:GapGeom} lies in that it enables to decide whether a gap exists without checking Shearer's bound. 

\begin{remark}
Given a bigraph $\BipartiteGraph=([n],[m],E)$ and a vector $\vec{p}\in (0,1)^n$, consider three real numbers that are of special interest. $\lambda_1,\lambda_2$ are such that $\lambda_1 \vec{p}\in \partial(\BipartiteGraph)$ and $\lambda_2 \vec{p}\in \partial_a(\DependencyGraph_{\BipartiteGraph})$, respectively. $\lambda_3$ is the maximum $\lambda$ such that there is an exclusive cylinder set $\EventSet\sim\BipartiteGraph$ with $ \mu(\EventSet)=\lambda\vec{p}$. 
It is not difficult to see that $\lambda_1\geq\lambda_2\geq\lambda_3$. %always holds. 
An equivalent form of Theorem \ref{Conj:GapGeom} is that the three numbers are either all equal or pairwise different. 
% This might be a surprise. 
\end{remark}

\subsection{Reduction Rules}
Given a bigraph $\BipartiteGraph$, we define the following 5 types of operations on $\BipartiteGraph$.
\begin{enumerate}
\item Delete-Variable: Delete a vertex $j\in R(\BipartiteGraph)$ with $|\Neighbor(j)|\leq 1$, and remove the incident edge if any. 
\item Duplicate-Event: Given a vertex $i\in L(\BipartiteGraph)$, add a vertex $i'$ to $L(\BipartiteGraph)$, and add edges incident to $i'$ so that $\Neighbor(i')=\Neighbor(i)$.
\item Duplicate-Variable:  Given a vertex $j\in R(\BipartiteGraph)$, add a vertex $j'$ to $R(\BipartiteGraph)$, and add some edges incident to $j'$ so that $\Neighbor(j')\subseteq \Neighbor(j)$.
\item Delete-Edge: Delete an edge from $E$ provided that the base graph remains unchanged.
\item Delete-Event: Delete a vertex $i\in L(\BipartiteGraph)$, and remove all the incident edges.
\end{enumerate}
We also define the inverses of the above operations. 
The inverse of an operation $O$ is the operation $O'$ such that for any $\BipartiteGraph$, $O'(O(\BipartiteGraph))=O(O'(\BipartiteGraph))=\BipartiteGraph$.

The next theorems show how these operations influence the existence of gaps.

\begin{theorem}\label{reductioneffect}
A bigraph $\BipartiteGraph$ is gapful, if and only if it is gapful after applying Delete-Variable, Duplicate-Event, Duplicate-Variable, or their inverse operations.
\end{theorem}

\LongVersion
\begin{proof}
(Delete-Variable): It is trivial.

(Duplicate-Event): Without loss of generality, assume that the vertex $n+1$ is added to $L(\BipartiteGraph)$ and $\Neighbor(n+1)=\Neighbor(n)$. Let $\BipartiteGraph'=([n+1],[m],E')$ be the resulting bigraph.

On the one hand, suppose that $\BipartiteGraph$ is gapless. Arbitrarily choose $\vec{p}'$. Let $\vec{p}=(p'_1,...,p'_{n-1},p'_n+p'_{n+1})$. From Lemma \ref{le:probabilityboundexist}, we have there exists a unique $\lambda>0$ such that $\lambda\vec{p}\in\partial(\BipartiteGraph)$. From Theorem \ref{Conj:GapGeom}, we have there is a exclusive cylinder set $\EventSet=\{\Event_1,...,\Event_{n}\}$ in $\mathbb{I}^m$ such that $\EventSet$ conforms with $\BipartiteGraph$ and $\mu(\EventSet)=\lambda\vec{p}$. Partition the base of cylinder $\Event_n$ such that the resulting disjoint cylinders $\Event'_n$ and $\Event'_{n+1}$ satisfy $\mu(\Event'_n) = \lambda p'_n,\mu(\Event'_{n+1}) = \lambda p'_{n+1}$. Let $\EventSet'=\{\Event_1,...,\Event_{n-1},\Event'_{n},\Event'_{n+1}\}$. One can check that $\EventSet'$ is exclusive with respective to with $\BipartiteGraph'$, $\mu(\EventSet')=\lambda \vec{p}'$, and  $\mu(\cup_{\Event\in\EventSet'}\Event)=1$. This means that $\BipartiteGraph'$ is gapless, by Theorem \ref{Conj:GapGeom}.
  
On the other hand, suppose that $\BipartiteGraph'$ is gapless. Arbitrarily choose $\vec{p}$. Let $\vec{p'}=(p_1,...,p_{n-1},p'_n,p'_{n+1})$ such that $p'_n + p'_{n+1} = p_n$. From Lemma \ref{le:probabilityboundexist}, we have there exists a unique $\lambda>0$ such that $\lambda\vec{p'}\in\partial(\BipartiteGraph')$. From Theorem \ref{Conj:GapGeom}, we have there is a exclusive cylinder set $\EventSet'=\{\Event'_1,...,\Event'_{n},\Event'_{n+1}\}$ in $\mathbb{I}^m$ such that $\EventSet'$ conforms with $\BipartiteGraph'$ and $\mu(\EventSet')=\lambda\vec{p'}$. By Lemma \ref{le:boundaryisunique}, it is easy to see that $\mu(\Event'_n\cap\Event'_{n+1})=0$. Let $\EventSet=\{\Event'_1,...,\Event'_n\cup\Event'_{n+1}\}$. One can check that $\EventSet$ is exclusive with respective to with $\BipartiteGraph$, $\mu(\EventSet)=\lambda \vec{p}$, and  $\mu(\cup_{\Event\in\EventSet}\Event)=1$. This means that $\BipartiteGraph$ is gapless, by Theorem \ref{Conj:GapGeom}. 

(Duplicate-Variable): Without loss of generality, assume that the vertex $m+1$ is added to $R(\BipartiteGraph)$ and $\Neighbor(m+1)\subseteq \Neighbor(m)$. Let $\BipartiteGraph'=([n],[m+1],E')$ be the resulting bigraph. Since $G_\BipartiteGraph=G_{\BipartiteGraph'}$, we only have to show $\partial(\BipartiteGraph)=\partial(\BipartiteGraph')$. Arbitrarily fix $\vec{p}\in (0,1)^n$.
Suppose $\lambda\vec{p}\in\partial(\BipartiteGraph)$ and $\lambda'\vec{p}\in\partial(\BipartiteGraph')$.

Since $\lambda\vec{p}\in\partial(\BipartiteGraph)$, there is a cylinder set $\EventSet=\{\Event_1,...,\Event_n\}$ in $\mathbb{I}^m$ such that $\EventSet\sim\BipartiteGraph$, $\mu(\EventSet)=\lambda\vec{p}$, and $\mu(\cup_{i\in[n]}\Event_i)=1$. For any $i\in[n]$, define $\Event'_i=\Event_i\times \Interval^{\{m+1\}}$. Let $\EventSet'=\{\Event'_1,...,\Event'_n\}$. We have $\EventSet'$ conforms with $\BipartiteGraph'$, $\mu(\EventSet')=\lambda\vec{p}$, and $\mu(\cup_{i\in[n]}\Event'_i)=1$, so $\lambda'\leq \lambda$.

On the other hand, since $\lambda'\vec{p}\in\partial(\BipartiteGraph')$, there is a discrete cylinder set $\EventSet'=\{\Event'_1,...,\Event'_n\}$ in $\mathbb{I}^{m+1}$ such that $\EventSet'$ conforms with $\BipartiteGraph'$, $\mu(\EventSet')=\lambda'\vec{p}$, and $\mu(\cup_{i\in[n]}\Event'_i)=1$. By discreteness, one can partition $\Interval^{\{m,m+1\}}$ into disjoint rectangles $\Delta_1,...,\Delta_K$ such that for each $i\in [n]$, there are sets $\Event_{ik}\subseteq \mathbb{I}^{m-1}$ for $k\in[K]$ satisfying $\Event'_i=\cup_{k\in[K]}\Event_{ik}\times \Delta_k$. For each $i\in L(\BipartiteGraph)\setminus \Neighbor(m)$, since $\{m,m+1\}\cap \Neighbor(i)=\emptyset$,  $\Event_{ik}$ does not depend on $k$, and is denoted by $B_i$. Since $\mu(\cup_{i\in[n]}\Event'_i)=1$, we have $\mu(\cup_{i\in[n]}\Event_{ik})=1$ for any $k$. Now partition $\mathbb{I}^{\{m\}}$ into disjoint intervals $\Gamma_1,...,\Gamma_K$ with $\mu(\Gamma_k)=\mu(\Delta_k)$ for each $k$. Define $\EventSet=\{\Event_1,...,\Event_n\}$ in $\Interval^m$ such that $\Event_i=\cup_{k\in[K]}\Event_{ik}\times \Gamma_k$ for $i\in \Neighbor(m)$ and $\Event_i=B_i\times \Interval^m=\Event'_i$ for $i\in L(\BipartiteGraph)\setminus \Neighbor(m)$. It is straightforward to check that $\EventSet$ conforms with $\BipartiteGraph$, $\mu(\EventSet)=\lambda'\vec{p}$, and $\mu(\cup_{i\in[n]}\Event_i)=1$. Hence, $\lambda\leq \lambda'$.

As a result, $\partial(\BipartiteGraph)=\partial(\BipartiteGraph')$. Recall that $G_\BipartiteGraph=G_{\BipartiteGraph'}$, so $\BipartiteGraph$ is gapful if and only if so is $\BipartiteGraph'$.
\end{proof}
\LongVersionEnd

\begin{theorem}\label{fromgaplesstogapless}
A gapless bigraph remains gapless after applying Delete-Event or the inverse of Delete-Edge.% If the bigraph obtained by Delete-Event or the inverse of Delete-Edge is gapful, then the original bigraph must be gapful. 
\end{theorem}

\begin{theorem}\label{fromgapfultogapful}
A gapful bigraph remains gapful after applying Delete-Edge or the inverse of Delete-Event.% If the bigraph obtained by Delete-Edge or the inverse of Delete-Event is gapless, then the original bigraph must be gapless.
\end{theorem}
 
\LongVersion
The proofs of the above two theorems are similar to that of Theorem \ref{reductioneffect}, so they are omitted.
\LongVersionEnd

Because the operations can be pipelined, applying them in combination may produce interesting results. The following corollaries are some examples. 

\begin{definition}[Combinatorial bigraph]
Given two positive integers $m<n$, let $\BipartiteGraph_{n,m}=([(^n_m)],[n], E_{n,m})$ where $(i,j)\in E_{n,m}$ if and only if $j$ is in the $m$-sized subset of $[n]$ represented by $i$.
 $\BipartiteGraph_{n,m}$ is called the $(n,m)$-combinatorial bigraph.
\end{definition}

\begin{corollary}\label{cor:densegraphgapless}
If $\BipartiteGraph_{n,m}$ is gapless, then so is $\BipartiteGraph_{n+c,m+c}$ for any integer $c\geq 1$.
\end{corollary}

\LongVersion
\begin{proof} We only need to prove for $c=1$. % that if $\BipartiteGraph_{n,m}$ is gapless, $\BipartiteGraph_{n+1,m+1 }$ is also gapless. %hen the theorem becomes obvious by reduction. 

First, apply Delete-Edge to $\BipartiteGraph_{n+1,m+1}$ as follows. %Let the vertices of $R(\BipartiteGraph_{n,m})$ be $X_1,\cdots,X_n$, and the vertices of $R(\BipartiteGraph_{n+1,m+1})$ be $X_1,\cdots,X_n,X_{n+1}$. 
For each vertex $i$ in $[(^{n+1}_{m+1})]$, if $(i,n+1) \in E_{n+1,m+1}$, delete $(i,n+1)$. Otherwise, delete an arbitrary edge of $i$. 

Then, apply Delete-Variable to the bigraph, i.e., delete the vertex $n+1$. 

Finally, apply the inverse operation of Duplicate-Event to the bigraph. 

For any $m$-set $S\subset [n]$, suppose the set $S\cup\{n+1\}$ is represented by $i\in [(^{n+1}_{m+1})]$. After applying Delete-Edge to $i$, the neighborhood of $i$ is exactly $S$. This means that the final bigraph is exactly $\BipartiteGraph_{n,m}$. Because $\BipartiteGraph_{n,m}$ is gapless, from Theorem \ref{reductioneffect} and Theorem \ref{fromgapfultogapful}, we have that $\BipartiteGraph_{n+1,m+1}$ is also gapless.
\end{proof}

\LongVersionEnd
\begin{corollary}\label{cor:sparsegraphgapful}
If $\BipartiteGraph_{n,m}$ is gapful, then for any integer $c \geq 1$, $\BipartiteGraph_{cn,cm}$ is also gapful.
\end{corollary}
\LongVersion
\begin{proof}  
We apply operations to $\BipartiteGraph_{cn,cm}$ in two steps.

First, apply Delete-Event to $\BipartiteGraph_{cn,cm}$. Given an $m$-set $S\subset [n]$, define $f(S)=\cup_{i\in S}\{ki: k\in [c]\}$. Delete all vertices from $L(\BipartiteGraph_{cn,cm})$ except those representing $f(S)$ for some $S\subset [n]$. Let $\BipartiteGraph'$ be the resulting bigraph.

Second, apply the inverse operation of Duplicate-Variable to $\BipartiteGraph'$. It is easy to see that for any $k_1 j,k_2 j\in R(\BipartiteGraph')$ with $k_1,k_2\in [c]$ and $j\in [n]$, $\Neighbor_{\BipartiteGraph'}(k_1 j)=\Neighbor_{\BipartiteGraph'}(k_2 j)$. Hence, we delete all vertices in $R(\BipartiteGraph')\setminus [n]$ from $R(\BipartiteGraph')$, preserving gapful/gapless.

It is easy to verify that the final bigraph is exactly $\BipartiteGraph_{n,m}$. Because $\BipartiteGraph_{n,m}$ is gapful, from Theorem \ref{reductioneffect} and Theorem \ref{fromgaplesstogapless}, we have that $\BipartiteGraph_{cn,cm}$ is also gapful.
\end{proof}
\LongVersionEnd

\begin{definition}[Sparsified bigraphs]
A bigraph $H' = ([n'],[m'],E')$ is called a \emph{sparsification} of $\BipartiteGraph= ([n],[m],E)$ if $[n'] = [n],[m'] \subseteq [m], E' \subseteq E$ and their base graphs are the same. 
\end{definition}

By Theorem \ref{reductioneffect} and Theorem \ref{fromgaplesstogapless}, we know that if $\BipartiteGraph$ is gapful, all sparsifications of $\BipartiteGraph$ must be gapful. Applying Corollary \ref{cor:sparsegraphgapful}, we get the following result.

\begin{corollary}\label{cor:Sparsifiedgraphgapful}
If $\BipartiteGraph_{n,m}$ is gapful, all sparsifications of $\BipartiteGraph_{cn,cm}$ are also gapful for any integer $c \geq 1$. 
\end{corollary}

\section{Relationship between gaps and cycles}\label{sec:characterization}
In this section, we show that a bigraph has a gap is almost equivalent to that its base graph has an cycle. The only case that is not completely known is when the bigraph does not  contain any cyclic bigraph but its base graph has a 3-clique. Many examples in this case is gapless, but we find one that turns out to be gapful. 

\LongVersion
We also study gaps from a dependency-graph-oriented perspective. Namely, a dependency graph is a-gapful if at least one corresponding bigraph is gapful, while is strongly a-gapful if all corresponding bigraphs are gapful. Intuitively speaking, the two concepts serve as a lower bound and an upper bound of the notion of gapfulness. Characterization of strongly a-gapful graphs was initiated by Kolipaka et al. \cite{kolipaka2011moser}  and has been open for 6 years. 

\subsection{Gaps are not equivalent to cycles}
\LongVersionEnd

First of all, we prove that any treelike bigraph is gapless. Recall that a bigraph is called treelike if its base graph is a tree. Basically, for a vector on boundary, we construct an exclusive cylinder set, which leads to the result by Theorem \ref{Conj:GapGeom}. To ensure exclusiveness, the unit interval in each dimension is divided into two disjoint parts, each of which is assigned to one of the two cylinders depending on this dimension. The construction is feasible because the base graph is a tree.
%\begin{restatable}{theorem}{graph}
%\label{dependencygaplessdecision}
%Treelike bigraphs are gapless.
%\end{restatable}
\graph*

%\begin{theorem}\label{treeisgapless}
%Treelike bigraphs are gapless.
%\end{theorem}
\LongVersion
\begin{proof}
Arbitrarily choose a treelike bigraph $\BipartiteGraph= ([n+1],[m],E)$. Since the case where $n=1$ is trivial, we just consider $n>1$. $G_\BipartiteGraph$ is a tree means that any vertex in $R(\BipartiteGraph)$ has at most two neighbors in $\BipartiteGraph$ . Hence, by Theorem \ref{reductioneffect}, it does not lose generality to assume that: 1. any vertex in $R(\BipartiteGraph)$ has exactly two neighbors in $\BipartiteGraph$, and 2. any two vertices in $L(\BipartiteGraph)$ have no more than one common neighbor in $\BipartiteGraph$. Since $G_\BipartiteGraph$ is a tree, one has $m=n$.

Let $\vec{p}\in (0,1)^{n+1}$ be a boundary vector of $\BipartiteGraph$. We will construct a set $\EventSet$ of cylinders $\Event_1,...,\Event_{n+1}\subset\mathbb{I}^m=\mathbb{I}^n$ such that $\mu(\EventSet)=\vec{p}$ and $\EventSet$ conforms with $\BipartiteGraph$. Recall that for any $j\in [n]$, $X_j$ is the coordinate variable of the $j$-th dimension of $\mathbb{I}^n$.

We regard $G_\BipartiteGraph$ as a tree rooted at the vertex $n+1$. For any vertex $i\in [n+1]$, let $\mathcal{C}(i)$ be the set of children of $i$. % Again without loss of generality, for any $k\in \mathcal{C}(i)$, assume that the unique common neighbor of $i,k\in [n+1]$ in $\BipartiteGraph$ is $k\in [m]$. 
Without loss of generality, for any $k\in \mathcal{C}(i)$, assume that $\Neighbor_\BipartiteGraph(k)\cap \Neighbor_\BipartiteGraph(i)=\{k\}$, which means that both $\Event_i$ and $\Event_k$ depend on $X_k$. 

Define $\vec{q}=(q_1,...,q_n)\in \mathbb{R}^n$ to be %$q_i = p_i$ if vertex $i$ is a leaf of $G_\BipartiteGraph$ and $q_i = \lambda p_i / \prod_{k\in \mathcal{C}(i)}(1 - q_k)$ otherwise.
\begin{equation}
\begin{split}
q_i=\left\{\begin{array}{ll} p_i& \textrm{ if vertex } i \textrm{ is a leaf of } G_\BipartiteGraph\\
p_i / \prod_{k\in \mathcal{C}(i)}(1 - q_k)&\textrm{ otherwise}
\end{array}\right.
\end{split}
\end{equation}

\textbf{Claim}: $\vec{q}\in (0,1)^n$.

\textbf{Proof of the Claim}: Suppose for contradiction that there is $i\in [n]$ such that $q_i\notin (0,1)$. Fix such an $i$ each of whose descendant $k$ satisfies $q_k\in (0,1)$. By the definition of $\vec{q}$, we must have $q_i\geq 1$. 

Let $T_i$ be the subtree of $G_\BipartiteGraph$ rooted at $i$. For each $k\in [n+1]$, if $k$ is not a vertex of $T_i$, define $\Event'_k=\emptyset\subset\mathbb{I}^n$. When $k$ is in $T_i$, construct a cylinder $\Event'_k\subset\mathbb{I}^n$ which consists of all the vectors
$(x_1,x_2,...,x_n)$ such that
\begin{equation*}
\begin{split} \left\{\begin{array}{ll} 0\leq x_k\leq q_k & \textrm{ if } k \textrm{ is a leaf of } T_i\\
q_{l}< x_{l}\leq 1, \forall l\in \mathcal{C}(i) & \textrm{ if } k= i\\
0\leq x_k\leq q_k,q_{l}< x_{l}\leq 1, \forall l\in \mathcal{C}(k) & \textrm{ otherwise}
\end{array}\right.
\end{split}
\end{equation*}

Define vector $\vec{p}'=(p'_1,...,p'_{n+1})$ such that 
\begin{equation*}
\begin{split}
p'_k=\left\{\begin{array}{ll} \prod_{k\in \mathcal{C}(i)}(1 - q_k) & \textrm{ if } k=i\\
                                          p_k & \textrm{ if }  k \textrm{ is in }T_i  \textrm{ and }  k\neq i\\
                                          0 & \textrm{ otherwise}
\end{array}\right.
\end{split}
\end{equation*}

%If $k$ is a leaf of $T_i$, $\EventB_k$ is the set $0\leq X_k\leq q_k$. If $k\neq i$ is a non-leaf vertex of $T_i$, $\EventB_k$ is defined to be $0\leq X_k\leq q_k$ and $q_{l}< X_{l}\leq 1$ for any $l\in \mathcal{C}(k)$. If $k= i$, $\EventB_k$ is the set $q_{l}< X_{l}\leq 1$ for any $l\in \mathcal{C}(i)$. For others $k$, define $\EventB_k=\emptyset$. 
%Let $p'_i=\prod_{k\in \mathcal{C}(i)}(1 - q_k), p'_k=p_k$ for any vertex $k\neq i$ of $T_i$, and $p'_l=0$ for any other $l\in [n+1]$. 
Then the cylinder set $\EventSet'=\{\Event'_k | k\in[n+1]\}$ conforms with $\BipartiteGraph$, and $\mu(\EventSet')=\vec{p}'<\vec{p}$. 

Now we prove that $\cup _{k\in[n+1]}\Event'_k =\mathbb{I}^n$. Arbitrarily fix $\vec{x}=(x_1,...,x_n)\in \mathbb{I}^n$. 

Let $l=i$. Then, if there is $k\in \mathcal{C}(l)$ such that $0\leq x_k\leq q_k$, let $l$ be such a $k$. Iterate this process and finally one of the following two cases must be reached.

Case 1: $\mathcal{C}(l)=\emptyset$, namely $l$ is a leaf.

Case 2: $\mathcal{C}(l)\neq\emptyset$ and $q_k<x_k\leq 1$ for any $k\in \mathcal{C}(l)$.

Let the final $l$ be $l_0$.
We can see that $\vec{x}\in \Event'_i$ if  $l_0=i$. Otherwise, the iteration guarantees that $0\leq x_{l_0}\leq q_{l_0}$, so it also holds that $\vec{x}\in \Event'_{l_0}$. To sum, we always have $\vec{x}\in \Event'_{l_0}$, which implies that $\cup _{k\in[n+1]}\Event'_k =\mathbb{I}^n$. Considering that $\vec{p}\in \partial(\BipartiteGraph)$ and $\vec{p}>\vec{p}'$, we reach a contradiction due to Lemma \ref{le:boundaryisunique}. \textbf{The Claim is proven.}

Then we can construct cylinders $\Event_1,...,\Event_{n+1}\subset\mathbb{I}^n$ as follows. For any $i\in [n+1]$, $\Event_i\subset\mathbb{I}^n$ consists of all the vectors $(x_1,x_2,...,x_n)$ such that 
\begin{equation*}
\begin{split}
\left\{\begin{array}{ll} 0\leq x_i\leq q_i & \textrm{ if } i \textrm{ is a leaf of } G_\BipartiteGraph\\
q_{k}< x_{k}\leq 1, \forall k\in \mathcal{C}(i) & \textrm{ if } i= n+1\\
0\leq x_i\leq q_i,q_{k}< x_{k}\leq 1, \forall k\in \mathcal{C}(i) & \textrm{ otherwise}
\end{array}\right.
\end{split}
\end{equation*}

%For any $i\in [n+1]$, $\Event_i$ is defined by $0\leq X_i\leq q_i$ if $i$ is a leaf of $G_\BipartiteGraph$, $\Event_i$ is such that $q_{k}< X_{k}\leq 1$ for any $k\in \mathcal{C}(i)$ if $i$ is the root $n+1$, otherwise $\Event_i$ is such that $0\leq X_i\leq q_i$ and $q_{k}< X_{k}\leq 1$ for any $k\in \mathcal{C}(i)$. 
Define $\vec{p}'=(p_1,...,p_n,p'_{n+1})$ where $p'_{n+1}=\prod_{k\in \mathcal{C}(n+1)}(1 - q_k)$. 
It is easy to observe three facts.

First,  $\EventSet=\{\Event_i | i\in[n+1]\}$ is exclusive and conforms with $\BipartiteGraph$.

Second, $\mu(\EventSet)=\vec{p}'$.

Third, $\cup _{i\in[n+1]}\Event_i=\mathbb{I}^n$, which follows from the proof of $\cup _{k\in[n+1]}\Event'_k =\mathbb{I}^n$ in the above Claim. 

Since $\vec{p}\in\partial(\BipartiteGraph)$, we have $\vec{p}'\geq\vec{p}$. Arbitrarily choose $\Event''_{n+1}\subseteq \Event_{n+1}$ such that: 1. $\Event''_{n+1}$ only depends on $X_i$'s with $i\in \mathcal{C}(n+1) $, and 2. $\mu(\Event''_{n+1})=p_{n+1}$. Let $\EventSet''=\{\Event_1,...,\Event_n, \Event''_{n+1}\}$. We know that $\mu(\EventSet'')=\vec{p}$ and $\EventSet''$ is exclusive with respect to $\BipartiteGraph$. Because $\vec{p}\in\partial(\BipartiteGraph)$, by Theorem \ref{Conj:GapGeom}, $\BipartiteGraph$ is gapless in the direction of $\vec{p}$.
\end{proof}
\LongVersionEnd

Using of the constructed cylinders, we obtain a system of equations whose solution determines the boundary of a treelike bigraph.

\begin{corollary}\label{computetreeboundary}
Given a bigraph $\BipartiteGraph=([n],[m],E)$ such that $G_\BipartiteGraph$ is a tree, appoint the vertex $n$ as the root of $G_\BipartiteGraph$. For any $\vec{p}\in (0,1)^n$, $\lambda\vec{p}\in\partial(\BipartiteGraph)$ if and only if $\lambda$ is the minimum positive solution to the equation system: $q_i = \lambda p_i$ if vertex $i$ is a leaf of $G_\BipartiteGraph$, $q_i = \lambda p_i / \prod_{k \text{ is a child of }i}(1 - q_k)$ if $i\neq n$ and is not a leaf, and $\lambda p_n = \prod_{k \text{ is a child of }n}(1 - q_k)$.
\end{corollary}
\LongVersion
\begin{proof}
This immediately follows from the construction of $\EventSet$ in the proof of Theorem \ref{dependencygaplessdecision}.
\end{proof}
\LongVersionEnd

Now we show that cyclic bigraphs are gapful. Though in principle this can be shown by a combination of \cite[Theorem 1]{shearer1985problem} and the results in Section \ref{sec:cyclicboundary}, it is tough since both Shearer's inequality system and the high degree polynomial in Theorem \ref{cycleboundary} are hard to solve. Hence we do it in another way. 
Specifically, for the vector $\vec{q}=(\frac{1}{4}+\epsilon,...,\frac{1}{4}+\epsilon)$ where $\epsilon>0$ is small enough, we show two facts. First, the vector $\vec{q}$ lies in the interior of the cyclic bigraph. Second, $\vec{q}$ does not allow any exclusive cylinder set. By Theorem \ref{Conj:GapGeom}, these facts immediately imply Theorem \ref{cyclesaregapful}.

\cycle*

% \begin{theorem}\label{cyclesaregapful}
% All cycles are gapful.
% \end{theorem}

\LongVersion
\begin{proof}
It is enough to consider the canonical $n$-cyclic bigraphs $\BipartiteGraph_n=([n], [n], E_n)$ where $E_n=\{(i,i), (i,(i+1)(\overline{\textrm{mod}} ~n)): i\in [n]\}$. Again for convenience of presentation, ``$(\overline{\textrm{mod}} ~n)$" will be omitted when it is clear from the context. Arbitrarily fix $n$.

%We will construct $n$ cylinders in the unit cube $\Interval^n$.

For any $i\in [n]$, let $\Event_i=\{(x_1,...,x_n): \frac{1}{2}\leq x_i\leq 1, 0\leq x_{i+1}< \frac{1}{2}\}$. Let $\EventSet=\{\Event_1,...,\Event_n\}$, and $\vec{p}=(\frac{1}{4},...,\frac{1}{4})\in (0,1)^n$. It is straightforward to check that $\EventSet$  is exclusive with respect to $\BipartiteGraph_n$, $\mu(\EventSet)=\vec{p}$, and $\mu(\cup_{i\in [n]}\Event_i)<1$. Arbitrarily choose $0<\epsilon<\frac{1}{n}(1-\mu(\cup_{i\in [n]}\Event_i))$. Let $\vec{q}=(\frac{1}{4}+\epsilon,...,\frac{1}{4}+\epsilon)\in (0,1)^n$. Then we prove two claims.

\textbf{Claim 1}: $\vec{q}\in \mathcal{I}(\BipartiteGraph_n)$.

Assume for contradiction that there is a cylinder set $\EventSetB=\{\EventB_1,...,\EventB_n\}\sim\BipartiteGraph_n$ such that $\mu(\EventSetB)=\vec{q}$ and $\Pr(\cup_{i\in[n]}\EventB_i)=1$. For each $i\in[n]$, arbitrarily choose a cylinder $B'_i$ such that $\EventB'_i\subset \EventB_i$, $\mu(\EventB'_i)=1/4$,  and $\EventB'_i$ only depends on $X_i$ and $X_{i+1}$. Let $\EventSetB'=\{\EventB'_1,...,\EventB'_n\}$. We have that $\EventSetB'$ conforms with $\BipartiteGraph_n$ and $\mu(\EventSetB')=\vec{p}$. On the one hand, $\mu(\cup_{i\in[n]}\EventB'_i)\geq 1-n\epsilon >\mu(\cup_{i\in [n]}\Event_i)$. On the other hand, since $\EventSet$ is exclusive, by Lemma \ref{WorstPlacement}, $\mu(\cup_{i\in[n]}\EventB'_i)\leq \mu(\cup_{i\in [n]}\Event_i)$. We reach a contradiction, so \textbf{Claim 1 holds}.

\textbf{Claim 2}:  For any cylinder set $\EventSetB\sim\BipartiteGraph_n$ with $\mu(\EventSetB)=\vec{q}$, $\EventSetB$ is not exclusive.

Arbitrarily fix a cylinder set $\EventSetB=\{\EventB_1,...,\EventB_n\}\sim\BipartiteGraph_n$ with $\mu(\EventSetB)=\vec{q}$. For each $i\in [n]$, let $\widetilde{\EventB_i}\subset\Interval^{\{i,i+1\}}$ be a base of $\EventB_i$, and choose the minimum subsets $\Delta'_i\subseteq \Interval^{\{i\}}$ and $\Delta_{i+1}\subseteq \Interval^{\{i+1\}}$ such that $\mu(\widetilde{\EventB_i}\setminus (\Delta'_i\times \Delta_{i+1}))=0$. Let $x_i=\mu(\Delta_i), x'_i=\mu(\Delta'_i)$. Then $\mu(\widetilde{\EventB_i}\setminus (\Delta'_i\times \Delta_{i+1}))=0$ implies that $x'_i x_{i+1}\geq \mu(\widetilde{\EventB_i})=\mu(\EventB_i)>\frac{1}{4}$. Hence, $\prod^n_{i=1}(x_i x'_i)>\frac{1}{4^n}$. There must be some $i\in [n]$ such that $x_i x'_i>\frac{1}{4}$, which in turn means that $x_i+x'_i>1$. As a result, $\mu(\Delta_i\cap\Delta'_i)>0$, implying that $\mu(B_{i-1}\cap B_i)>0$. \textbf{Claim 2 holds}.

Altogether, by Theorem \ref{Conj:GapGeom}, $\BipartiteGraph_n$ is gapful.
\end{proof}
\LongVersionEnd

By Theorem \ref{cyclesaregapful}, we can get a large class of gapful bigraphs.
\begin{definition}[Containing]
We say that a bigraph $\BipartiteGraph$ contains another bigraph $\BipartiteGraph'$, if there are injections $\pi_L: L(\BipartiteGraph')\rightarrow L(\BipartiteGraph)$ and $\pi_R: R(\BipartiteGraph')\rightarrow R(\BipartiteGraph)$ such that the following two conditions hold simultaneously:
\begin{enumerate}
\item For any $i\in L(\BipartiteGraph')$ and $j\in R(\BipartiteGraph')$,  $\pi_R (j)\in \Neighbor_{\BipartiteGraph} (\pi_L(i))$ if and only if $j\in \Neighbor_{\BipartiteGraph'} (i)$.
\item For any $j\in R(\BipartiteGraph)\setminus \pi_R(R(\BipartiteGraph'))$, $j\notin \Neighbor_{\BipartiteGraph}(\pi_L(i))\cap \Neighbor_{\BipartiteGraph}(\pi_L(k))$ for any $i,k\in L(\BipartiteGraph')$.
\end{enumerate}
%
%
%Given two bigraphs $\BipartiteGraph_1=([n_1],[m_1],E_1)$ and $\BipartiteGraph_2=([n_2],[m_2],E_2)$ with $n_1\leq n_2$ and $m_1\leq m_2$, we say that $\BipartiteGraph_2$ contains $\BipartiteGraph_1$ if there are injections $\pi_L: [n_1]\rightarrow [n_2]$ and $\pi_R: [m_1]\rightarrow [m_2]$ such that the following two conditions hold simultaneously:
%\begin{enumerate}
%\item For any $i\in[n_1],j\in[m_1]$, if $(\pi_L(i),\pi_R(j))\in E_2$% if $(i,j)\in E_1$
%, then$(i,j)\in E_1$ .
%\item For any $i,k\in[n_1]$, there is no $j\in[m_2]\setminus \pi_R([m_1])$ such that $(\pi_L(i), j)\in E_2$ and $(\pi_L(k), j)\in E_2$.
%\end{enumerate}
\end{definition}
Intuively, $\BipartiteGraph$ contains $\BipartiteGraph'$ means that $\BipartiteGraph'$ can be \emph{embedded} in $\BipartiteGraph$ without incurring extra dependency.

By Theorem \ref{reductioneffect} and Theorem \ref{fromgaplesstogapless}, a bigraph $\BipartiteGraph$ is gapful if it contains a gapful one. According to Theorem \ref{cyclesaregapful}, we obtain the following result.
\begin{corollary}\label{cor:Cyclegapful} Any bigraph containing a cyclic one is gapful.
\end{corollary}

%/********gapful dependency graph
%
%Symmetrically to the concept of gapful/gapless dependency bigraphs, we can also define gapful/gapless dependency graphs.
%
%\begin{definition}[Gap in dependency graphs]
%A dependency graph $G$ is said to be gapless if $\BipartiteGraph$ is gapless for every $\BipartiteGraph$ such that $G_\BipartiteGraph=G$. Otherwise, $G$ is called gapless.
%\end{definition}
%
%%%\graph* 
%
% \begin{theorem}\label{dependencygaplessdecision}
% A dependency graph is gapless if and only if it is a tree.
% \end{theorem}
%\begin{proof}
%This immediately follows from Theorems \ref{treeisgapless} and \ref{cyclesaregapful}.
%\end{proof}

Based on Theorem \ref{dependencygaplessdecision} and Corollary
\ref{cor:Cyclegapful}, it is natural to have the following conjecture:

\begin{conjecture}[Gap conjecture]
A bigraph is gapful if and only if it contains a cyclic bigraph. 
\end{conjecture}

% Gap conjecture: A bigraph is gapful if and only if it contains a cyclic bigraph. 

We have already known that the sufficiency does hold. 
As to the necessity, assume that the bigraph $\BipartiteGraph$ does not contain any cyclic one. We analyze case by case.

Case 1: the base graph is a tree. By Theorem \ref{dependencygaplessdecision}, $H$ is gapless, as desired.

Case 2: the base graph has cycles. Since $\BipartiteGraph$ does not contain a cyclic bigraph, its base graph does not have induced cycles longer than three. As a result, solving the conjecture is equivalent to answering the following question \textbf{$Q$: Is a bigraph gapless if it does not contain any cyclic one but its base graph has 3-cliques?}

First have look at a simple example of bigraph $\BipartiteGraph=([3],[1],E)$ with $ E=[3]\times [1]$. It satisfies the condition of question $Q$. One can easily check that $\partial(\BipartiteGraph)=\{(p_1,p_2,p_3): p_1+p_2+p_3=1\}=\partial_a(\DependencyGraph_\BipartiteGraph)$. So, $\BipartiteGraph$ is gapless.

For more evidence, recall $\BipartiteGraph_{n,m}$, the $(n,m)$-combinatorial bigraph.
As a special case, $\BipartiteGraph_{3,2}$ is the canonical $3$-cyclic bigraph $\BipartiteGraph_3$. Generally, we have the following observations:

First, $m=1 $: Only sets of independent events can conform with $\BipartiteGraph_{n,m}$.

Second, $2\leq m\leq \frac{2}{3}n$: $\BipartiteGraph_{n,m}$ contains $3$-cyclic bigraphs, so it is gapful.

Third, $m> \frac{2}{3}n$: $\BipartiteGraph_{n,m}$ does not contain cyclic bigraphs, but the base graph have 3-cliques since it is a complete graph. 
We mainly consider bigraphs in this category.
%
%
%\begin{description}
%\item[ $m=1 $: ] Only sets of independent events can conform with $\BipartiteGraph_{n,m}$.
%\item[ $2\leq m\leq \frac{2}{3}n$: ] $\BipartiteGraph_{n,m}$ contains $3$-cyclic bigraphs, so it is gapful.
%\item[ $m> \frac{2}{3}n$: ] $\BipartiteGraph_{n,m}$ does not contain cyclic bigraphs, but the base graph have 3-cliques since it is a complete graph.
%\end{description}

\begin{theorem}\label{43isgapless}
$\BipartiteGraph_{4,3}$ is gapless.
\end{theorem}
\LongVersion
\begin{proof}
%We construct Assume that $\mathcal{X}=\{x,y,z,w\},\mathbb{X}_1=\{x,y,w\}, \mathbb{X}_2=\{x,y,z\}, \mathbb{X}_3=\{x,z,w\},\mathbb{X}_4=\{y,z,w\}$
By \cite[Theorem 1]{shearer1985problem}, $\vec{p}\in\partial_a(\BipartiteGraph_{4,3})$ if and only if $\sum^4_{i=1}p_i=1$. Arbitrarily fix $\vec{p}\in\partial_a(\BipartiteGraph_{4,3})$. Without loss of generality, assume that $p_i\geq p_{i+1}$, for any $1\leq i\leq 3$. Then $ p_4\leq \frac{1}{4}$, $p_3 \leq \frac{1}{3}$, $p_1 \geq \frac{1}{4}$.

We construct four cylinders in the unit 4-cube. Let the four dimensions be $X_1,X_2,X_3,X_4$.
%If $p_A +...+p_D=1$, then there exist $A,B,C,D$ such that $Pr(\bar{A} \cap \bar{B} \cap \bar{C} \cap \bar{D}) = 0$.
Specifically,  define the cylinders as follows:

$A_3$:  $X_4> \frac{1}{2}, X_1\leq 2p_3$.

$A_4$: $X_4 \leq \frac{1}{2}, X_2 \leq 2p_4$.

$A'_1$: $(X_4 \leq \frac{1}{2}, X_2> 2p_4, X_1\leq 2p_3)$ or $(X_4> \frac{1}{2}, X_2 \leq 2p_4, X_1> 2p_3)$.

%
%
%\begin{description}
%\item[$A_3$:]  $w > \frac{1}{2}, x\leq 2p_3$.
%\item[$A_4$:] $w \leq \frac{1}{2}, y \leq 2p_4$.
%\item[$A'_1$:] $(w \leq \frac{1}{2}, y > 2p_4, x\leq 2p_3)$ or $(w > \frac{1}{2}, y \leq 2p_4, x> 2p_3)$
%\end{description}

One can see that $\mu(A'_1)=p_3+p_4-4p_3 p_4$. 
Furthermore, if $p_3>\frac{1}{4}$, $p_3+p_4-4p_3 p_4\leq p_3+p_4-p_4=p_3\leq p_1$. When $p_3\leq \frac{1}{4}$, $p_3+p_4-4p_3 p_4= p_3+(1-4p_3)p_4\leq p_3+(1-4p_3) \frac{1}{4}=\frac{1}{4}\leq p_1$. We always have that $\mu(A'_1)\leq p_1$.

%\begin{description}
%\item[If $p_3>\frac{1}{4}$:] $p_3+p_4-4p_3 p_4\leq p_3+p_4-p_4=p_3\leq p_1$.
%\item[Otherwise:]  $p_3+p_4-4p_3 p_4\leq p_3+(1-4p_3)p_4\leq p_3+(1-4p_3) \frac{1}{4}=\frac{1}{4}\leq p_1$.
%\end{description}
%So, we always have $\Pr(A'_1)\leq p_1$. 

Arbitrarily choose a set $S\subset \Interval^{\{1,2\}}$ in the area $X_1>2p_3,X_2> 2p_4$  such that $\mu(S)=p_1 + 4p_3 p_4 - p_3 -p_4$. Let $S'$ be the cylinder with base $S$.

Define $A_1=A'_1\cup S'$ and $A_2=\overline{A_1\cup A_3\cup A_4}$. It is easy to see that $\mu(A_i)=p_i$ for $1\leq i\leq 4$. The bases $B_i$ of $A_i$, $1\leq i\leq 4$, can be chosen such that $\dim(B_1)=\{1,2,4\}$, $\dim(B_2)=\{1,2,3\}$, $\dim(B_3)=\{1,3,4\}$, $\dim(B_4)=\{2,3,4\}$.
\end{proof}
\LongVersionEnd

Theorem \ref{43isgapless}, together with Corollary \ref{cor:densegraphgapless}, immediately implies the following result. 
\begin{corollary}\label{cor:nminus1gapful} For $n \geq 4$, $\BipartiteGraph_{n,n-1}$ is gapless.
\end{corollary}

Actually, Corollary \ref{cor:nminus1gapful} can be generalized to $\BipartiteGraph_{n,n-m}$ for any fixed $m$ and large enough $n$, as shown in Theorem \ref{nn-2isgapless}.
%Renark: Actually, in the proof of Lemma \ref{Case(4,3)}, $vrbl(A)=\{x,y,w\}$, $vrbl(B)=\{x,y\}$, $vrbl(C)=\{x,w\}$, $vrbl(D)=\{y,w\}$.

%
%\begin{lemma}\label{le:smalleventsdisjoint}
%For any $n$ and $m>n/2$, there is $\epsilon>0$ such that for any $\vec{p}\leq (\epsilon,...,\epsilon)$, there exists an exclusive cylinder set $\EventSet$ satisfying $\EventSet$ conforms with $\BipartiteGraph^\geq_{n,m}$ and $\mu(\EventSet)=\vec{p}$.
%\end{lemma}
%\begin{proof}
%Let $K=T_{n,m}$. Uniformly divide each $\Interval^{\{i\}}$ into $K$ disjoint segments, denoted by $\Delta^{\{i\}}_k,k\in[K]$. For any $k\in[K]$, let $B_k$ be the cylinder whose base is $\prod_{(k,i)\in E^\geq_{n,m}}\Delta^{\{i\}}_k$. Let $\EventSetB$ be the set of the events. Define $\vec{q}\in (0,1]^{K}$ be such that $q_k=\mu(B_k)$ for each $k\in[K]$. We can check that $\EventSetB$ conforms with $\BipartiteGraph^\geq_{n,m}$ and is an exclusive cylinder set. As a result, letting $\epsilon=\min_{k\in[K]}q_k$, the lemma must hold.
%\end{proof}

\LongVersion
\begin{definition}[Upper combinatorial bigraph]
Given positive integers $m<n$, let $T_{n,m}=\sum_{t\geq m}{n \choose t}$. Then each $k\in [T_{n,m}]$ naturally represents a set in $[n]$ that has size at least $m$. Define bigraph $\BipartiteGraph^\geq_{n,m}=([T_{n,m}],[n], E^\geq_{n,m})$ where $(i,j)\in E^\geq_{n,m}$ if and only if $j$ is in the set represented by $i$. $\BipartiteGraph^\geq_{n,m}$ is called the upper $(n,m)$-combinatorial bigraph.
\end{definition}
\LongVersionEnd

Theorem \ref{nn-2isgapless} is proved by construction. Basically, given a boundary vector $\vec{p}$ of $\BipartiteGraph_{n,n-m}$, we identify a small number of dimensions, partition the unit cube $\mathcal{C}$ spanned by these dimensions into ${n \choose m}$ parts, and use  each part as the base to construct a cylinder in $\Interval^n$. Essentially this means projecting all cylinders to a low-dimensional cube. For this end, we first show that when $n$ is big enough, there are $10$ dimensions such that any cylinder independent of at least one of these dimensions has very small probability. Then Lemma \ref{le:smalleventsdisjoint} ensures that the bases of these cylinders can be chosen as exclusive. Finally, the other cylinders are obtained by partitioning the part of $\mathcal{C}$ that has not yet been covered. Altogether, we get an exclusive set of cylinders whose measure vector is $\vec{p}$. 

%It is proved by construction. Basically, the proof has three steps. First, for any boundary vector $\vec{p}$, there are 10 dimensions such that cylinders independent of at least one of the dimensions must have measures small enough. Second, Lemma \ref{le:smalleventsdisjoint} ensures that these cylinders can be chosen as exclusive. Third, the other cylinders are obtained by properly dividing the part of the unit cube that has not yet been covered. Altogether, we get an exclusive set of cylinders whose measure vector is $\vec{p}$. 

\begin{theorem}\label{nn-2isgapless}
For any constant $m$, when $n$ is large enough, $\BipartiteGraph_{n,n-m}$ is gapless.
\end{theorem}
\LongVersion
\begin{proof}
We just consider $m=2$, since the method can be easily generalized to other $m$.

Apply Lemma \ref{le:smalleventsdisjoint} to $\BipartiteGraph^{\geq}_{10,8}$,  and we get an $\epsilon>0$. Let $K=\frac{2}{\epsilon}, n=10K, N={n \choose m}$. Arbitrarily fix a vector $\vec{p}\in(0,1)^N$ with $\sum_{i\in[N]}p_i=1$. Let $f$ be an arbitrary bijective function which maps unordered pairs on $[n]$ to $N$.

%Let $\BipartiteGraph_{n,n-2}=(\mathcal{X}, \mathbb{X})$.
Arbitrarily partition the set $[n]$ into $K$ disjoint groups with each containing $10$ elements.

Arbitrarily fix a group $T$. For any 9-subset $S$ of $T$ with $\{i\}=T\setminus S$, define $q_S=\sum_{j\notin T}p_{f(i,j)}$. For any 8-subset $S$ of $T$ with $\{i,j\}=T\setminus S$, define $q_S=p_{f(i,j)}$. The vector consisting of all these $q_S$ is denoted by $\vec{q}^T$. The $l_1$ norm of $\vec{q}^T$ is denoted by $v_T$.

We claim that there is a $T$ such that all entries of $\vec{q}^T$ are at most $\epsilon$. If it is not the case, $v_T>\epsilon$ for all $T$, so $\sum_T v_T\geq K\epsilon>2$. However, $\sum_T v_T\leq 2\sum_{1\leq i\leq N}p_i=2$. Hence, the claim is true.

Choose such a $T$. By the choice of $\epsilon$, there is an exclusive cylinder set $\EventSetB$ in the unit cube $\Interval^T$ that conforms with $\BipartiteGraph^{\geq}_{10,8}$ and satisfies $\mu(\EventSetB)=q^T$. For any 8- or 9-subset $S\subset T$, let $\EventB_S$ denote the cylinder in $\EventSetB$ that corresponds to $S$.

For each 8-subset $S=T\setminus \{i,j\}$, rename $B_S$ as $B_{f(i,j)}$.

For each 9-subset $S\subset T$ with $\{i\}=T\setminus S$, divide the cylinder $B_S$ into $10(K-1)$ disjoint cylinders $B_{f(i,j)}$, for each $j\notin T$. These $B_{f(i,j)}$ can be chosen so that they only depend on those $X_k$ with $k\in S$% and $\sum_{j:X_j\not T} \Pr(B_{ij}=\Pr(B_S)$
.

Arbitrarily partition $\Interval^T\setminus (\cup_{B\in\EventSetB}B)$ into ${n-10 \choose 2}$ disjoint sets, denoted by $B_{f(i,j)}$ where $i\neq j$ and $i,j\notin T$.  These $B_{f(i,j)}$ can be chosen such that $\mu(B_{f(i,j)})=p_{f(i,j)}$.

For each of the above $B_{f(i,j)}$, define a cylinder $A_{f(i,j)}=B_{f(i,j)}\times \Interval^{[n]\setminus T}$. Let $\EventSet=\{A_1,...,A_N\}$. It is straightforward to check that $\EventSet$ is exclusive with respect to $\BipartiteGraph_{n,n-m}$, $\mu(\EventSet)=\vec{p}$, and $\Pr(\cup_{A\in \EventSet}A)=1$.

As a result, $\BipartiteGraph_{n,n-2}$ is gapless.
\end{proof}
\LongVersionEnd

In spite of so much confirmative evidence, the general answer to the question $Q$ turns out to be NO! The following bigraph is an example where gap is not caused by containing cyclic bigraphs. Specifically, it is the bigraph $\BipartiteGraph^*=([5],[5],E)$ with $E= (\{1\}\times\{1,4,5\})\cup (\{2\}\times\{2,4,5\})\cup (\{3\}\times\{3,4,5\})\cup (\{4\}\times\{1,2,3,4\})\cup (\{5\}\times\{1,2,3,5\})$.

%where $X=\{x,y,u,v,w\}$ and $E$ is such that $\Neighbor(1)=\{x,u,v,w\}$, $\Neighbor(2)=\{y,u,v,w\}$, $\Neighbor(3)=\{x,y,u\}$, $\Neighbor(4)=\{x,y,v\}$, $\Neighbor(5)=\{x,y,w\}$.

%Let $\mathcal{X}=\{x,y,u,v,w\}$ and $\mathbb{X}$ be a list of length 5 with $\mathbb{X}_1=\{x,u,v,w\},\mathbb{X}_2=\{y,u,v,w\},\mathbb{X}_3=\{x,y,u\},\mathbb{X}_4=\{x,y,v\},\mathbb{X}_5=\{x,y,w\}$.

\begin{theorem}\label{gapwithoutcycle}
$\BipartiteGraph^*$ is gapful.
\end{theorem}
\LongVersion
\begin{proof}
The base graph $G_{\BipartiteGraph^*}$ is complete, so $\partial_a(\BipartiteGraph^*)=\{\vec{p}\in (0,1)^5: p_1+...+p_5=1\}$. Arbitrarily fix $\vec{p}\in\partial_a(\BipartiteGraph^*)$ with $p_4=p_5=\rho$ where $\rho$ is a constant.

Suppose $\EventSet=\{A_1,...,A_5\}$ is a set of cylinders in $\Interval^5$ which is exclusive with respect to $\BipartiteGraph^*$ and satisfies $\mu(\EventSet)=\vec{p}$. 
Let the coordinate variables of $\Interval^5$ be $X_1, X_2,...,X_5$. Since $\EventSet$ is exclusive and $\vec{p}\in\partial_a(\BipartiteGraph^*)$, we know that $\Pr(\cup_{A\in\EventSet}A)=1$ due to Lemma \ref{WorstPlacement}. By Theorem \ref{thm:boundaryfills}, further suppose that $\EventSet$ is $d$-discrete in every dimension, where $d$ is a positive integer. Namely, the unit interval $\Interval^{\{l\}}$, for any $l\in R(\BipartiteGraph^*)$, is partitioned into $d$ disjoint subintervals denoted by $\Delta^{\{l\}}_i,i\in [d]$. For any pair of integers $i,j\in [d]$ and a set $A\subseteq \Interval^{5}$, let $\pi^A_{i,j}$  denote the set $A\cap \left(\Delta^{\{4\}}_i\times \Delta^{\{5\}}_j \times \Interval^{3}\right)$; When $A$ lies in the $\sigma$-algebra of $\EventSet$, there must be a set in $\Interval^3$, denoted by $\tau^A_{i,j}$, such that $\pi^A_{i,j}=\left(\Delta^{\{4\}}_i\times \Delta^{\{5\}}_j \times \tau^A_{i,j}\right)$. A set $B\subseteq \Interval^{3}$ is said to have $e$-type if $\mu(B)=0$, $f$-type if $\mu(B)=1$, or $i$-type if $B=B^{\{i\}}\times \Interval^{[3]\setminus \{i\}}$ for some $B^{\{i\}}\subset \Interval^{\{i\}}$ with $0<\mu(B^{\{i\}})<1$, for $i\in [3]$. Let $T$ be the set of the five types.
For notational simplicity, let $A_{4,5}$ stand for $A_4\cup A_5$.

%For any $1\leq i,j\leq d$ and event $A\subseteq \Delta^x_i\times \Delta^y_j \times \Interval^{X\setminus \{x,y\}}$, $A$ is called 0-type if $\Pr(A)=0$, 1-type if $A= \Delta^x_i\times \Delta^y_j \times \Interval^{X\setminus \{x,y\}}$, $u$-type if $A=\Delta^x_i\times \Delta^y_j \times B^u\times \Interval^{X\setminus \{x,y,u\}}$ for some $B^u\subsetneq \Interval^{\{u\}}$, $v$-type or $w$-type defined likewise.

For any $i,j\in [d]$, we observe the following facts. % three facts. 

\begin{description}
\item[ Fact 1:]  For any $k\in [3]$, $\tau^{A_k}_{i,j}$ have either $e$-type, $f$-type, or $k$-type.
\item[ Fact 2:] There is at most one $k\in [3]$ such that $\tau^{A_k}_{i,j}$ does not have $e$-type. This follows from the exclusiveness of $\EventSet$ and the property that for any $k \neq k'\in[3]$, $\mu(\tau^{A_{k}}_{i,j}\cap \tau^{A_{k'}}_{i,j})\neq 0$ if neither $\tau^{A_{k}}_{i,j}$ nor $\tau^{A_{k'}}_{i,j}$ have $e$-type.
%The reason is as follows. For any $k_1,k_2\in[3]$, the intersection of any $k_1$-type and $k_2$-type sets must have nonzero measure. This, together with the assumption that $\EventSet$ is exclusive, implies that there is at most one $k\in [3]$ such that $\tau^{A_k}_{i,j}\$ is not of $e$-type. 
\item[ Fact 3:]  $\tau^{A_{4,5}}_{i,j}$ must have one of the five types in $T$. It follows from Fact 2, the exclusiveness of $\EventSet$, and the property that $\tau^{A_{4,5}}_{i,j}\cup_{k\in[3]}\tau^{A_{k}}_{i,j}=\Interval^3$.
\item[ Fact 4:]  Given $k\in [3]$, if $\tau^{A_{4,5}}_{i,j}$ has $k$-type, so does $\tau^{A_k}_{i,j}$.
\end{description} 
%First, $\pi^{A_k}_{i,j}$ is either 0-type, 1-type, or $s$-type, where $s=u/v/w$ for $k=3/4/5$ respectively. Second, because $u/v/w$-type events pairwise have non-empty intersection and $\mathcal{S}$ is exclusive, there is at most one $k$ out of $\{3,4,5\}$ such that $\pi^{A_k}_{i,j}\neq\emptyset$. As a result, $\pi^{A_1\cup A_2}_{i,j}$ can only be one of the five types: 0-,1-, $u$-,$v$-,$w$-type.

We now focus on $\tau^{A_{4,5}}_{i,j}$ and proceed case by case.
%We will focus on $\pi^{A_1\cup A_2}_{i,j}$ and proceed case by case. 

\begin{description}
\item[ Case 1:] There is an $i_0\in [d]$ such that $\tau^{A_{4,5}}_{i_0,j}$ has $e$-type for any $j\in [d]$. Because $\tau^{A_{4,5}}_{i_0,j}=\tau^{A_4}_{i_0,j}\cup \tau^{A_5}_{i_0,j}$,  $\mu(\tau^{A_5}_{i_0,j})=0$ for any $j\in [d]$. Recalling that $A_5$ is independent of $X_4$, $\mu(\tau^{A_5}_{i,j})=0$ for any $i,j\in [d]$. Hence, $\mu(A_5)=0$, contradictary to the choice of $\vec{p}$. Symmetrically, we also reach a contradiction if there is $j_0\in [d]$ such that $\tau^{A_{4,5}}_{i,j_0}$ has $e$-type for any $i\in [d]$.
\item[ Case 2:] There exist $i_0,i_1, j_0,j_1\in [d]$ such that $\tau^{A_{4,5}}_{i_0,j_0}$ has $e$-type while both $\tau^{A_{4,5}}_{i_0,j_1}$ and $\tau^{A_{4,5}}_{i_1,j_0}$ have other types. Without loss of generality, we assume that $i_0=1$, $j_0=1$, $\tau^{A_{4,5}}_{1,j}$ has $e$-type if and only if $1\leq j< j_1$, and $\tau^{A_{4,5}}_{i,1}$ has $e$-type if and only if $1\leq i< i_1$.% If $\tau^{A_{1,2}}_{i,1}$ has $e$-type for any $i\in [d]$, we reach a contradiction as in Case 1. Hence, assume there is $i_1$ such that $\Pr(\tau^{A_{1,2}}_{i,1})=0$ for all $1\leq i< i_1$ and $\Pr(\tau^{A_{1,2}}_{i,1})\neq 0$  for all $i_1\leq i\leq d$.

Since $A_4$ is independent of $X_5$ and $A_5$  is independent of $X_4$,  $\tau^{A_4}_{i,j}=\tau^{A_4}_{i,j'}$ and $\tau^{A_5}_{i,j}=\tau^{A_5}_{i',j}$ for any $i,i',j,j'\in [d]$. Hence, for any $i,j\in [d]$, we have $\tau^{A_{4,5}}_{i,j}=\tau^{A_4}_{i,j}\cup \tau^{A_5}_{i,j}=\tau^{A_4}_{i,1}\cup \tau^{A_5}_{1,j}$, and $\mu(\tau^{A_4}_{i,1}\cap \tau^{A_5}_{1,j})=0$ since $A_4$ and  $A_5$ are disjoint. In addition, for any $i\in [d]$, $\tau^{A_{4,5}}_{i,1}=\tau^{A_4}_{i,1}\cup \tau^{A_5}_{i,1}=\tau^{A_4}_{i,1}\cup \tau^{A_5}_{1,1}$, so $\tau^{A_{4,5}}_{i,1}$ and $\tau^{A_4}_{i,1}$ have the same type and $\mu(\tau^{A_{4,5}}_{i,1})=\mu(\tau^{A_4}_{i,1})$. Symmetrically, for any $j\in [d]$, $\tau^{A_{4,5}}_{1,j}$ and $\tau^{A_5}_{1,j}$ have the same type and $\mu(\tau^{A_{4,5}}_{1,j})=\mu(\tau^{A_5}_{1,j})$.

Now consider any $i\geq i_1$ and $j\geq j_1$. Since $\mu(\tau^{A_4}_{i,1})+\mu(\tau^{A_5}_{1,j})=\mu(\tau^{A_{4,5}}_{i,j})\leq 1$, the assumption $\mu(\tau^{A_{4,5}}_{i,1})>0$ and $\mu(\tau^{A_{4,5}}_{1,j})>0$ implies that $\tau^{A_{4,5}}_{i,1}$ and $\tau^{A_{4,5}}_{1,j}$ are neither $e$-type nor $f$-type.
%Due to the assumption of exclusion, for any $1\leq i,j\leq d$, we actually have $\tau^{A_4}_{i,1}=\tau^{A_{4,5}}_{i,1}$, $\tau^{A_5}_{1,j}=\tau^{A_{4,5}}_{1,j}$, $\tau^{A_{4,5}}_{i,j}=\tau^{A_4}_{i,j}\cup \tau^{A_5}_{i,j}=\tau^{A_4}_{i,1}\cup \tau^{A_5}_{1,j}$.% Note that all the equalities are up to a difference of zero measure.
%Also in the sense of type, The facts $\tau^{A_{4,5}}_{i_1,j_1}=\tau^{A_4}_{i_1,1}+\tau^{A_5}_{1,j_1},\Pr(\tau^{A_4}_{i_1,1})>0$ and $\Pr(\tau^{A_5}_{1,j_1})>0$ imply that $\tau^{A_4}_{i_1,1}$ and $\tau^{A_5}_{1,j_1}$ are neither $0$-type nor $1$-type. 
Assume that $\tau^{A_{4,5}}_{i,1}$ is $1$-type and $\tau^{A_{4,5}}_{1,j}$ is $2$-type. Then $\tau^{A_4}_{i,1}$ is $1$-type and $\tau^{A_5}_{1,j}$ is $2$-type, contradictory to the property that $\mu(\tau^{A_4}_{i,1}\cap \tau^{A_5}_{1,j})=0$. 

As a result, without loss of generality, assume that $\tau^{A_4}_{i,1}$ and $\tau^{A_5}_{1,j}$ have $1$-type for any $i\geq i_1$ and $j\geq j_1$. Since $\tau^{A_{4,5}}_{i,j}=\tau^{A_4}_{i,1}\cup \tau^{A_5}_{1,j}$, $\tau^{A_{4,5}}_{i,j}$ have either $1$-type or $f$-type if $i\geq i_1$ or $j\geq j_1$. 
By Fact 2 and Fact 4, both $\tau^{A_2}_{i,j}$ and $\tau^{A_3}_{i,j}$ have $e$-type when $i\geq i_1$ or $j\geq j_1$. Therefore, 
$\mu(A_2)+\mu(A_3)\leq (1-\mu(\Delta^{\{4\}}))(1-\mu(\Delta^{\{5\}}))$, where %$\mu$ the Lebesgue measure on the unit interval, 
$\Delta^{\{4\}}=\cup_{i_1\leq i\leq d}\Delta^{\{4\}}_i$ and $\Delta^{\{5\}}=\cup_{j_1\leq j\leq d}\Delta^{\{5\}}_j$.  We first prove Claim 1:

\textbf{Claim 1:} $(1-\mu(\Delta^{\{4\}}))(1-\mu(\Delta^{\{5\}}))\leq (1-2\rho)^2$.

Proof of the claim: Let $r_i=\mu(A_4\cap \Delta^{\{4\}}_i\times \Interval^{[5]\setminus\{4\}}), r_{i,j}=\mu(\pi^{A_4}_{i,j}), s_j=\mu(A_5\cap \Delta^{\{5\}}_j\times \Interval^{4}), s_{i,j}=\mu(\pi^{A_5}_{i,j})$ for $i,j\in [d]$. We have $\rho=\sum_{i\geq i_1}r_i=\sum_{i\geq i_1,j\in [d]}r_{i,j}=\sum_{j\geq j_1}s_j=\sum_{i\in [d],j\geq j_1}s_{i,j}$.

Because $A_4$ is independent of $X_5$, it holds that $r_{i,j}=r_i \mu(\Delta^{\{5\}}_j)$, so $\sum_{j_1\leq j\leq d}r_{i,j}=r_i \mu(\Delta^{\{5\}})$ and $\sum_{i_1\leq i\leq d,j_1\leq j\leq d}r_{i,j}=\rho\mu(\Delta^{\{5\}})$. Likewise, %since $A_5$ is independent of $x$, 
we have $\sum_{i_1\leq i\leq d,j_1\leq j\leq d}s_{i,j}=\rho\mu(\Delta^{\{4\}})$.

On the other hand, since $A_4$ and $A_5$ are disjoint, $r_{i,j}+s_{i,j}\leq \mu(\Delta^{\{4\}}_i)\mu(\Delta^{\{5\}}_j)$ for any $i_1\leq i\leq d, j_1\leq j\leq d$, which implies that $\sum_{i_1\leq i\leq d,j_1\leq j\leq d}(r_{i,j}+s_{i,j})\leq\mu(\Delta^{\{4\}})\mu(\Delta^{\{5\}})$.

Hence, $\mu(\Delta^{\{4\}})\mu(\Delta^{\{5\}})\geq \rho(\mu(\Delta^{\{4\}})+\mu(\Delta^{\{5\}}))\geq 2\rho \sqrt{\mu(\Delta^{\{4\}})\mu(\Delta^{\{5\}})}$, which in turn means that $\sqrt{\mu(\Delta^{\{4\}})\mu(\Delta^{\{5\}})}\geq 2\rho$. We further have
\begin{eqnarray*}
   (1-\mu(\Delta^{\{4\}}))(1-\mu(\Delta^{\{5\}}))&= &   1-(\mu(\Delta^{\{4\}})+\mu(\Delta^{\{5\}}))+\mu(\Delta^{\{4\}})\mu(\Delta^{\{5\}})\\
   &\leq &1-2\sqrt{\mu(\Delta^{\{4\}})\mu(\Delta^{\{5\}})}+\mu(\Delta^{\{4\}})\mu(\Delta^{\{5\}})\\
   &= &(1-\sqrt{\mu(\Delta^{\{4\}})\mu(\Delta^{\{5\}})})^2\leq (1-2\rho)^2.
\end{eqnarray*}
\textbf{Claim 1 is proven.}

By Claim 1, one has $\mu(A_2)+\mu(A_3)\leq (1-2\rho)^2$.  Since $\EventSet$ is exclusive, it holds that $\mu(A_1)=1-\sum_{2\leq k\leq 5} \mu(A_k) \geq 1-(1-2\rho)^2-2\rho=2\rho-4\rho^2$.
\item[ Case 3:] There is $k\in[3]$ such that $\tau^{A_{4,5}}_{i,j}$ is neither $e$-type nor $k$-type for any $i,j\in[d]$. Suppose $k=1$ satisfies the condition. Then, for any $i,j\in[d]$, among the candidate $e$-type, $f$-type, or $1$-type of $\tau^{A_{1}}_{i,j}$, the only possibility is $e$-type. Hence $\mu(A_1)=0$, which is a contradiction.

\item[ Case 4:] $\tau^{A_{4,5}}_{i,j}$ is not $e$-type for any $i,j\in [d]$, and for each $k\in [3]$, there are $i,j\in [d]$ such that $\tau^{A_{4,5}}_{i,j}$ has $k$-type. 

\textbf{Claim 2:} There are $i_0,j_0,j_1\in[d]$ and $k_1\neq k_2\in [3]$ such that $\tau^{A_{4,5}}_{i_0,j_0}$ has $k_1$-type and $\tau^{A_{4,5}}_{i_0,j_1}$  has $k_2$-type, or there are $i_0,i_1,j_0\in[d]$ and $k_1\neq k_2\in [3]$ such that $\tau^{A_{4,5}}_{i_0,j_0}$ has $k_1$-type and $\tau^{A_{4,5}}_{i_1,j_0}$  has $k_2$-type.

Proof of the claim: Suppose for contradiction that Claim 2 does not hold in Case 4. There must be $i_0,i_1,j_0,j_1\in [d]$ and $k_1\neq k_2\in [3]$ such that $\tau^{A_{4,5}}_{i_0,j_0}$ has $k_1$-type, $\tau^{A_{4,5}}_{i_1,j_1}$  has $k_2$-type, and both $\tau^{A_{4,5}}_{i_1,j_0}$  and $\tau^{A_{4,5}}_{i_0,j_1}$ have $k_3$-type or $f$-type. Without loss of generality, assume that both $\tau^{A_{4,5}}_{i_1,j_0}$  and $\tau^{A_{4,5}}_{i_0,j_1}$ have $f$-type.  and both $\tau^{A_{4,5}}_{i_1,j_0}$  and $\tau^{A_{4,5}}_{i_0,j_1}$ has $f$-type. Note that for any $i,i',j,j'\in[d]$,  
\begin{eqnarray}\label{eq:diagonal}
\begin{array}{rcl}
   \tau^{A_{4,5}}_{i,j}\cup\tau^{A_{4,5}}_{i',j'}&= & \tau^{A_4}_{i,j}\cup \tau^{A_5}_{i,j}\cup \tau^{A_4}_{i',j'}\cup \tau^{A_5}_{i',j'}=\tau^{A_4}_{i,j'}\cup \tau^{A_5}_{i',j}\cup \tau^{A_4}_{i',j}\cup \tau^{A_5}_{i,j'}\\
   &= &\tau^{A_5}_{i',j}\cup \tau^{A_4}_{i',j}\cup\tau^{A_4}_{i,j'}\cup \tau^{A_5}_{i,j'}=\tau^{A_{4,5}}_{i',j}\cup\tau^{A_{4,5}}_{i,j'}.
\end{array}
\end{eqnarray}
Hence $\tau^{A_{4,5}}_{i_0,j_0}\cup\tau^{A_{4,5}}_{i_1,j_1}=\tau^{A_{4,5}}_{i_1,j_0}\cup \tau^{A_{4,5}}_{i_0,j_1}$. Namely, an $f$-type set equals the union of a $k_1$-type set and  a $k_2$-type set, which is impossible. The cases where $\tau^{A_{4,5}}_{i_1,j_0}$ and $\tau^{A_{4,5}}_{i_0,j_1}$ have other types can be proved similarly. 

\textbf{Claim 2 is proven.}

By Claim 2, without loss of generality, assume that $\tau^{A_{4,5}}_{1,1}$ has $1$-type and $\tau^{A_{4,5}}_{1,2}$  has $2$-type. 

\textbf{Claim 3:} For any $i,j\in [d]$, $\tau^{A_{4,5}}_{i,j}$ and $\tau^{A_{4,5}}_{1,j}$ have the same type in $T$. 
%It is easy to check that either there is an  $i$ such that $\tau^{A_{4,5}}_{i,j},j\geq 1$ have two distinct types in addition to $1$-type, or  there is a $j$ such that $\tau^{A_{4,5}}_{i,j},i\geq 1$ have two distinct types in addition to $1$-type. Hence, without loss of generality, assume that $\tau^{A_{4,5}}_{1,1}$ is $u$-type and $\tau^{A_{4,5}}_{1,2}$  is $v$-type. 

Proof of the claim: We first show that for any $i,j\in [d]$, if $\tau^{A_{4,5}}_{1,j}$ does not have $3$-type, $\tau^{A_{4,5}}_{i,j}$ can't have $3$-type. Suppose for contradiction that there are $i,j\in [d]$ such that $\tau^{A_{4,5}}_{i,j}$ has $3$-type while $\tau^{A_{4,5}}_{1,j}$ does not. If $\tau^{A_{4,5}}_{1,j}$ has $1$-type, by formula (\ref{eq:diagonal}), $\tau^{A_{4,5}}_{1,j}\subset \tau^{A_{4,5}}_{1,j}\cup\tau^{A_{4,5}}_{i,2}=\tau^{A_{4,5}}_{1,2}\cup\tau^{A_{4,5}}_{i,j}$, meaning that a $1$-type set is inside the union of a $2$-type set and $3$-type set, which is impossible. Likewise, we also reach a contradiction if $\tau^{A_{4,5}}_{1,j}$ has $2$-type or $f$-type. As a result, under the condition of Case 4, there must be $j\in [d]$ such that $\tau^{A_{4,5}}_{1,j}$ has $3$-type. Without loss of generality, assume that $\tau^{A_{4,5}}_{1,3}$ has $3$-type.

Now for contradiction, suppose that there is $i,j\in [d]$ such that $\tau^{A_{4,5}}_{i,j}$ and $\tau^{A_{4,5}}_{1,j}$ have different types in $T$. If $\tau^{A_{4,5}}_{1,j}$ has $1$-type and $\tau^{A_{4,5}}_{i,j}$ has $2$-type,  by $\tau^{A_{4,5}}_{1,j}\subset \tau^{A_{4,5}}_{1,j}\cup\tau^{A_{4,5}}_{i,3}=\tau^{A_{4,5}}_{1,3}\cup\tau^{A_{4,5}}_{i,j}$, we again reach a contradiction. Likewise, there is a contradiction whenever $\tau^{A_{4,5}}_{i,j}$ and $\tau^{A_{4,5}}_{1,j}$ have different types.

\textbf{Claim 3 is proven.}
%We claim that for any $i\geq 2$, $\tau^{A_{4,5}}_{i,1}$ is not $w$-type. For contradiction, suppose that $\tau^{A_{4,5}}_{2,1}$ is $w$-type. Note that $\tau^{A_{4,5}}_{1,1}\cup\tau^{A_{4,5}}_{2,2}=\tau^{A_4}_{1,1}\cup \tau^{A_5}_{1,1}\cup \tau^{A_4}_{2,2}\cup \tau^{A_5}_{2,2}=\tau^{A_4}_{1,2}\cup \tau^{A_5}_{2,1}\cup \tau^{A_4}_{2,1}\cup \tau^{A_5}_{1,2}=\tau^{A_5}_{2,1}\cup \tau^{A_4}_{2,1}\cup\tau^{A_4}_{1,2}\cup \tau^{A_5}_{1,2}=\tau^{A_{4,5}}_{2,1}\cup\tau^{A_{4,5}}_{1,2}$. Since $\tau^{A_{4,5}}_{1,1}$ is $u$-type, $\tau^{A_{4,5}}_{1,2}$  is $v$-type, and $\tau^{A_{4,5}}_{2,1}$ is $w$-type, it is impossible that $\tau^{A_{4,5}}_{1,1}\cup\tau^{A_{4,5}}_{2,2}=\tau^{A_{4,5}}_{2,1}\cup\tau^{A_{4,5}}_{1,2}$ no matter what type $\tau^{A_{4,5}}_{2,2}$ is. We reach a contradiction. So, for any $i\geq 2$, $\tau^{A_{4,5}}_{i,1}$ is not $w$-type. Likewise, for any $i\geq 2$, $\tau^{A_{4,5}}_{i,2}$ is not $w$-type.

We further show that for any $i,j\in [d]$, $\tau^{A_{4,5}}_{i,j}=\tau^{A_{4,5}}_{1,j}$. To see this, again use formula (\ref{eq:diagonal}). Take $j=1$ as an example. For any $i\neq 1$,  we have $\tau^{A_{4,5}}_{1,1}\cup\tau^{A_{4,5}}_{i,2}=\tau^{A_{4,5}}_{1,2}\cup\tau^{A_{4,5}}_{i,1}$. Since $\tau^{A_{4,5}}_{1,1}$ and $\tau^{A_{4,5}}_{i,1}$ have $1$-type while $\tau^{A_{4,5}}_{i,2}$ and $\tau^{A_{4,5}}_{1,2}$ have $2$-type, the equality holds only if both $\tau^{A_{4,5}}_{i,1}=\tau^{A_{4,5}}_{1,1}$ and $\tau^{A_{4,5}}_{i,2}=\tau^{A_{4,5}}_{1,2}$.

As a result, $A_4\cup A_5$ is independent of $X_4$. This, together with the fact that $A_5$ is independent of $X_4$, implies that $A_4$ is independent of $X_4$. 

Furthermore, for each $j\in [d]$ such that $\tau^{A_{4,5}}_{1,j}$ has $1$-type, $\tau^{A_{4,5}}_{i,j}$ also has $1$-type for any $i\in [d]$, so $\tau^{A_3}_{i,j}=\Interval^{3}\setminus \tau^{A_{4,5}}_{i,j}$. Since $\tau^{A_{4,5}}_{i,j}$ is independent of $i$, so is $\tau^{A_3}_{i,j}$. Consequently, $A_3$ is independent of $X_4$. Likewise, both $A_1$ and $A_2$ are also independent of $X_4$.

Altogether, in Case 4, all cylinders are independent of $X_4$.

\end{description}

The case study above indicates that only Case 2 and Case 4 are possible. 

Now consider the probability vector $\vec{p}=(\frac{2}{9}-\frac{2}{3}\epsilon,\frac{2}{9}-\frac{2}{3}\epsilon,\frac{2}{9}-\frac{2}{3}\epsilon,\frac{1}{6}+\epsilon,\frac{1}{6}+\epsilon)\in\partial_a(\BipartiteGraph^*)$, where $\epsilon>0$ is constant to be determined. Arbitrarily choose an exclusive cylinder set $\EventSet=\{A_1,...,A_5\}\sim \BipartiteGraph^*$ such that $\mu(\EventSet)=\vec{p}$. Because $p_1<2p_4-4p^2_4$ when $\epsilon$ is small enough, only Case 4 is possible for $\vec{p}$. Assume that these events are independent of $X_4$. We can choose $B_1\subset \Interval^{\{1,5\}},B_2\subset \Interval^{\{2,5\}},B_3\subset \Interval^{\{3,5\}},B_4\subset \Interval^{\{1,2,3\}},B_5\subset \Interval^{\{1,2,3,5\}}$ as bases of $A_1,...,A_5$, respectively. Choose the minimum sets $\Gamma_1,\Gamma_2,\Gamma_3\subseteq\Interval^{\{5\}}$ and $\Lambda_i\subseteq\Interval^{\{i\}}$ for $i\in [3]$ such that $\mu(B_i\setminus (\Gamma_i\times \Lambda_i))=0$ for $i\in [3]$.  Since the cylinder set $\EventSet$ is exclusive, $\mu(\Gamma_i\cap \Gamma_j)=0$ for any $i,j\in[3]$. Let $\gamma_i=\mu(\Gamma_i), \lambda_i=\mu(\Lambda_i)$ for $i\in\{1,2,3\}$. The inequalities must hold simultaneously:
\begin{itemize}
\item $\gamma_i \lambda_i\geq \frac{2}{9}-\frac{2}{3}\epsilon, i\in\{1,2,3\} $
\item $\gamma_1+\gamma_2+\gamma_3\leq 1$
\item $(1-\lambda_1)(1-\lambda_2)(1-\lambda_3)\geq \frac{1}{6}+\epsilon$
\end{itemize}
where the last inequality is because $B_4$ must lies inside of  $(\Interval^{\{1\}}\setminus \Lambda_1)\times (\Interval^{\{2\}}\setminus \Lambda_2)\times(\Interval^{\{3\}}\setminus \Lambda_3)$.
However, these inequalities can't hold simultaneously when $\epsilon$ is small enough.

As a result, there is no cylinder set $\EventSet$ which is exclusive with respect to  $\BipartiteGraph^*$ and $\mu(\EventSet)=\vec{p}$. Since $\vec{p}\in\partial_a(\BipartiteGraph^*)\subset \partial(\BipartiteGraph^*)\cup\Interior(\BipartiteGraph^*)$, $\BipartiteGraph^*$ is gapful due to Theorem \ref{Conj:GapGeom}.
\end{proof}
\LongVersionEnd

By Delete-Event and the inverse operation of Duplicate-Variable, it is not difficult to reduce $H_{7,5}$ to $H^*$. Because $H^*$ is gapful, $H_{7,5}$ is also gapful. From Corollary \ref{cor:Sparsifiedgraphgapful}, we have the following corollary.
\begin{corollary}\label{cor:7c5cgapful} For any integer $c \geq 1$, every sparsification of $H_{7c,5c}$ is gapful.
\end{corollary}

In summary, we get some instances of gapful/gapless bigraphs, listed in Table \ref{tab:consecutive}.

\begin{table}
\centering
\caption{Instances of gapful/gapless bigraphs}\label{tab:consecutive}
\begin{tabular}{ | p{4cm}<{\centering} | p{3.5cm}<{\centering} |}
  \hline
   Gapful & Gapless    \\
  \hline
   $H^*$ &  $H_{n,n-c}$ for large $n$   \\
  %\hline
   Sparsifications of $H_{7c,5c}$ & $H_{n,n-1}$ for $n\ge 4$    \\
  %\hline
  cyclic bigraphs & treelike bigraphs \\
  \hline
\end{tabular}
%\noindent In $\dag$ case, this upper bound of $X$ is tight for constant $\alpha$. See section~\ref{sec:main}
\end{table}

\LongVersion

\subsection{Characterizing a-gapful and strongly a-gapful graphs}
Another interesting perspective of gaps is dependency-graph-oriented: we say that a graph $\DependencyGraph$ is a-gapful if there is a gapful bigraph whose base graph is $\DependencyGraph$, otherwise it's called a-gapless. Kolipaka et al. \cite{kolipaka2011moser} considered another closely-related concept: a graph $\DependencyGraph$ is strongly a-gapful if any bigraph with $\DependencyGraph$ as base graph is gapful, otherwise it's called strongly a-gapless.

One can easily observe that a bigraph $\BipartiteGraph$ is gapless if $\DependencyGraph_\BipartiteGraph$ is a-gapless, while it is gapful if $\DependencyGraph_\BipartiteGraph$ is strongly a-gapful. With the above mentioned results, we can completely characterize a-gapless or strongly a-gapful bigraphs, solving the 6-year open problem proposed by Kolipaka et al. \cite{kolipaka2011moser}. %We will also study the weak gap problem, namely to determine whether a graph has a gap. The problem is weak in the sense that if bigraph $\BipartiteGraph$ has a gap, so does $\DependencyGraph_\BipartiteGraph$, but not vice versa. 
\treegraph*

\begin{proof}
It immediately follows from Theorems \ref{dependencygaplessdecision} and \ref{cyclesaregapful}.
\end{proof}

For strong a-gapfulness, we need the following definition, where $Cliq(\DependencyGraph)$ is the set of maximal cliques of the graph $\DependencyGraph$.

\begin{definition}
Given a graph $\DependencyGraph=([n],E)$ with $Cliq(\DependencyGraph)=\{C_1,...,C_m\}$, its canonical bigraph, denoted by $\BipartiteGraph_\DependencyGraph$, is the bigraph $([n],[m],E')$ where $E'=\{(i,j)\in [n]\times[m]: i\in C_j\}$. 
\end{definition}

Intuitively, $\BipartiteGraph_\DependencyGraph$ models the variable generated event system where each maximal clique has a distinct variable and an event depends on a variable if it is in the corresponding maximal clique.

We claim that among the bigraphs whose base graph is $\DependencyGraph$, $\BipartiteGraph_\DependencyGraph$ has the minimum interior. This means that $\DependencyGraph$ is strongly a-gapful if and only if $\BipartiteGraph_\DependencyGraph$ is gapful.

\begin{lemma}\label{canonicalbigraphworst}
Given a graph $\DependencyGraph$, for any bigraph $\BipartiteGraph$ with $\DependencyGraph_\BipartiteGraph=\DependencyGraph$, we have $\Interior(\BipartiteGraph)\supseteq \Interior(\BipartiteGraph_\DependencyGraph)$.
\end{lemma}
\begin{proof} We prove the lemma in two steps.

Step 1: For any bigraph $\BipartiteGraph=([n],[m],E)$ with $S$ being a clique in $\DependencyGraph_\BipartiteGraph$, define the bigraph $\BipartiteGraph'=([n],[m+1],E')$ such that for any $j\in R(\BipartiteGraph')$, $\Neighbor_{\BipartiteGraph'}(j)=S$ if $j= m+1$, otherwise $\Neighbor_{\BipartiteGraph'}(j)=\Neighbor_{\BipartiteGraph}(j)$. Arbitrarily fix $\vec{p}\in \mathcal{E}(\BipartiteGraph)$. There is a set $\EventSet$ of cylinders in $\Interval^m$ such that $\EventSet\sim \BipartiteGraph$, $\mu(\EventSet)=\vec{p}$, and $\mu(\cup_{\Event\in\EventSet}\Event)=1$. Let $\EventSet'=\{\Event\times \Interval^{\{m+1\}}: \Event\in \EventSet\}$. Then $\EventSet'$ is a set of cylinders in $\Interval^{m+1}$, $\EventSet'\sim \BipartiteGraph'$, $\mu(\EventSet')=\vec{p}$, and $\mu(\cup_{\Event\in\EventSet'}\Event)=1$. Hence, $\mathcal{E}(\BipartiteGraph')\subseteq \mathcal{E}(\BipartiteGraph)$.

As a result, given $\BipartiteGraph$ with $\DependencyGraph_\BipartiteGraph=\DependencyGraph$, for each of the maximal clique in $\DependencyGraph_\BipartiteGraph$, modify $\BipartiteGraph$ as in Step 1.  Let $\overline{\BipartiteGraph}$ be the resulting bigraph. We have $\Interior(\BipartiteGraph)\supseteq \Interior(\overline{\BipartiteGraph})$.

Step 2. $\overline{\BipartiteGraph}$ is the same as $\BipartiteGraph_\DependencyGraph$ except that there might be $j\neq k\in R(\overline{\BipartiteGraph})$ such that $\Neighbor_{\overline{\BipartiteGraph}}(j)\subseteq \Neighbor_{\overline{\BipartiteGraph}}(k)$. Apply the inverse operation of Variable-Duplicate to $\overline{\BipartiteGraph}$ as many times as possible, and the final bigraph is exactly $\BipartiteGraph_\DependencyGraph$. Recall that in proving that Variable-Duplicate preserves gapful (see Theorem \ref{reductioneffect}), we actually prove that Variable-Duplicate preserves boundary, hence also preserving interior. This meaning that $\Interior(\overline{\BipartiteGraph})=\Interior(\BipartiteGraph_\DependencyGraph)$.

Altogether, we know that $\Interior(\BipartiteGraph)\supseteq \Interior(\BipartiteGraph_\DependencyGraph)$.
\end{proof}

A graph is called chordal, if it has no induced cycle of length greater than three. A well known property of chordal graphs is that it has a vertex which lies in exactly one maximal clique. Now we have the following result for chordal graphs.

\begin{lemma}\label{chordalgapless}
Any chordal graph is strongly a-gapless.
\end{lemma}
\begin{proof} 
Let $G=([n],E)$ be a chordal graph. We prove by induction on $n$.

Basis: $n=1$. It is trivial that $G$ is strongly a-gapless.

Hypothesis: The lemma holds when $n< N$.

Induction: Now consider the case $n=N$. Let $\BipartiteGraph=\BipartiteGraph_\DependencyGraph=([n],[m],E)$. Without loss of generality, assume that the vertex $n$ of $\DependencyGraph$ lies in exactly one maximal clique $S=\{n-k+1,...,n\}$. That is, $n\in L(\BipartiteGraph)$ has only one neighbor, say $m$, in $\BipartiteGraph$, and $\Neighbor_\BipartiteGraph(m)=S$.
Let $\BipartiteGraph'$ be the bigraph obtained from $\BipartiteGraph$ by deleting the vertex $n\in L(\BipartiteGraph)$, and $\DependencyGraph'$ be the chordal graph obtained by deleting the vertex $n$ from $\DependencyGraph$. Obviously, if $S\setminus\{n\}$ remains a maximal clique in $\DependencyGraph'$, $\BipartiteGraph'=\BipartiteGraph_{\DependencyGraph'}$; otherwise, $\BipartiteGraph_{\DependencyGraph'}$ can be obtained by applying the inverse operation of Variable-Duplicate to $\BipartiteGraph'$. We always have that $\BipartiteGraph'$ is gapless if and only if so is $\BipartiteGraph_{\DependencyGraph'}$. 

 Arbitrarily fix $\vec{p}\in\partial(\BipartiteGraph)$. Let $\vec{p}'=(p_1,...,p_{n-k},\frac{p_{n-k+1}}{1-p_n},...,\frac{p_{n-1}}{1-p_n})$. Choose $\lambda>0$ such that $\lambda\vec{p}'\in \partial(\BipartiteGraph')$.
Applying the induction hypothesis to $\DependencyGraph'$, by Theorem \ref{Conj:GapGeom}, there is a set $\EventSet'$ of cylinders $\Event'_1,...,\Event'_{n-1}\subseteq \Interval^m$ such that $\EventSet'$ is exclusive with respect to $\BipartiteGraph'$, $\mu(\EventSet')=\lambda\vec{p}'$, and $\mu(\cup_{\Event'\in\EventSet'}\Event')=1$. Define $\Event_i=\Event'_i$ for $1\leq i\leq n-k$, $\Event_i=\{(x_1,...,x_{m-1},x_m(1-p_n)): (x_1,...,x_{m-1},x_m)\in\Event'_i\}$ for $n-k<i<n$, $\Event_n=\{(x_1,...,x_m)\in \Interval^m: x_m\geq 1-p_n\}$. Since $\mu(\cup_{\Event'\in\EventSet'}\Event')=1$ and $\Event'_1,...,\Event'_{n-k}$ are independent of $X_m$, we have $\mu((\Interval^{m-1}\times [0,1-p_n])\cap (\cup_{1\leq i\leq n-1}\Event_i))=1-p_n$. As a result, $\mu(\cup_{\Event\in\EventSet}\Event)=1$, $\mu(\EventSet)=\vec{q}\triangleq(\lambda p_1,...,\lambda p_{n-1},p_n)$, and $\EventSet$ is exclusive with respect to $\BipartiteGraph$, where $\EventSet=\{\Event_1,...,\Event_n\}$. By 
Corollary \ref{exclusivenonexterior}, we know that $\vec{q}\in \partial(\BipartiteGraph)$, which in turn means that $\vec{q}=\vec{p}$ by Lemma \ref{le:boundaryisunique}. 

Altogether, for any $\vec{p}\in\partial(\BipartiteGraph)$, there is a cylinder set $\EventSet$ which is exclusive with respect to $\BipartiteGraph$ and satisfies $\mu(\EventSet)=\vec{p}$. By Theorem \ref{Conj:GapGeom}, $\BipartiteGraph$ is gapless, implying that $\DependencyGraph$ is strongly a-gapless.
\end{proof}

We are ready to present an exact characterization of strongly a-gapful graphs.

\chordal*
\begin{proof}
Arbitrarily fix a graph $\DependencyGraph$. 

If it is not chordal, there must be an induced cycle of length at least four.  By Corollary \ref{cor:Cyclegapful}, $\BipartiteGraph_\DependencyGraph$ is gapful, so $\DependencyGraph$ is strongly a-gapful.

On the other hand, if $\DependencyGraph$ is chordal, it is strongly a-gapless by Lemma \ref{chordalgapless}.
\end{proof}

Theorem \ref{dependencygaplessdecision} is an immediately corollary of Theorem \ref{stronglyagapfuliffchordal}. Its constructive proof is retained since the construction leads to the explicit equation of boundary vectors. 

\LongVersionEnd

\section{Hardness Results}
We define some computational problems that are closely related to the variable-LLL problem and show that they are difficult to solve. 

\begin{definition}[\Mup ~ Problem]
Given a bigraph $\BipartiteGraph=([n], [m], E)$ and vector $\vec{p} \in (0,1]^n$,  %\Mup \ asks the maximum probability of the union of the events, 
compute $\Psi(\BipartiteGraph,\vec{p})\triangleq \max_{\EventSet \sim \BipartiteGraph, \mu(\EventSet)=\vec{p} } \mu\left(\cup_{\Event\in \EventSet} \Event \right)$, where $\EventSet$ ranges on sets of cylinders in $\mathbb{I}^m$ and $\mu$ is Lebesgue measure.
\end{definition}

\begin{definition}[\Int ~Problem]
Given a bigraph $\BipartiteGraph$ and a vector $\vec{p}$ on $(0,1)$, decide whether $\vec{p}\in\mathcal{I}(\BipartiteGraph)$.
\end{definition}

\begin{theorem}{\label{thm: hardness1}}
\Mup ~is \#P-hard.
\end{theorem}
\begin{proof}
It is enough to show that \Mup \ is \#P-hard even if $\BipartiteGraph$ is a $(3,2)$-regular %planar 
bigraph and $\vec{p}=(\frac{1}{8},\frac{1}{8},...,\frac{1}{8})$.

Arbitrarily fix a $(3,2)$-regular bigraph $\BipartiteGraph=([n], [m], E)$. We will construct a set  $\EventSet$ of cylinders in $\mathbb{I}^m$ such that $\mu(\EventSet)=\vec{p}$, $\EventSet$ is exclusive with respect to $\BipartiteGraph$, and the probability of the union is maximized.

Arbitrarily choose a function $f: [m]\rightarrow [n]$ which maps each vertex in $[m]$ to one of its neighbors. For each $i\in [n]$, a cylinder $\Event_i\subset \mathbb{I}^m$ is defined in this way: for each neighbor $j$ of $i$, $0\leq X_j \leq 1/2$ if $f(j)=i$, otherwise $1/2< X_j \leq 1$. Let $\EventSet=\{\Event_1,...,\Event_n\}$.

Since each $i\in [n]$ has exactly three neighbors in $\BipartiteGraph$ and each $j\in [m]$ has exactly two neighbors, we observe that $\mu(\EventSet)=\vec{p}$ and $\EventSet$ is exclusive with respect to $\BipartiteGraph$. Hence $\mu(\cup_{i\in [n]} \Event_i)=\Psi(\BipartiteGraph,\vec{p})$, by Lemma \ref{WorstPlacement}.

The construction actually partitions $\mathbb{I}^m$ into $2^m$ blocks each having measure $2^{-m}$. Any cylinder in $\EventSet$ consists of some blocks. Let $B_{k_1, k_2,...,k_m}$, $k_j\in\{0,1\}$ for any $j$, denote the block defined by $0\leq X_j\leq 1/2$ if $k_j=0$ or $1/2< X_j\leq 1$ if $k_j=1$, for any $j\in[m]$. Given $k_1, k_2,...,k_m\in\{0,1\}$ and $i\in [n]$, the following two conditions are equivalent.

1. $B_{k_1,k_2,...,k_m}\subseteq \Event_i$.

2.  For each neighbor $j$ of $i$ in $\BipartiteGraph$, $k_j=0$ if and only if $f(j)=i$.

Let $N$ be the number of blocks outside of $\cup_{i\in [n]} \Event_i$. Then we have $\mu(\cup_{i\in [n]} \Event_i)=1-N/2^m$, so computing $\Psi(\BipartiteGraph,\vec{p})$ is equivalent to computing $N$.

On the other hand, computing $N$ is related to the 3SAT problem. Let $\{y_1,...,y_m\}$ be a set of boolean variables. For each $i\in [n]$, assume $j_1,j_2,j_3$ are its neighbors in $\BipartiteGraph$; define a 3SAT clause $\phi_i=z_{j_1}\vee z_{j_2} \vee z_{j_3} $ where the literal $z_{j_k}\triangleq y_{j_k}$ if $f(j_k)=i$, otherwise $z_{j_k}\triangleq \overline{y_{j_k}}$ , for $k=1,2,3$. The constraint-variable graph of $\phi\triangleq \phi_1\land ...\land \phi_n $ is $\BipartiteGraph$. Note that each variable appears twice oppositely, so $\phi$ is a Holant$([0,1,0]|[0,1,1,1])$ or Rtw-Opp-\#3SAT instance.

Now consider an assignment $y_j=k_j,j\in [m]$.  It is straightforward to check that $\phi$ is satisfied if and only if the block $B_{k_1,k_2,...,k_m}$ is outside $\cup_{i\in [n]} \Event_i$. Thus, $N$ is the number of satisfying assignments of $\phi$, which is \#P-hard to compute even if $\BipartiteGraph$ is (3,2)-regular, by \cite[Theorem 8.1]{CaiLX08}. 
\end{proof}

\LongVersion
\begin{remark}
The proof is inspired by but substantially different from the proofs in Section C of \cite{HarveySV16}.% In fact, except that we quote a different hardness source problem, Theorem \ref{thm: hardness1} can be proved in a very analogical route to the proof in \cite{HarveySV16}. In that route, we need to show that a value of the alternating-sign independence polynomial is \#P-hard. That value is nothing but exactly the answer to the \#P-hard problem Holant$([1, 1, 0]|[1,0,0,-\frac{1}{8}])$ (\cite{KowalczykC16}).  (In fact, the complexity dichotomy theorem  in \cite{KowalczykC16} is very powerful involving many other \#P-hard problems.) We choose the more self-contained proof route. The two possible routes are somehow parallel under a base transformation. More specifically, there is a holographic reduction (\cite{Valiant08}) between Holant$([0,1,0]|[0,1,1,1])$ and Holant$([1, 1, 0]|[1,0,0,-\frac{1}{8}])$.
\end{remark}

\begin{remark}
Given $\BipartiteGraph$ and $\vec{p}$ as in Theorem \ref{thm: hardness1}, $\Psi(\BipartiteGraph,\vec{p})$ is a proper fraction whose denominator is $2^m$. This fact will be used in proving the next theorem.
\end{remark}
\LongVersionEnd

By Theorem \ref{thm: hardness1}, one can prove the following result.
\begin{theorem}{\label{thm: hardness2}}
\Int \ is \#P-hard.
\end{theorem}

\LongVersion
\begin{proof}
%Given an instance ($\BipartiteGraph=([n], [m], E), \vec{p}$) of \Mup. According to Theorem {\ref{thm: hardness1}}, we may assume that $\BipartiteGraph$ is a  $(3,2)$-regular %planar 
%bigraph, and $\vec{p}=(\frac{1}{8}, \frac{1}{8}, \ldots, \frac{1}{8})$. Suppose the answer is $1- \frac{N}{2^m}$.
Given a  $(3,2)$-regular bigraph $\BipartiteGraph$ and $\vec{p}=(\frac{1}{8}, \frac{1}{8}, \ldots, \frac{1}{8})$, suppose that $\Psi(\BipartiteGraph,\vec{p})=1- \frac{N}{2^m}$.

Let's construct $\BipartiteGraph'=([n+1], [m], E')$ and $ \vec{p}^{(r)}$, where $E'=E \cup \{(n+1,1),(n+1,2), \ldots, (n+1, m)\}$ and $\vec{p'}=(\frac{1}{8}, \frac{1}{8}, \ldots, \frac{1}{8}, r)$ with $r\in[0,1]$. 

It is not hard to see $\vec{p}^{(r)} \in \mathcal{I}(\BipartiteGraph')$ if and only if $1- \frac{N}{2^m}+r < 1$. Obviously, $r_{max}=\frac{N-1}{2^m}$, where $r_{max}$ is the maximum $r$, among all proper fractions whose denominator is $2^m$ , such that $\vec{p}^{(r)} \in \mathcal{I}(\BipartiteGraph')$. 

Using binary search and solving \Int  ~on \ poly$(m)$ instances of the form $(\BipartiteGraph',\vec{p}^{(r)})$, we can find out $r_{max}$ and in turn get $\Psi(\BipartiteGraph,\vec{p})$. By Theorem {\ref{thm: hardness1}}, \Int \ is \#P-hard.
\end{proof}
\LongVersionEnd

\section*{Acknowledgement}
We are grateful to the anonymous referees for detailed corrections and suggestions.
%感谢的顺序是不是也按字母排？
We also thank Prof.\ Zhiwei Xu, Prof.\ Yungang Bao,  Prof.\ Xiaoming Sun and Prof.\ Roger Wattenhofer for their encouragement and support.

%Part of this work was done independently, by Mingji Xia \& Liang Li and by the other authors, before their meeting.

Kun He's work is supported by the National Key Research and Development Program of China (Grant No. 2016YFB1000200) and the National Natural Science Foundation of China (Grant No. 61433014, 61502449, 61602440) and the China National Program for support of Top-notch Young Professionals. Xingwu Liu's work is partially supported by the National Natural Science Foundation of China (Grant No. 61420106013). Yuyi Wang's work is partially supported by the National Natural Science Foundation of China (Grant no. 11601375). Mingji Xia's work is supported by China National 973 Program (Grant No. 2014CB340300). 

\bibliographystyle{abbrv}
\bibliography{DiscretizationAndGap}

%\balancecolumns

%\balancecolumns

\end{document}